\pdfoutput=1
\documentclass[11pt]{article}

\usepackage[a4paper, margin=3cm]{geometry}
 
\usepackage{color}
\usepackage{amsfonts, amsmath,  amssymb, mathrsfs, amscd, amsthm}
\usepackage{graphicx}
\usepackage{enumitem}
\usepackage{hyperref}
\usepackage{mathtools} 
\usepackage{ulem} 
\usepackage{amsthm}  
\usepackage{bbm}
\usepackage{nicefrac}
\usepackage{MnSymbol}
\usepackage{makebox}
\usepackage{changepage}
\usepackage{tikz-cd}
\tikzcdset{ampersand replacement=\&}

\graphicspath{{./}{../}{./Inkscape/}{./Figs/}{./Figs2/}{./Figs3/}}  

\usepackage[
backend=biber,
style=alphabetic
]{biblatex} 
\newtheorem{theorem}{Theorem}[section]
\newtheorem{example}[theorem]{Example}
\newtheorem{proposition}[theorem]{Proposition}
\newtheorem{lemma}[theorem]{Lemma}
\newtheorem{corollary}[theorem]{Corollary}
\newtheorem{conjecture}[theorem]{Conjecture}
\newcommand{\exa}[1]{\begin{example} #1 \end{example}}
\newcommand{\prop}[1]{\begin{proposition} #1 \end{proposition}}
\newcommand{\lemm}[1]{\begin{lemma} #1 \end{lemma}}

\theoremstyle{definition}
\newtheorem{definition}[theorem]{Definition}
\newcommand{\defn}[1]{\begin{definition} #1 \end{definition}}

\theoremstyle{remark}
\newtheorem{remark}{Remark}[section]
\newcommand{\rem}[1]{\begin{remark} #1 \end{remark}}

\newcommand{\myaddress}%
{\parbox{3.3in}{\footnotesize \begin{center} 
			Department of Pure Mathematics, \\ University of Leeds, 
			Leeds LS2 9JT, UK.\end{center}}}

\newcommand{\ignore}[1]{}

\newcommand{\ali}[1]{\begin{align*} #1 \end{align*}}

\newcommand{\GG}{{\mathcal G}}
\newcommand{\HH}{{\mathcal H}}
\newcommand{\CC}{{\mathcal C}}
\newcommand{\DD}{{\mathcal D}}

\newcommand{\V}{\mathcal{V}}
\newcommand{\Topo}{{\mathbf{Top}}}  
\newcommand{\TOPO}{{\mathbf{TOP}}}
\newcommand{\Set}{{\mathbf{Set}}}   
\newcommand{\Cat}{{\mathbf{Cat}}}
\newcommand{\Grpd}{{\mathbf{Grpd}}}
\newcommand{\Vect}{{\mathbf{Vect}}}
\newcommand{\VectC}{\mathbf{Vect}_\mathbb{C}}
\newcommand{\Vectk}{{\mathbf{Vect}}_{\mathbbm{k}}}

\newcommand{\II}{{\mathbb I}}  

\newcommand{\Vectotimes}{\otimes_{\mathbbm{k}}}

\newcommand{\lan}{\langle}
\newcommand{\ran}{\rangle}

\newcommand{\Power}{{\mathcal P}}

\newcommand{\C}{\mathbb C}
\newcommand{\N}{\mathbb N}  
\newcommand{\Z}{\mathbb Z}      
\newcommand{\R}{{\mathbb R}}

\newcommand{\push}[2]{{ \sqcup_{#2} }}

\newcommand{\Grpdforget}{\mathrm{U}_{\GG}}
\newcommand{\Topoforget}{\mathrm{U_T}}

 \newcommand{\simp}{\stackrel{p}{\sim}}   
 \newcommand{\classp}[1]{[#1]_{\!\mbox{\tiny p}}}

\newcommand{\PsiOps}{\Gamma_{\!\frac{1}{2}}}

\newcommand{\id}{\mathrm{id}}

\newcommand{\rev}{\mathrm{rev}}

\newcommand{\graphh}[2]{\xymatrix{#1 \ar@{=>}[r]^s_t & #2}}
\newcommand{\graphhh}[2]{\xymatrix{#1 \ar@<-.5ex>[r]_t \ar@<.5ex>[r]^s & #2}} 
\newcommand{\mM}{{\mathsf M}}

\newcommand{\MAG}{{\mathbf{Mag}}}

\newcommand{\tTopo}{{\mathsf{Top}}}  
\newcommand{\sSet}{{\mathsf{Set}}}  
\newcommand{\vVectk}{{\mathsf{Vect}_\mathbbm{k}}}  
\newcommand{\vVectC}{{\mathsf{Vect}_\mathbbm{C}}}

\newcommand{\Path}{{\mathfrak P}}

\newcommand{\colim}
{{\mathrm{colim}}}
\newcommand{\copr}[2]{\langle #1,#2 \rangle} 
\newcommand{\HomCob}{{\mathrm{HomCob}}}
\newcommand{\cHomCob}{\mathsf{HomCob}}
\newcommand{\CofCsp}{\mathrm{CofCos}}
\newcommand{\cCofCsp}{\mathsf{CofCos}}
\newcommand{\bHomCob}{\mathsf{bHomCob}}
\newcommand{\ccc}{concrete cofibrant cospan}

\newcommand{\cc}{cofibrant cospan}

\newcommand\sbullet[1][.5]{\mathbin{\vcenter{\hbox{\scalebox{#1}{$\bullet$}}}}}

\newcommand{\homfin}{homotopically $1$-finitely generated}
\newcommand{\homcob}{homotopy cobordism}
\newcommand{\bhomcob}{based homotopy cobordism}
\newcommand{\chomcob}{concrete homotopy cobordism}
\newcommand{\cbhomcob}{concrete based homotopy cobordism}

\newcommand{\FinSetX}{\mathbf{FinSet^*}(X)}

\newcommand*\lon{%
	\nobreak
	\mskip6mu plus1mu
	\mathpunct{}%
	\nonscript
	\mkern-\thinmuskip
	{:}%
	\mskip2mu
	\relax
}

\newcommand{\chip}{{\pmb \chi}} 
\newcommand{\cmor}[5]{#2\colon #3 \to #1 \leftarrow #5\lon #4}
\newcommand{\cbmor}[5]{#2\colon (#3,#3_0) \to (#1,#1_0) \leftarrow (#5,#5_0)\lon #4} 
\newcommand{\scbmor}[8]{#3\colon (#4,#5) \to (#1,#2) \leftarrow (#7,#8)\lon #6} 
\newcommand{\cmortikz}[5]{\scalebox{1}{\begin{tikzcd}[ampersand replacement= \&, cramped , sep=tiny]
			#3 \ar[dr,"#2"' near start] \& \& #5 \ar[dl,"#4" near start] \\
			\& #1 \& \\
	\end{tikzcd}}
}
\newcommand{\cbmortikz}[5]{\scalebox{0.6}{\begin{tikzcd}[ampersand replacement= \&, cramped ,sep=tiny]
			(#3, #3_0) \ar[dr,"#2"' near start] \& \& \makebox*{MM}{$(#5,#5_0)$} \ar[dl,"#4" near start] \\
			\& (#1,#1_0) \& \\
		\end{tikzcd}
		\hspace*{5pt}}
}
\newcommand{\lprime}{\scalebox{0.4}{$\prime$}}

\newcommand{\che}{cospan homotopy equivalence}
\newcommand{\simche}{\stackrel{ch}{\sim}}
\newcommand{\classche}[1]{\left[#1\right]_{\!\mbox{\tiny ch}}}

\newcommand{\Cobcat}[1]{\mathbf{Cob}_{#1}}
\newcommand{\Cob}[1]{\mathrm{Cob}_{#1}} 

\newcommand{\FG}{\mathsf{Z}_G}
\newcommand{\bFG}{\mathsf{Z}_G^{!}}
\newcommand{\bbFG}{\mathsf{Z}_G^{!!}}
\newcommand{\tqft}{\mathsf{Z}_{G}}

\newcommand{\VV}{\mathcal{V}}

\newcommand{\Motion}{\mathrm{Mot}}
\newcommand{\Mot}[2]{\Motion_{#1}^{#2}}
\newcommand{\mcg}[2]{\mathrm{MCG}_{#1}^{#2}}

\newcommand{\hfMot}[2]{\mathrm{hf}\Mot{#1}{#2}}
\newcommand{\hfmcg}[2]{\mathrm{hf}\mcg{#1}{#2}}

\newcommand{\simrp}{\stackrel{\scriptscriptstyle{rp}}{\sim}}  
\newcommand{\classrp}[1]{[#1]_{\!\mbox{\tiny rp}}} 
\newcommand{\simi}{\stackrel{i}{\sim}}  
\newcommand{\classi}[1]{[#1]_{\!\mbox{\tiny i}}} 
\newcommand{\sh}[1]{\mathfrak{#1}}
\newcommand{\shmor}[4]{\mathfrak{#1}^{#2}\colon #3 \stackrel{\lcurvearrowright}{} #4}
\newcommand{\shmot}[4]{#1
	\colon #3 \stackrel{\lcurvearrowright}{} #4}

\newcommand{\mcgfunctor}[2]{\mathcal{MCG}_{#1}^{A}}
\newcommand{\Motfunctor}[2]{\mathcal{MOT}_{#1}^{A}}

\newcommand{\premo}[2]{\mathrm{Flow}_{#1}^{#2}}

\newcommand{\axiomM}{manifold}

\newcommand{\too}{\rcurvearrowright}    

\newcommand{\Mtcmag}{\mathrm{Mt}_M}

\newcommand{\mot}[4]{#1
	\colon #3\too #4}

\newcommand{\phomotopy}[3]{\Topo(\II^2, #1)\!(\! #2 (\!\!\! ) #3 ) }

\newcommand{\pathphisotop}[6]{\Topo(\II, \TOPO_{#6}^h(#1,#1))\!(\! {\tiny #2}^{#4}_{#3}\!\!  {\tiny #5} ) }

\newcommand{\rel}[4]{(#1^{#3}_{#2}\!#4)} 

\newcommand{\Id}{{\mathrm{Id}}}  

\newcommand{\W}{\mathbf{W}}
\newcommand{\Wc}{\mathbf{W}'}

\newcommand{\classm}[1]{[#1]_{\!\mbox{\tiny m}}}

\newcommand{\mott}[1]{\iota_{{#1}_t}}
\newcommand{\motzero}[1]{\iota_{{#1}_0}}
\newcommand{\motone}[1]{\iota_{{#1}_1}}
\newcommand{\pomotzero}[1]{\tilde{\iota}_{{#1}_0}}
\newcommand{\pomotone}[1]{\tilde{\iota}_{{#1}_1}}

\newcommand{\mcgim}[1]{\sh{#1}_{{\tiny \mathrm{mc}}}}
\newcommand{\mcgmotim}[1]{{#1}_{{\tiny \mathrm{mc}}}}

\newcommand{\F}{\mathsf{F}} 

\newcommand{\gammaf}{\sh{f}}
\newcommand{\gammaff}{\sh{g}}

\author{Fiona Torzewska\footnote{f.m.torzewska@leeds.ac.uk}
	\\ \myaddress
}
\date{}
\title{Topological quantum field theories and homotopy cobordisms}

\begin{document} \maketitle
	
\begin{abstract}	
	We construct a category $\HomCob$ whose objects are {\it \homfin } topological spaces, and whose morphisms are {\it cofibrant cospans}.
	Given a manifold submanifold pair $(M,A)$, we prove that there exists functors into $\HomCob$ from the full subgroupoid of the mapping class groupoid $\mcg{M}{A}$, and from the full subgroupoid of the motion groupoid $\Mot{M}{A}$, whose objects are \homfin{}.

	We also construct a family of functors $\FG\colon \HomCob\to \Vect$, one for each finite group $G$.
	 These generalise topological quantum field theories previously constructed by Yetter, and an untwisted version of Dijkgraaf-Witten. 
	 Given a space $X$, we prove that $\FG(X)$ can be expressed as the $\C$-vector space with basis natural transformation classes of maps $\{\pi(X,X_0)\to G\} $ for some finite representative set of points $X_0\subset X$, demonstrating that $\FG$ is explicitly calculable.
\end{abstract}

\medskip 

\noindent \textbf{Acknowledgements:} {FT was funded by a University of Leeds PhD Scholarship and is now funded by EPSRC.
	FT thanks Paul Martin and Jo\~{a}o Faria Martins for useful discussions. 
}

\tableofcontents

\section{Introduction}

This motivating aim of this work is to describe certain condensed matter systems known as topological phases.
This in turn has applications to quantum computation:
topological quantum computing (TQC) is a proposed framework for carrying out computation by controlling the movement of emergent particles in topological phases \cite{rowellwang, wen, nayak, Wilczek1982}.
A topological phase may be described by assigning a space of states to each possible particle configuration, and assigning a linear operator to each allowed particle trajectory. 
Several categories are potentially of interest for the modelling of particle configurations and trajectories in topological phases: motion groupoids \cite{motion,qiu}, tangle categories (\cite{Kassel} for example), generalised tangle categories \cite{BaezDolan}, defect cobordism categories \cite{Carqueville}, and embedded cobordism catgeories \cite{Picken,witten} for example.
The objective then is to find representations of these categories, in other words functors which eventually map to $\Vect_{\mathbb{C}}$.

Here we aim to construct representations of particle trajectories which are invariant up to a notion of homotopy equivalence of the complement of the particle trajectory.
Note that, in general complements of particle trajectories in the aforementioned categories will not be compact manifolds.
Field theories arising from homotopical properties of manifolds are common in the literature.
The topological quantum field theories (TQFTs) of \cite{yetter, Kitaevfault} and an untwisted version of \cite{dijkgraaf} can all be shown to assign to a manifold $\Sigma$ the vector space with basis natural transformation classes of maps 
$\pi(\Sigma)\to G$.
In fact, in $(3+1)$-dimensional it appears that all bosonic topological phases where all quasiparticles are bosons, are described by Dijkgraaf-Witten theories \cite{wen}. 
Each of these examples are generalised by \cite{quinn}, which gives a class of TQFTs constructed using the `homotopy content'.
Further, the approach of looking at the homotopy of the complement is taken in certain invariants of knots \cite{Crowell}, Artin's representation of braids \cite{Birman} and its lift to loop braids \cite{damiani21}, for example.
It is also possible to consider the complement whilst retaining extra structure, as is the situation with the knot quandle, which completely classifies a knot \cite{joyce}. 

Such a functor into $\VectC$ may factor through other categories.
In some cases, these categories will be more convenient to work with.
Our first main result (Theorem~\ref{th:HomCob}) is the construction of a symmetric, monoidal category $\HomCob$ whose objects are topological spaces and whose morphisms are equivalence classes of {\it cofibrant cospans}. Roughly, functors into the category $\HomCob$, from the aforementioned categories modelling topological phases, may be constructed by taking the complement of particle trajectories.
We prove that there exists functors into $\HomCob$ from the motion groupoid of a manifold, as constructed in \cite{motion}, as well as from $\Cobcat{n}$, the category of cobordisms as given in \cite{lurie}. 
The approach of constructing representations by going via categories of cospans has its origins in the work of Grandis, and Morton \cite{morton,grandis}; and these can each be seen as an example of the groupoidification explained in \cite{BaezHoffnung}.

In the second half of the paper, we construct a family of functors from $\HomCob$ into $\VectC$ based on a finite group $G$. Our construction follows the approach of \cite{yetter}. This allows us to give an interpretation of our functor in terms of a choice of a finite set of basepoints in each of the spaces involved, and hence a method for explicit calculation. We expect it to be possible to vastly generalise the functors constructed here, to functors which take a finite $n$-type as input (cf. \cite{yetter93}).

\medskip

We now introduce our construction in more detail.

Our first objective is the construction of a new category of cospans of topological spaces.
The `natural' formalism for such constructions depends on one's
perspective, i.e. upon one's aims.
For example we have the categorical/`join' perspective following Benabou
\cite{Benabou67}.
	One of Benabou's archetypes is the bicategory $\mathrm{Sp}(V)$ of spans over a category
$V$ with pullbacks and a distinguished choice of pullback for each span.
And a `dual', $\mathrm{Cosp}(V)$ of cospans over a category with pushouts and choices.
But this comes at a cost of inducing categories with properties 
that are undesirable in our setting.
Precisely, a manifold cobordism maps naturally to a cospan, but the image of the usual representative of the identity cobordism is of the form $M\to M\times [0,1] \xleftarrow{} M$.
In contrast the identity in $\mathrm{Cosp}(V)$ is of the form $ 1_M\colon M\to M \xleftarrow{} M\lon 1_M$, hence the mapping does not extend to a functor.
 General pushouts of topological spaces also fail to preserve homotopy.
If one follows this line then a fix ({\em the} fix, essentially tautologically)
is some form of `collaring'.
The approach of \cite{morton,grandis} is to explicitly construct such collars.
These can be compared with \cite{fong}, for example, whose decorated cospans do not include a collaring.
Here we proceed by instead insisting that all maps are cofibrations.

Physically this collaring can be thought of 
as conditions on the way we are allowed to `cut' spacetime into finite time slices; we require that sufficient information is retained at the cut to be able to reconstruct the full theory from evaluations on each slice. The existence of such a formalism is implied by the path integral formalism of quantum field theory \cite{Hibbs}. 
The appropriate conditions will depend on the particular field theory.

We show that there is a category (Theorem~\ref{th:CofCsp})
\[
\CofCsp=\left(Ob(\Topo),\CofCsp(X,Y),\; \sbullet \;, \; \classche{\cmortikz{X\times \II}{\iota_{0}^X}{X}{\iota_1^X}{X}}\;\right)
\]
whose objects are all topological spaces, elements in  $\CofCsp(X,Y)$ are equivalence classes of \textit{concrete cofibrant cospans} -- diagrams $\cmor{M}{i}{X}{j}{Y}$ with the condition that $\lan i, j\ran\colon X\sqcup Y\to M$, the map obtained from the universal property of the coproduct, is a closed cofibration \footnote{Note these are cofibrations in the Str\o{}m model structure on topological spaces (see \cite{dwyer}), and many of our results could alternatively be proved using tools from model categories.}. 
A large class of examples of cofibrant cospans come from cospans of CW complexes and subcomplexes, which are essentially used to construct the TQFT in \cite{quinn}.
Composition is via pushout.
We then obtain the category $\HomCob$ as a subcategory of $\HomCob$ with all spaces {\homfin{}} (Theorem~\ref{th:HomCob}).
Pushouts of cofibrations behave well with respect to homotopy, this is formalised by a version of the van Kampen theorem due to Brown \cite{brownt+g}, here it is Theorem~\ref{th:vk}. This is why there exists maps from $\HomCob$ defined using homotopical properties of the spaces involved which preserve composition.
The equivalence relation is a notion of homotopy equivalence which commutes with the cospans.
The category $\HomCob$ can be thought of as a generalisation of the usual category of cobordisms \cite{lurie};
precisely we show in Proposition~\ref{pr:fun_cob} that there is a functor into $\Cob{n}\colon \Cobcat{n} \to \HomCob$, thus functors $\HomCob \to \Vect$ restrict to ordinary TQFTs.
(An alternative approach is to work with cobordisms with corners, see for example \cite{morton}.)
We also show that $\Cob{n}$ is a symmetric monoidal functor, where the monoidal product on objects in $\HomCob$ is given by disjoint union.
 
In Section~\ref{sec:funcHomcob}, we prove that there exist functors 
\[\Motfunctor{M}{A}\colon \hfMot{M}{A} \to \HomCob \textrm{ and } \mcgfunctor{M}{A}\colon \hfmcg{M}{A}\to \HomCob,\] from respectively the motion groupoid and the mapping class groupoid of a manifold $M$ with fixed submanifold $A\subset M$, each with a finiteness condition on the objects of the categories. For reference, the motion group of a subset $X\subset M$ is the automorphism group $\Mot{M}{\emptyset}(X,X)$.
 In the case $M=\R^2$, and $X$ a finite subset, this is a group isomorphic to the braid group on $|X|$ strands, and in the case $M=\R^3$ and $X$ is some configuration of unknotted, unlinked loops, this is a group isomorphic to the loop braid group in a number of loops equal to the number of connected components of $X$. 
 Hence our results lead to representations of motion groupoids, in all dimensions and for all subsets.  Existing results on representations of the loop braid group are limited, and focus on representations obtained by going via finite dimensional quotients, and extending braid representations \cite{kadar,damianiloophecke, damiani21}.

In the second half of this paper we construct a family of symmetric monoidal functors $\FG\colon \HomCob \to \Vect$, one associated to each finite group $G$. 
The construction of our functor is largely based on the approach taken in \cite{yetter},  
although Yetter constructs the functor with triangulations of manifolds, whereas we work entirely with the fundamental groupoid, following \cite{morton}.
The novel part of our construction with respect to \cite{yetter} is our generalised source category.
Another key aspect of our construction is that we present our functor in an explicitly calculable way, in terms of a choice of basepoints.

We construct the map on $X\in Ob(\HomCob)$ as a colimit over all sets of groupoid maps  $\{f\colon \pi(X,X_0)\to G\}$ for each choice of finite representative set $X_0$.
In general the colimit construction is a global equivalence relation on an uncountably infinite set.
We prove that for each space $X\in \chi$ we can fix a choice of finite representative subset $X_0$, and then
	\[ \FG(X)=\C((\Grpd(\pi(X,X_0),G)/\cong),\]  where $\cong$ denotes taking maps up to natural transformation and $\pi(X,X_0)$ is the full subgroupoid of the fundamental groupoid with objects $X_0$ (Theorem~\ref{natiso_to_colim}). 
	This gives $\FG$ in terms of a local equivalence relation on a finite set, making explicit calculation possible.
	We also prove that, for $X\in Ob(\HomCob)$, $\FG(X)$ is isomorphic to the vector space with basis $\left\{\pi(X)\to G\right\}/\cong$ (Theorem~\ref{natiso_to_colim}).
	
	The equivalence class of the cospan $\cmor{M}{i}{X}{j}{Y}$ is mapped to the matrix where, for basis elements  $[f]\in \FG(X)$ and $[g]\in \FG(Y)$ the element in the column corresponding to $[f]$ and the row corresponding to $[g]$ is
	 	\begin{multline*}\label{eq:FG_sim}
		\langle [g] | \FG(M) | [f] \rangle \hspace*{-0.2em}=\hspace*{-0.2em} 
		|G|^{-(|M_0| - |X_0|)} \left\vert\left\{  h:\pi(M,M_0) \to G  \,|\, h\circ \pi(i)(\pi(X,X_0))=f \right. \right. \\ \left.\left. \wedge  h\circ\pi(j)(\pi(Y,Y_0))\sim g \right\}\right\vert
	\end{multline*}
and $M_0\subset M$ is finite subset with some conditions, $\pi(i)$ is the map of fundamental groupoids induced by $i$, and $\sim$ indicates equivalence up to natural transformation. This is (Lemma~\ref{le:FG_sum} and Remark~\ref{FG_equivalence}).

\subsection{Paper Overview}
In Section~\ref{sec:prelims} we have preliminaries, fixing notation and recalling results that we will make use of.
First we have magmoids, category like structures without axioms that we will find a useful starting point for our constructions.
Then in Section~\ref{sec:path} we have the fundamental groupoid, and a partial product-hom adjunction in the category of topological spaces.
In Section~\ref{sec:colim} we fix representative colimits in the categories we will need.
In Section~\ref{sec:moncats} we recall the definition of symmetric monoidal categories, and functors between them. 
Finally in Section~\ref{sec:cofibs+vk} we recall the definition of a cofibration, as well as giving some key properties. In particular, we have Corollary~\ref{co:vk}, which is a corollary of a version of the van Kampen theorem using cofibrations, due to \cite{brownt+g}.

In Section~\ref{sec:homcobs}, we begin by constructing a magmoid whose objects are topological spaces and whose morphisms are \textit{\ccc}s.
We then quotient by a congruence in terms of homotopy equivalences to obtain a category $\CofCsp$ which has \textit{cofibrant cospans} as morphisms, this is Theorem~\ref{th:CofCsp}.
We then have Theorem~\ref{th:CofCsp_symm} which proves there is a monoidal structure on $\CofCsp$ with monoidal product which, on objects, is given by disjoint union.
Next we have the category $\HomCob$ (Theorem~\ref{HomCob}) which is a subcategory of $\CofCsp$ with a finiteness condition on spaces. 
In Theorem~\ref{th:HomCob} we have a version with the monoidal structure from $\CofCsp$.

In Section~\ref{sec:funcHomcob} we prove that we have a functor $\Motfunctor{M}{A}\colon \hfMot{M}{A} \to \HomCob$ (Theorem~\ref{th:functorMot_HomCob}), and $\mcgfunctor{M}{A}\colon \mcg{M}{A} \to \HomCob$ (Lemma\ref{le:functormcg_HomCob}) for $M$ a manifold and $A\subset M$ a subset.
We also prove that we have $\Motfunctor{M}{A}=\F\circ \mcgfunctor{M}{A}$ where $\F\colon \Mot{M}{A}\to \mcg{M}{A}$ is as constructed in \cite[Sec.7]{motion}, this is Theorem~\ref{th:comp_mcgfunctorF}.

We begin Section~\ref{sec:tqft} by constructing, in Lemma~\ref{le:comp_bhomcob}, another magmoid which has as morphisms cospans of pairs of a topological space and a subset of basepoints, we call this $\bHomCob$.
We then construct a magmoid morphism $\bFG\colon \bHomCob \to \VectC$ (Lemma~\ref{le:tqft_bFGmagmor}), which depends on a finite group $G$.
Under $\bFG$, pairs $(X,X_0)$ are mapped to the vector space with basis the set of maps $\pi(X,X_0)$ to $G$.
We then take a colimit over a diagram whose vertices are indexed by all allowed sets of basepoints. This leads to a map $\FG\colon Ob(\HomCob)\to Ob(\VectC)$, given in Definition~\ref{de:FG_obj}.
We then extend this to morphisms so we have, in Lemma~\ref{le:tqft_FGmagmor}, a magmoid morphism $\FG\colon \HomCob \to \VectC$.
In Theorem~\ref{th:FG_functor} we prove equivalence is preserved and thus we have a functor $\FG\colon \HomCob \to \VectC$. 
Theorem~\ref{natiso_to_colim} gives the alternative interpretation of $\FG(X)$, as the vector space with basis $\left\{f\colon \pi(X,X_0)\to G\right\}/\cong$ for some choice of basepoints $X_0$, where $\cong$ denotes functors up to natural transformation.
Finally we have Theorem~\ref{natiso_to_colim} which says, for $X\in Ob(\HomCob)$, $\FG(X)$ is isomorphic to the vector space with basis $\left\{\pi(X)\to G\right\}/\cong$.

\section{Preliminaries}
\label{sec:prelims}

We spend some time on this section for number of reasons. One is that this paper draws from diverse areas of mathematics, and recalling the relevant results makes this work accessible to a wide audience.
Another is to take the opportunity to fix some non-standard definitions/ notation that will be helpful.
In each section we give references directing the reader to more complete approaches. 

We start, in Section~\ref{sec:magmoids}, with magmoids, and magmoid congruences which may lead to categories.
In Section~\ref{sec:path} we have the fundamental groupoid and a product-hom adjunction in the category of topological spaces.
In Section~\ref{sec:colim} we recall the definition of a colimit as well as fixing choices of specific colimits in the categories of sets, topological spaces and groupoids. We then have monoidal categories in Section~\ref{sec:moncats}.
Finally we define and recall some useful properties of cofibrations in Section~\ref{sec:cofibs+vk}, as well as giving examples which will demonstrate the flexibility of the subsequent construction.

\subsection{Magmoids, categories and groupoids}
\label{sec:magmoids}

In this work constructions of categories are a recurrent theme, although these are neither a convenient starting point, nor representative of the underlying physics. 
Such constructions will often start from something concrete with a composition. Equivalence classes of these concrete things eventually become the morphisms of the constructed category.
It will be useful to have a general machinery for studying such constructions. We can think of the underlying idea of a category as sets of objects, morphisms and a non-associative composition - a {\it magmoid}.
We can then study congruences on these magmoids, some of which will lead to categories.

We are unaware of a reference to magmoids in the literature although the construction is a straightforward extension of the use of magmas for the underlying structure of a group.
We assume familiarity with category theory such as product categories, adjunctions and natural transformations as well as with groupoids; more complete introductions can be found in e.g. \cite{maclane,riehl,adamekjoy}.

\defn{
	A {\em magmoid} ${\mM}$ 
	is a triple 
	$$
	{\mM} \; = \;  (Ob(\mM),\mM(-,-),\Delta_{\mM})
	$$ 
	consisting of
	\begin{itemize}
		\item[(I)] a collection $Ob(\mM)$ of \textit{objects},
		\item[(II)] for each pair $X,Y\in Ob(\mM)$ a collection $\mM(X,Y)$ of \textit{morphisms from $X$ to $Y$}, and 
		\item[(III)] for each triple $X,Y,Z \in Ob(\mM)$ a {\textit composition}
		\[
		\Delta_\mM\colon \mM(X,Y)\times \mM(Y,Z)\to \mM(X,Z).
		\]
	\end{itemize}
We use $f\colon X\to Y$ to indicate that $f$ is a morphism from $X$ to $Y$ and $f\in \mM$ to indicate there exists a pair of objects $X,Y\in Ob(M)$ such that $f\in \mM(X,Y)$. \\
Where convenient we will replace instances of $-$ in the triple with generic symbols.
}

\exa{\label{ex:SetTopVect_Mags}
		The following are magmoids. In each case we give the objects and morphisms, the composition is then the usual composition of maps of each structure.
		\begin{enumerate}[noitemsep, label={},topsep=0pt,label=(\roman*)]
			\item	$\sSet$: Objects are sets and morphisms from $X$ to $Y$ are all functions $f\colon X\to Y$.
			\item	$\vVectk$: Objects are $\mathbbm{k}-$vector spaces and morphisms from $V$ to $W$ are $\mathbbm{k}$-linear maps $f\colon V\to W$.
	\end{enumerate}}

\exa{\label{ex:Top_Mag}
	There is a magmoid 
	\[\tTopo=(Ob(\tTopo),\tTopo(-,-),\Delta_T)\]
	where $Ob(\tTopo)$ is the set of all topological spaces, for $X,Y\in Ob(\tTopo)$, $\tTopo(X,Y)$ is the set of continuous maps from $X$ to $Y$, and $\Delta_T$ is given by the composition of the underlying functions in $\sSet$.
}

\defn{
	A magmoid $\mM$ is called {\em small} if $Ob(M)$ is a set and for each pair $X,Y\in Ob(\mM)$, $\mM(X,Y)$ is a set.}

\defn{Let $\mM$ and $\mM'$ be magmoids. A {\em magmoid morphism} 
	$F\colon \mM\to \mM'$
	is a map sending each object $X\in Ob(\mM)$ to an object $F(X)\in Ob(\mM')$ and each morphism $f\colon X\to Y$ in $\mM$ to a morphism $F(f)\colon F(X)\to F(Y)$ in $\mM'$ such that for any pair of composable morphisms $f,g \in \mM$ 
	\[ F(\Delta_{\mM}( f,g))= \Delta_{\mM'}( F(f),F(g)).
	\]
}

Let $\mM=(Ob(\mM),\mM(-,-),*_{\mM})$ be a magmoid. We will find it convenient to have an alternative notation for composition, for which we use function order. For composable morphisms $f$ and $g$ in $\mM$, 
\[
g*_\mM f\coloneq *_\mM(f,g) .
\]

It is straightforward to check that for a magmoid  ${\mM}=(Ob(\mM),\mM(-,-),*_{\mM})$, there exists a magmoid $\mM^{op}=(Ob(\mM),\mM^{op}(-,-),*_{\mM^{op}})$ where for $X,Y\in Ob(\mM)$, $\mM^{op}(X,Y)=\mM(Y,X)$ and for composable morphisms $f,g\in \mM^{op}$ we have $g*_{\mM^{op}}f=f*_\mM g$.

\defn{A magmoid ${\mM} =  (Ob(\mM),\mM(-,-),\Delta_{\mM})$ is called {\em reversible} if there exists a
	constructible family of bijections (possibly of proper classes)
	\[
	\rev_{N,N'}\colon \mM(N,N')\to \mM(N',N)
	\]
	for all pairs $N,N'\in Ob(\mM)$.}

We now give the familiar definitions of a category and a functor in terms of  magmoids.

\defn{
	A {\em category} is a quadruple 
	\[{\CC}=(Ob(\CC),\CC(-,-),*_{\CC},1_{-})\]
	 consisting of a magmoid $(Ob(\CC),\CC(-,-),*_{\CC})$ and
	\begin{itemize}
		\item[(IV)] for each $X\in Ob(\CC)$
		a distinguished morphism $1_X\in \CC(X,X)$
		called the \textit{identity},
	\end{itemize}
	such that the following axioms are satisfied.
	\begin{itemize}
		\item[$(\CC1)$] \textbf{Identity:} for any morphism $f\colon X \to Y$, we have $1_Y*_\CC f = f = f*_\CC 1_Y$.
		\item[$(\CC2)$] \textbf{Associativity:} for any triple of morphisms $ f\colon X\to Y$, $g\colon Y \to Z$ and $h\colon Z\to W$ we have $h*_\CC(g*_\CC f)=(h*_\CC g)*_\CC f$.
	\end{itemize}
We refer to $(Ob(\CC),\CC(-,-),*_{\CC})$ as the underlying magmoid of $\CC$.
By abuse of notation we refer also to the underlying magmoid as $\CC$. 
}
\defn{Let $\CC$ and $\CC'$ be categories. A {\em functor} 
	$F\colon \CC\to \CC'$
	is a magmoid morphism such that for any $X\in Ob(\CC)$
	\[ 1_{F(X)}=F(1_{X}).\]
}

It is straightforward to show that there exists a category with underlying magmoid $\sSet$ and the usual identity functions, we denote this $\Set$.
Analogously $\Vectk$ has underlying magmoid $\vVectk$, and $\Topo$ has underlying magmoid $\tTopo$.
We denote the identity map on $X\in Ob(\Topo)$ by $\id_X\colon X\to X$.

\medskip
It is also straightforward to check that there is an associative partial composition of magmoid morphisms which sends a pair of magmoid morphisms $F\colon \mM\to \mM'$ and $F'\colon \mM'\to \mM''$ to $F'\circ F\colon \mM\to \mM''$ defined in the obvious way using the composition in $\Set$ of the underlying maps. It is interesting to note that this composition, leads to  category 
$\MAG=(Ob(\MAG),\MAG(-,-),\circ,1_{-})$
whose objects are all small magmoids,  and whose identities are morphisms comprising of identity maps in $\Set$. Thus magmoid representation theory can be framed in terms of magmoid morphisms in this category.

It can similarly be checked that the partial composition of magmoid morphisms extends to a partial composition of functors, we denote this also with $\circ$.
We denote by $\Cat$ the analogously defined category of small categories. 

Where we feel a more concise notation is helpful we may use the null composition symbol (i.e. just juxtaposition) for composition of functors and magmoid morphisms.

\begin{definition}
	A category $\CC$ is called {\em finitely generated} if there exists a finite set $X$ of morphisms (including identities) in $\CC$ such that every morphism in $\CC$ can be obtained by composing morphisms in $X$. 
\end{definition}

We give the familiar definition of a groupoid in order to fix notation. Recall that a category is called {\em small} if the collection of all morphisms is a set.

\defn{
	A {\em groupoid} ${\GG}$ is a pentuple 
	\[
	\GG=(Ob(\GG),\GG(-,-),*_\GG,1_-,(-)\mapsto (-)^{-1})
	\]
	consisting of a small category $(Ob(\GG),\GG(-,-),*_\GG,1_{-})$, and 
	\begin{itemize}
		\item[(V)] for each pair $(X,Y)\in Ob(\GG)\times Ob(\GG)$ a function
		\ali{
			(-)^{-1}\colon \GG(X,Y)&\to \GG(Y,X)\\
			f&\mapsto f^{-1}
		}
		called the \textit{inverse assigning} function;
	\end{itemize}
	such that the following is satisfied.
	\begin{itemize}
		\item [$(\GG1)$] \textbf{Inverse:}
		for any morphism $f\colon X\to Y$, we have $f^{-1}*_\GG f=1_X$ and $f*_\GG f^{-1}=1_Y$.
	\end{itemize}
}

In other words, a groupoid is precisely a small category in which all morphisms are isomorphisms.

\rem{Let $\GG$ be a groupoid. By abuse of notation we will refer also to the underlying magmoid as $\GG$. 
	Note that the identities and inverses of $\GG$ are uniquely determined from the underlying magmoid of $\GG$.}

\rem{Notice that the underlying magmoid of a groupoid is necessarily reversible, although a reversible magmoid $\mM$ does not imply the existence of a groupoid with underlying magmoid $\mM$.}

A magmoid morphism of underlying magmoids always lifts to functor preserving the groupoid structure. The proof is straightforward.

\begin{proposition}\label{pr:grpd_fun_comp}
	Let $\GG$ and $\GG'$ be groupoids and $F\colon \GG\to \GG'$ a magmoid morphism.
	Then we have 
	\begin{itemize}
		\item for any $X\in Ob(\GG)$, $1_F(X)=F(1_X)$, and 
		\item for any morphism $f\in \GG$, $F(f^{-1})=(F(f))^{-1}$.\qed
	\end{itemize}
\end{proposition}

\rem{The previous result cannot be replicated for categories: a magmoid morphism between the underlying magmoids of a pair of categories can fail to be a functor.}

It is straightforward to see that there is a full subcategory of $\Cat$ whose objects are groupoids. We denote this $\Grpd$.

\subsubsection{Magmoid congruence}
\label{sec:congruences}

Often the magmoids we construct are too large to be interesting objects of study themselves. Here we introduce congruences and quotient magmoids, our main tool for obtaining a category from a magmoid.
Congruences are families of relations on the morphism sets of magmoids. Note in particular that the object set is always fixed. Allowing equivalence relations on objects in magmoids potentially leads to extra morphisms, and so is not really a quotient in the usual sense.

As with the previous section we are unaware of a reference that explicitly discusses congruences of structures which are not yet categories.
However,
if we take a category and quotient the underlying magmoid by a congruence, the quotient magmoid can also be given a categorical structure in a canonical way (Proposition~\ref{pr:cat_cong}), in this case we obtain the same quotient category as in \cite[Ch.2]{maclane}.

\defn{A {\em congruence} $C$ on a magmoid $\mM=(Ob(\mM),\mM(-,-),\Delta_{\mM})$ consists of, for each pair $X,Y\in Ob(\mM)$ an equivalence relation $R_{X,Y}$ on $\mM(X,Y)$, such that 
	$f'\in[f]$ and $g'\in[g]$ implies $\Delta_{\mM}(f',g')\in [\Delta_{\mM}(f,g)]$ where defined.
}

\defn{
	Let $\mM$ be a magmoid and $R=\{R_{X,Y}\}_{X,Y\in Ob(\mM) }$ a 
	collection of relations on the sets $\mM(X,Y)$.
	Then let $\bar{R}$ be the closure of $R$ to a congruence, this means we take the reflexive, symmetric, transitive closure of each $R_{X,Y}$ and insist that for any composition $\Delta_\mM(f,g)\sim \Delta_\mM(f',g')$ if $f\sim f'$ and $g\sim g'$. 
}

\defn{Let $\mM=(Ob(\mM),\mM(-,-),\Delta_{\mM})$ be a magmoid and $C$ a congruence on $\mM$. The {\em quotient magmoid} of $\mM$ by $C$ is $\mM/C=(Ob(\mM),\mM(X,Y)/R_{X,Y},\Delta_{\mM/C})$ where for each triple $X,Y,Z\in Ob(\mM)$ 
		\ali{
			\Delta_{\mM/C}\colon\mM/C(X,Y)\times \mM/C(Y,Z)&\to \mM/C(X,Z)\\
			([f],[g])&\mapsto [\Delta_{\mM}(f,g)].	
		}
		(That the composition is well defined follows directly from the definition of a congruence.)
}
In practice we will use the notation for the composition in $\mM$ to also denote the composition $\mM/C$.

\medskip

It is straightforward to prove the following:
\prop{\label{pr:cat_cong}
	Suppose $\CC=(Ob(\CC),\CC(X,Y),*_\CC,1_{-})$ is a category. For any congruence $C$ on $(Ob(\CC),\CC(X,Y),*_\CC)$, we have that $\CC/C=(Ob(\CC),\CC(X,Y)/R_{X,Y},*_{\CC/C},[1_{-}])$
	is a category.\\
	We call this $\CC/C$ the {\em quotient category} of $\CC$ by $C$.
	\qed}

\subsection{
	A product-hom adjunction in \texorpdfstring{$\Topo$}{Top} and the fundamental groupoid}
\label{sec:path} \label{sec:tensor_hom}
We begin by recalling a partial lift of the classical product-hom adjunction in $\Set$ to $\Topo$. We will make heavy use of this result in Section~\ref{sec:funcHomcob}.

We then recall the construction of the fundamental groupoid (Proposition \ref{pr:fundamentalgroupoid}).
Spanier \cite{spanier} and Brown \cite{brownt+g} were among the first to consider fundamental groupoids, and careful constructions can be found in \cite{dieck} and \cite{brownt+g} for example.
Here we also discuss the relationship between fundamental groupoids obtained by varying a finite number of basepoints, which will be necessary for the construction in Section~\ref{sec:tqft} (Lemmas~\ref{le:extension} and \ref{le:add_bps}).

Throughout the rest of this paper we will encounter several equivalence relations so we also introduce some careful notation here. 

\medskip

For spaces $X,Y$, we use $\TOPO(X,Y)$ to denote the set $\Topo(X,Y)$ together with the compact open topology. For the precise definition see \cite[p.285]{munkres2}; a useful characterisation is given by the fact that if $Y$ is a metric space, the compact-open topology coincides with the topology obtained from the sup-norm metric on $\Topo(X,Y)$ --  the distance between two maps $f,g$ is given by the least upper bound of the distance between all pairs $f(x), g(x)$.

There exists a functor  
$-\times Y\colon\Topo \to \Topo$ called the \textit{ product functor} which sends a space $X$ to the product space $X\times Y$, and a continuous map $f\colon X\to X'$ to the map 
$f\times \id_Y \colon X\times Y\to X'\times Y$,
$(x,y)\mapsto (f(x),y)$.
There also exists a functor $\TOPO(Y,-)\colon\Topo \to \Topo$  called the \textit{hom functor} which sends a space $Z$ to the space $ \TOPO(Y,Z)$, and which sends a continuous map $f\colon Z\to Z'$ to
$f\circ -\colon \TOPO(Y,Z)\to \TOPO(Y,Z')$,
$g\mapsto f\circ g$.

We will make extensive use of the following classical result in Section~\ref{sec:funcHomcob}. For full details see for example Section~2.2 of \cite{motion}.
\lemm{\label{le:producthom}
	Let $Y$ be a locally compact Hausdorff topological space.
	The product functor 
	$-\times Y $
	is left adjoint to the hom functor $\TOPO(Y,-)$.
	In particular,
	for objects $X,Y,Z\in\Topo$,
	this gives a set map
	\ali{
		\Phi\colon\Topo(X, \TOPO(Y,Z)) &\to \Topo(X\times Y, Z) \\
		f &\mapsto ((x,y)\mapsto f(x)(y))
	}
	that is a bijection{, natural in the variables $X$ and $Z$}.
}

\subsubsection*{The fundamental groupoid}
	\defn{We define the topological space $\II$ as $[0,1]\subset \R$ with the subset topology.}

\defn{ \label{de:pathspace}
	Let $X$ be a topological space.
	An element of $\Topo(\II,X)$ is called a {\em path} in $X$.\\
	We will use $\gamma_t$ for $\gamma(t)$, and we say $\gamma$ is a path from $x$ to $x'$,
	denoted $\gamma \colon x \to x'$, 
	when $\gamma_0 = x$ and $\gamma_1=x'$.
	For $x,x'\in X$, let 
	\[\Path X(x,x')=\{\gamma \colon \II \to X \; \vert \; \gamma\in \Topo(\II,X), \, \gamma_0=x, \, \gamma_1=x'\}.
	\]
}

Our convention for path composition is as follows:

\prop{ \label{pr:path_comp}
	Let $X$ be a topological space.
	For any $x,x',x''\in X$, there exists a composition 
	\ali{
		\PsiOps \colon \Path X(x,x')\times \Path X(x',x'') &\to \Path X(x,x'')\\
		(\gamma, \gamma')&\mapsto \gamma'\gamma
	}
	with 
	\begin{align*} 
		\label{eq:pathcomp} 
		(\gamma'\gamma)_t= \begin{cases}
			\gamma_{2t} & 0\leq t\leq 1/2, \\
			\gamma'_{2(t-1/2)} & 1/2\leq t \leq 1.
		\end{cases}
	\end{align*}	
}
\begin{proof}
	Straightforward.
\end{proof}

Thus $\Path X  \;=\; (X,\Path X(-,-), \PsiOps ) $ is a magmoid. In particular, it is a reversible magmoid, which can be seen by mapping $\gamma \mapsto \gamma^{rev}$ where $\gamma^{rev}_t=\gamma_{1-t}$.

\medskip

Let $X$ be a space. 
Paths 
$\gamma,\gamma'\in \Topo(\II,X)$
are {\it path homotopic}
if $ \phomotopy{X}{\gamma}{\gamma'} \neq \emptyset$ where
$$ 
\phomotopy{X}{\gamma}{\gamma'} 
\coloneq \{ H\in \Topo(\II^2,X) \;|\; H(-,0)=\gamma, H(-,1)=\gamma', H(0,-)=\gamma_0, H(1,-)=\gamma_1  
\}.
$$

It is straightforward to show that path homotopy is an equivalence relation on each $\Path X(-,-)$, and leads to a congruence on $\Path X$.  

If $\gamma$ and $\gamma'$ are path homotopic, we write $\gamma \simp \gamma'$, and we use $\classp{\gamma}$ for the path-equivalence class of $\gamma$.
Where we feel it simplifies the exposition, we may also use $\gamma$ for the path-equivalence class of $\gamma$. 

\medskip

It is straightforward to construct explicit homotopies proving that the quotient magmoid  $\Path X / \simp$  is moreover a groupoid:

\prop{ \label{pr:fundamentalgroupoid}
	Let $X$ be a topological space. 
	There exists a groupoid
	\[
	\pi(X) \;\coloneq\; (X,\Path X(-,-)/\simp, \PsiOps ,\classp{e_x},\classp{\gamma}\mapsto \classp{\gamma^{\rev}}) .
	\]
	Here the identity morphism $\classp{e_x}$ for $x\in X$ is the path-equivalence class of the constant path $\gamma_t=x$ for all $t\in \II$.
	This is the {\em fundamental groupoid of $X$}. \qed
}

\rem{
	Notice that, by Lemma~\ref{le:producthom}, we have $\gamma\simp \gamma'$ if and only if there is a path
$\tilde{H}\colon \II\to \TOPO(\II,X)$ such that 
$\tilde{H}(0)=\gamma$, $\tilde{H}(1)=\gamma'$ and for all $t\in \II$, $\tilde{H}(t)\in \Path X(x,x')$. }

It is straightforward to check that the map sending a space to its fundamental groupoid extends to a functor.

\begin{lemma}
	There is a functor $\pi\colon \mathbf{Top} \to \mathbf{Grpd}$ which sends a space $X$ to the fundamental groupoid $\pi(X)$ and is defined on morphisms as follows.
	Let $f\colon X \to Y$ be a continuous map, $\pi(f)\colon \pi(X) \to \pi (Y)  $ defined by $x\mapsto f(x)$ for a point $x\in X=Ob(\pi(X))$ and by $\classp{\gamma}\mapsto \classp{ f \circ \gamma }$ for a path $\gamma$ in $X$. \qed
\end{lemma}

\defn{\label{de:Afundgrpd}
	Let $X$ be a topological space and $A\subseteq X$ a subset.
	The {\em fundamental groupoid of $X$ with respect to $A$} is the full subgroupoid of $\pi(X)$ with object set $A$, denoted $\pi (X,A)$.\\
	We refer to $A$ as the set of basepoints.
}

We have $\pi(X,X)=\pi(X)$.
Let $X$ be a path-connected topological space and $x\in X$ be a point, then
$ \pi(X)(x,x)
$
is the fundamental group based at $x \in X$.
Notice that for any $A'\subseteq A$, there is an inclusion $\iota\colon\pi(X,A')\to \pi(X,A)$.

\begin{definition}\label{de:representative}
	For topological spaces $X$ and $A\subseteq X$, then $A$ is called  \textit{representative} in $X$ if $A$ contains a point in every path-component of $X$.
	(The nomenclature $(X,A)$ is a $0$-connected pair is also used.)
\end{definition}

\begin{lemma}\label{le:surjection_of_representative_set}
	Suppose $f\colon X\to Y$ is a surjection and $A$ is a representative subset of $X$,
	then $f(A)$ is representative in $Y$.
\end{lemma}	
\begin{proof}
	Let $y\in Y$ be any point. We must construct a path from $y$ to an element of $f(A)$.
	Let $y'\in f^{-1}(y)$ be any preimage,
	then there exists a path $\gamma$ from $y'$ to a point in $A$ and $f\circ \gamma$ is a path from $y$ to an element of $f(A)$. 
\end{proof}

We will need the following results about the fundamental groupoid with finite sets of basepoints in Section~\ref{sec:tqft}. 
In what follows we use the same labels for paths, and their equivalence classes to keep the notation readable. The meaning will be clear from context.
\begin{lemma}\label{le:extension}
	Let $\GG$ be a groupoid, $X$ a topological space, $X_0\subseteq X$ a finite subset and $y \in X \setminus X_0$ any point.
	Given a groupoid map $f\colon\pi(X,X_0) \to \GG$, a path $\gamma\colon x \to y$ where $x\in X_0$ and a morphism $g\colon f(x)\to \mathsf{g}$ in $\GG$ 
	there exists a unique $F\colon \pi(X,X_0 \cup \{y\}) \to \GG$ extending $f$ such that
	\begin{itemize}
		\item the diagram 
		\begin{eqnarray}
			\begin{tikzcd}
				\& \pi(X, X_0 \cup \{y\}) \ar[dr,"F"] \\
				\pi(X,X_0) \ar[ur,"\iota"] \ar[rr,"f"] \& \& \GG
			\end{tikzcd}
		\end{eqnarray}
		commutes, where $\iota$ is the inclusion map, and 
		\item $F(\gamma) = g $.
	\end{itemize}
\end{lemma}
\begin{proof}
	First we construct such an $F$. 
	On objects we have,
	\begin{equation*}
		F(a)= 
		\begin{cases}
			\mathsf{g}, &  \text{if } a=y \\
			f(a), &  \text{otherwise} . 
		\end{cases}
	\end{equation*}
	For a path $\phi\colon a\to y$ with $a \in X_0$ we must have
	\begin{align*}
		F(\phi)=F(\gamma\gamma^{-1}\phi ) = F(\gamma)F(\gamma^{-1}\phi)=gf(\gamma^{-1}\phi)
	\end{align*}
	Arguing similarly for all cases we have that for a morphism $\phi\colon a\to b$, 
	\begin{align*}
		F(\phi)=\begin{cases}
			f(\phi), & \text{if } a,b \in X_0 \\
			gf(\gamma^{-1}\phi), & \text{if } a \in X_0, \, b=y  \\
			f(\phi \gamma)g^{-1}, & \text{if } a=y, \, b\in X_0 \\
			gf(\gamma^{-1}\phi\gamma)g^{-1}, & \text{if } a=y, \, b=y.
		\end{cases}
	\end{align*}
	Notice that in each case $F$ is inferred from the conditions set out in the theorem and by functoriality.
	This gives uniqueness. Now it remains to check that functoriality is always preserved, i.e. for any two composable paths $\phi,\phi' \in \pi(X,X_0 \cup \{y\})$ we have $F(\phi') F(\phi)=F(\phi' \phi)$. This can be checked case by case, we give two examples, the other cases are checked similarly. 
	\begin{itemize}
		\item[(I)] If we have $\phi\colon y \to y$, $ \phi' \colon y \to y$, then
		\begin{multline*}
		F(\phi')F(\phi) = gf(\gamma^{-1}\phi'\gamma)g^{-1}gf(\gamma^{-1}\phi\gamma)g^{-1}=gf(\gamma^{-1}\phi'\gamma\gamma^{-1}\phi\gamma) g^{-1} \\ = gf(\gamma^{-1}\phi'\phi\gamma)g^{-1} =F(\phi' \phi). 
		\end{multline*}
		\item[(II)] If we have $\phi: a \to y, a \in X_0$, $ \phi' : y \to b$, $b\in X_0$, then 
		\begin{align*}
			F(\phi')F(\phi)&= f(\phi' \gamma)g^{-1}g f(\gamma^{-1}\phi) = f(\phi' \gamma\gamma^{-1}\phi) = f(\phi'\phi)=F(\phi'\phi). \qedhere
		\end{align*}
	\end{itemize}
\end{proof}
Given a group $G$, there is a groupoid $\GG_G=(\{*\},\GG_G(*,*),\circ_G,e_G,g\mapsto g^{-1})$, where $\GG_G(*,*)$ is the underlying set of $G$.
\begin{lemma}\label{le:add_bps}
	Let $X$ be a topological space, $G$ a group,
	$X_0\subseteq X$ a finite representative subset and $y\in X$ a point with with $y\notin X_0$.
	There is a non-canonical bijection of sets
	\ali{
		\Theta_{\gamma}\colon \Grpd(\pi(X,X_0),\GG_G)\times G &\to \Grpd (\pi(X,X_0\cup\{y\}),\GG_G)\\
		(f,g)&\mapsto F
	}
	where $\gamma$ is a choice of a path from some $x\in X_0$ to $y$ and $F$ is the extension along $\gamma$ and $g$ as described in Lemma~\ref{le:extension}. 
	
\end{lemma}

\begin{proof}
	First notice that for any $g\in G$, $g\colon f(x)\to f(x)$ is a morphism in $\GG_G$, so the map is well-defined.
	
	The map $\Theta_{\gamma}$ has inverse which sends a map $f'\in \Grpd (\pi(X,X_0\cup\{y\}),G)$ to the pair $(f'\circ \iota,f'(\gamma))$ where $\iota\colon \pi(X,X_0)\to \pi(X,X_0\cup\{y\})$ is the inclusion.
\end{proof}

\subsection{Colimits}
\label{sec:colim}
Colimits will play an integral role in the construction in Section~\ref{sec:tqft} so we use this section to recall the definition and review some key properties that we will use.
We also fix representative colimits in the categories we will work with throughout the paper.

The topics covered here can be found, for example, in \cite[Ch.3]{perrone}.

\begin{definition}
	Let $\CC$ be a category and $\mathbf{I}$ a small category. 
	A functor $D\colon \mathbf{I} \to \CC$ is called a {\em diagram} in $\CC$ of shape $\mathbf{I}$.
\end{definition}

Let $\mathbf{P} = \bullet \leftarrow \bullet \to \bullet$ be a category with three objects and two non identity morphisms as shown.
Then a functor $\mathbf{P} \to \CC$ for some category $\CC$ is uniquely specified by drawing a diagram
\[
\begin{tikzcd}[ampersand replacement=\&]
	C_1 \& C_0 \ar[l,"f_1"']\ar[r,"f_2"] \& C_2
\end{tikzcd}
\] 
in $\CC$ of the same shape as $\mathbf{P}$, hence the nomenclature.

\begin{definition}
	Let $\CC$ be a category, $\mathbf{I}$ a small category and $D\colon \mathbf{I} \to \CC$ a diagram in $\CC$.
	A {\em cocone} is an object $C \in Ob(\CC)$ together with a family of morphisms
	\[
	\phi=\left( \phi_i\colon D(i) \to C\right)_{i\in Ob(\mathbf{I})}
	\]
	indexed by the objects in $\mathbf{I}$,
	such that for all morphisms 
	$f\colon i \to j$ in $\mathbf{I}$
	the following triangle commutes.
	\[
	\begin{tikzcd}[ampersand replacement=\&,column sep =small ]
		D(i) \ar[rr,"D(f)"]\ar[dr,"\phi_i"'] \&
		\& D(j) \ar[dl,"\phi_j"]\\
		\& C \&
	\end{tikzcd}
	\]
	A {\em colimit} of $D$ is a cocone $(C,\phi)$ with the universal property that for any other cocone $(C',\psi)$ there exists a unique morphism $C\to C'$ making the following diagram commute for all morphisms $f\colon i \to j$ in $\mathbf{I}$.
	\[
	\begin{tikzcd}[ampersand replacement=\&]
		D(i)\ar[rr,"D(f)"]\ar[dr,"\phi_i"']\ar[bend right,ddr,"\psi_i"']\& \& D(j)\ar[dl,"\phi_j"]\ar[bend left,ddl,"\psi_j"]\\
		\& C\ar[dashed, d,"\exists !"] \& \\
		\& C' \&
	\end{tikzcd}
	\]
	We will refer to the object $C$ as $\colim(D)$.
\end{definition}

\begin{definition}
	A colimit of a diagram of shape $\mathbf{P}$ is a {\em pushout}.\end{definition}
\begin{definition}	
	Let $\mathbf{T}$ be the category with two objects and no non identity morphisms, then a colimit of a diagram of shape $\mathbf{T}$ is a {\em coproduct.}
	\end{definition}

\begin{definition}
	Let $\mathbf{E}=
	\begin{tikzcd}
		\bullet\ar[r,shift left]\ar[r,shift right] \& \bullet
	\end{tikzcd}$
	be the category with two objects and two non identity morphisms as shown.
	A colimit of a diagram of shape $E$ is called a {\em coequaliser}.
\end{definition}

In general colimits do not exist. 
However it is straightforward to show that, where colimits do exist, they are unique up to isomorphism.
Thus, where colimits do exist, we are free to choose a
representative element of each isomorphism class
to work with. 
Indeed we will {\it fix} representative elements for various colimits so a coproduct is uniquely defined by giving the appropriate diagram.
For example, in $\Set$ we will  
fix the
representative coproduct of a pair $X,Y\in Ob(\Set)$, to be the disjoint union
$$
X\sqcup Y \; := \left( X\times \{1\} \right) \cup 
\left( Y \times \{2\} \right) ,
$$
with the natural inclusions (i.e. those given by
$\iota_i(x) \;\coloneq(x,i)$).

\medskip
We explain our convention for pushouts in $\Set$.

Consider morphisms 
$f\colon Z \to X$ and $g\colon Z\to Y$ in $\Set$. 
Then the diagram below:
\[
\begin{tikzcd}
	Z \ar[r,"f"]\ar[d,"g"'] \& X \ar[d,"p_X"]  \\
	Y \ar[r,"p_Y"'] \& X \sqcup_{Z} Y
\end{tikzcd}
\]
is a pushout in $\Set$.
Here 
\[
X \sqcup_Z Y \; 
\coloneq (X\sqcup Y)/\sim,
\]
where $\sim$ is the reflexive, symmetric, transitive closure of the relation
\[
\big\{\big (\iota_1(f(z)),\iota_2(g(z))\big) \mid  z \in Z \big\},
\]
on $X\sqcup Y$, 
$p_X(x)$ is the equivalence class of $\iota_1(x)$ in $X\sqcup Y/\sim$
and $p_Y(y)$ is the equivalence class of $\iota_2(y)$ in $X\sqcup Y/\sim$.

We will fix the following representative general colimit in $\Set$. 

	Let $D\colon \mathbf{I}\to \Set$ be a diagram. Denote by $\sqcup_{i\in Ob(\mathbf{I})}D(i)$ the set of all pairs $(x,i)$, $i\in Ob(\mathbf{I})$, $x\in D(i)$. This is the {\em disjoint union} of the $D(i)$, and is a colimit of $D$ if $D$ has no non-identity morphisms.

For $i\in Ob(\mathbf{I})$, let $\iota_i\colon D(i) \to \sqcup_{i\in Ob(\mathbf{I})}D(i)$ denote the map $x\mapsto (x,i)$. 
Consider the relation 
\[
R=\left\{(\iota_i(x), \iota_j (D(f)(x))
\; \vert \;f\colon i\to j \in \mathbf{I}\right\}
\]
on $\sqcup_{i\in\mathbf{I}}D(i)$.
The colimit of $D$ is given by
\[
\colim(D)=\nicefrac{\sqcup_{i\in\mathbf{I}}D(i)}{\bar{R}}
\]
with maps $\phi_i\colon D(i)\to \nicefrac{\sqcup_{i\in\mathbf{I}}D(i)}{\bar{R}}$ which send $x$ to the equivalence class of $\iota_i(x)$. 

\medskip

The following theorem says that adjunctions interact nicely with colimits.

\begin{theorem}\label{th:lapcl}
	(\cite[Thm.~4.5.3]{riehl}) 
	Left adjoints preserve colimits. This means for any left adjoint $F\colon \CC\to \DD$, any diagram $D\colon \mathbf{I}\to \CC$ then 
\[
F(\colim(D))=\colim(F\circ D).
\] 
\end{theorem}	

\subsubsection{Colimits in \texorpdfstring{$\Topo$}{Top}}
\label{sec:TopColims}

The following well known adjunction allows us to write colimits in $\Topo$ in terms of colimits in $\Set$. 
	\exa{\label{ex:functor_TopSet}
	There is a forgetful functor $\Topoforget\colon \Topo \to \Set$ which sends a space to its underlying set and a continuous map to its underlying function.}

\lemm{\label{le:right_adjoint_Topoforget}
	$(I)$ There is a functor $\mathrm{G_T}\colon \Set\to \Topo$ which sends a set $X$ to the space with underlying set $X$ and the indiscrete topology, and which sends a function to the map which has the same action on the underlying sets. \\
	$(II)$ The functor $\mathrm{G_T}$ is right adjoint to $U_T$.
}

\begin{proof}
	(I) For sets $X$ and $Y$, let $f\colon X\to Y$ be a function. Then $f^{-1}(Y)=X$ and $f^{-1}(\emptyset)=\emptyset$ so $f$ defines a continuous function from $\mathrm{G_T}(X)$ to $\mathrm{G_T}(Y)$.
	Thus $\mathrm{G_T}$ is well defined.
	Clearly $\mathrm{G_T}$ sends identities to identities. Preservation of composition follows immediately from the definition.\\
	(II) The required family of set bijections send functions to continuous maps which act in the same way on the underlying set. It is straightforward to check the image of a function is continuous, and that the bijection is natural. 
\end{proof}

Since $\Topoforget\colon \Topo \to \Set$ is a left adjoint, by Theorem~\ref{th:lapcl}, $\Topoforget$ preserves colimits.
This means coproducts and pushouts of diagrams in $\Topo$ have the same underlying set as the coproducts and pushouts of their images in $\Set$.

Let $X$ and $Y$ be spaces.
Then 
\[
\tau_{X\sqcup Y}:=
\left\{ U \subseteq X\sqcup Y \;\vert\;
 \iota_1^{-1}(U) \text{ is closed in } X  \text{ and } \iota_2^{-1}(U) \text{ is closed in }  Y \right\}\]
is a topology on $X\sqcup Y$. 
It is straightforward to prove that $(X\sqcup Y,\tau_{X\sqcup Y})$ is a coproduct in $\Topo$ (see for example (3.1.2) of \cite{brownt+g}).
We will use the notation indicated by the following diagram to refer to the map given by the universal property of the coproduct in $\Topo$.
\begin{align}\label{eq:coprod}
\begin{tikzcd}[ampersand replacement=\&]
X\ar[r,"\iota_1"] \ar[dr,bend right,"i"'] \& X\sqcup Y \ar[d,"{\exists ! \copr{i}{j}}", dashed]\& Y\ar[l,"\iota_2"'] \ar[dl,bend left, "j"] \\
\& M \& 
\end{tikzcd}
\end{align}
(If we have maps of spaces $i\colon X\to M$ and $j\colon Y\to N$, we will use $i\sqcup j$ for the obvious map $X\sqcup Y \to M\sqcup N$.)

Let $X,Y,Z$ be topological spaces. Consider continuous maps $f\colon Z \to X$ and $g\colon Z\to Y$. 
The topology on  $X \push{fg}{Z} Y$ which makes it
into a pushout in $\Topo$ is the following:
\[
\tau_{{X\sqcup_{Z}Y}}:=\{U\subseteq X\push{fg}{Z} Y \;\vert\; p_X^{-1}(U) \text{ is closed in $X$ 
and }
p_Y^{-1}(U) \text{ is closed in $Y$}\}.
\]
This topology can be equivalently defined as the finest topology on $X \push{fg}{Z} Y$ making $p_X$ and $p_Y$ continuous.

\subsubsection{Colimits in \texorpdfstring{$\Grpd$}{}}\label{sec:colims_grpd}
The construction of the category $\HomCob$ will rely on the fact that the pushout of two finitely generated groupoids is a finitely generated groupoid, which we prove in Theorem~\ref{th:grpd_pushoutfg}. 
We feel it adds clarity to reproduce the proof (which is present in \cite{higgins}) in our notation, although it is by no means required to understand the paper, thus we relegate the required results, and their proofs to the appendix.

The difficulty in constructing colimits of groupoids stems from the fact that the image of a functor of groupoids is often not a groupoid, as it is not closed.
More precisely, suppose $F \colon \GG \to \mathcal{H}$ is a functor of groupoids, and $g_1\colon w\to x$ and $g_2\colon y\to z$ are morphisms in $\GG$.
Then $g_2*g_1$ is defined if and only if $x=y$ and then we must have $F(g_2*g_1)=F(g_2)*F(g_1)$, meaning $F(g_2)*F(g_1)$ must be the image of a morphism in $\GG$.
Suppose, however, that $x\neq y$ but $F(x)=F(y)$, then $F(g_2)*F(g_1)$ is defined in $\mathcal{H}$ but will not be the image of any single element in $\GG$. 
A consequence is that it is possible for the coequaliser of finite groupoids to be infinite. This is illustrated by the following example.

Let $(\Z,+)$ denote the category with one object and morphisms labelled by elements of $\Z$, with composition given by addition in $\Z$. 

\exa{Let $\{*\}$ be the groupoid with one object and only the identity morphism, and let $\mathbf{I}$ be the groupoid with two objects $\{a,b\}$ and one non-identity morphism from $a$ to $b$.
	Let $\iota_{a}$ be the functor uniquely defined by $\iota_a(*)=a$, and $\iota_{b}$ the functor uniquely defined by $\iota_b(*)=b$.
	The following diagram is a coequaliser 
	\[
	\begin{tikzcd}[ampersand replacement = \&]
	\{*\} \ar[r,"\iota_a",shift left]\ar[r,"\iota_b"',shift right] \& \mathbf{I}  \ar[r,"p"] \& (\Z,+),
	\end{tikzcd}
	\]
	where $p$ is a functor which maps the only non-identity morphism to $1\in \Z$.
	(Note $p$ must send $\{a\}$ and $\{b\}$ to the only object in $(\Z,+)$.)
}

Let $\GG_1 ,\GG_2\in Ob(\Grpd)$ be groupoids. Then we fix a representative coproduct, denoted $\GG_1\sqcup \GG_2$, as follows.
The object set $Ob(\GG_1\sqcup \GG_2)=Ob(\GG_1)\sqcup Ob(\GG_2)$, the coproduct in $\Set$, and the morphism set in $\GG_1\sqcup \GG_2 $ is the coproduct in $\Set$ of the morphisms in $\GG_1$ and the morphisms in $\GG_2$, where for $f\colon w\to x$ in $\GG_1$ and $g\colon y\to z$ in $\GG_2$, $(f,g)$ is a morphism from $ (w,y)$ to $(x,z)$.

\begin{theorem}\label{th:grpd_pushoutfg}
	Let $\GG_0$, $\GG_1$ and $\GG_2$ be finitely generated groupoids and let $f\colon \GG_0 \to \GG_1$ and $g\colon \GG_0 \to \GG_2$ be functors.
	The pushout of $f$ and $g$ is a finitely generated groupoid.
\end{theorem}
\begin{proof}
	By Lemma~\ref{le:pushout_coeq} we can construct the pushout of $f$ and $g$ by finding the coequaliser of $\tilde{f}\colon \GG_0 \to \GG_1 \sqcup \GG_2$ and $\tilde{g}\colon \GG_0 \to \GG_1 \sqcup \GG_2$, where the tilde indicates composition with the maps into the coproduct.
	By Lemmas~\ref{le:construct_universalfunctor} and \ref{le:construct_universalgroupoid}, equivalence classes of morphisms in the coequaliser (given in Lemma~\ref{le:grpd_coeq}) are represented by words in $\GG_1 \sqcup\GG_2$.
	By construction $\GG_1 \sqcup\GG_2$ is finitely generated if $\GG_1$ and $\GG_2$ are, as it is generated by the disjoint union of the generators of $\GG_1$ and $\GG_2$.
	Thus the coequaliser will be finitely generated.
\end{proof}
\subsection{Monoidal categories}
Here we recall the definition of monoidal and symmetric monoidal categories, and of functors preserving this extra structure.
We also give examples that we will make use of later.
A good reference for this section is \cite{turaev}.	 
\label{sec:moncats}
\defn{
	A {\em monoidal category} is a pentuple 
	\[
	(\CC,\otimes, \mathbbm{1},\alpha_{-,-,-},\lambda_{-},\rho_{-})
	\]
	consisting of a category $\CC$ and, 
	\begin{itemize}
		\item a functor $\otimes \colon \CC \times \CC \to \CC$ called the {\it monoidal product};
		\item an object $\mathbbm{1}\in Ob(\CC)$ called the {\it monoidal unit};
		\item for each triple of objects $X,Y,Z\in Ob(\CC)$, an isomorphism
		\[
		\alpha_{X,Y,Z}\colon (X\otimes Y)\otimes Z\to X\otimes (Y\otimes Z)
		\]
		called an {\it associator};
		\item for each $X\in Ob(\CC)$ an isomorphism $\lambda_X\colon \mathbbm{1}\otimes X\to X$ called a {\it left unitor};
		\item for each $X\in Ob(\CC)$ an isomorphism $\rho_X\colon X\otimes \mathbbm{1}\to X$ called a {\it right unitor}.
	\end{itemize}
	These are subject to the following constraints:
	\begin{itemize}
		\item[(M1)] for all $W,X,Y,Z\in Ob(\CC)$ the following diagram
		commutes:
		\[
		\begin{tikzcd}
			\& (W\otimes X)\otimes(Y \otimes Z)\ar[dr,"\alpha_{W,X,Y\otimes Z}"] \& \\
			((W\otimes X)\otimes Y)\otimes Z\ar[d,"\alpha_{W,X,Y}\otimes 1_Z"']\ar[ur,"\alpha_{W\otimes X,Y,Z}"] \& \& W\otimes (X\otimes (Y\otimes Z)) \\
			(W\otimes (X\otimes Y))\otimes Z \ar[rr,"\alpha_{W,X\otimes Y,Z}"']\& \& 
			W\otimes ((X\otimes Y)\otimes Z),
			\ar[u,"1_W\otimes \alpha_{X,Y,Z}"']
		\end{tikzcd}
		\]
		\item[(M2)] for all $X,Y\in Ob(\CC)$ the following diagram 
		commutes:
		\[
		\begin{tikzcd}
			(X\otimes \mathbbm{1})\otimes Y \ar[rr,"\alpha_{X,\mathbbm{1},Y}"]\ar[dr,"\rho_X\otimes 1_Y"'] \& \& X\otimes (\mathbbm{1}\otimes Y)\ar[dl,"1_X\otimes \lambda_Y"] \\
			\& X\otimes Y,
		\end{tikzcd}
		\]
		\item[(M3)]  that is for each morphism $f\colon X\to X'$ in $\CC$, the following diagram commutes:
		\[
		\begin{tikzcd}
			\mathbbm{1}\otimes X\ar[r,"1_{\mathbbm{1}}\otimes f"]\ar[d,"\lambda_{X}"'] \&[+1ex] \mathbbm{1}\otimes X'\ar[d,"\lambda_{X'}"] \\
			X\ar[r,"f"'] \& X',
		\end{tikzcd}
		\]
		\[
		\begin{tikzcd}
			X\otimes \mathbbm{1} \ar[r,"f\otimes 1_{ \mathbbm{1}}"]\ar[d,"\rho_{X}"'] \&[+1ex] X'\otimes \mathbbm{1}\ar[d,"\rho_{X'}"] \\
			X\ar[r,"f"'] \& X',
		\end{tikzcd}
		\]
		\item[(M4)]
		 for all morphisms $f\colon X\to X'$, $g\colon Y\to Y'$ and $h\colon Z\to Z'$ in $\CC$, the following diagram commutes:
		\[
		\begin{tikzcd}
			(X\otimes Y)\otimes Z\ar[r,"(f\otimes g)\otimes h"]\ar[d,"\alpha_{X,Y,Z}"'] \&[+15pt] (X'\otimes Y')\otimes Z'\ar[d,"\alpha_{X',Y',Z'}"] \\
			X\otimes (Y\otimes Z)\ar[r,"f\otimes(g\otimes h)"'] \& X'\otimes (Y'\otimes Z').
		\end{tikzcd}
		\]
	\end{itemize}
}
Conditions (M3) and (M4) are equivalent to saying that all the associators and the left and right unitors are natural isomorphisms.

\defn{An {\em initial object} in a category $\CC$ is an object $I\in Ob(\CC)$ such that for any $X\in Ob(\CC)$, there exists a unique morphism $f\colon I \to X$.}

\exa{
	In $\Topo$ the space with underlying set $\emptyset$ is an initial object.
}

\prop{\label{pr:cocartesian} \cite[Sec.VII.1]{maclane}
	If $\CC$ is a category with all coproducts and an initial object, then $\CC$ has a monoidal structure constructed as follows.
	The monoidal product is the coproduct and monoidal unit the initial object. 
	The associators are obtained by applying the universal property of the coproduct twice.
	It can be shown	these are isomorphisms by constructing inverses in the same way.
	The unitors are obtained by applying the universal property of the coproduct to the pair of $1_X\colon X\to X$ and the unique map $\mathbbm{1}\to X$. By construction, the map into the coproduct $X\to X\otimes \mathbbm{1}$ (or $X\to \mathbbm{1}\otimes X$) composed with the relevant unitor must commute with the identity, thus these are isomorphisms.
	
	This is called the {\em cocartesian} monoidal structure.}

\begin{proof}The triangle and pentagon identities can be proved by noting that objects in the same isomorphism class as a coproduct are also coproducts of the same pair of objects, and the isomorphism connecting them is unique.
	It is straightforward to check the naturality diagrams.
	\end{proof}

The following is an example of cocartesian monoidal structure. It is also straightforward to check each of the identities directly.

\prop{\label{pr:Top_mon}
	(I) There exists a bifunctor
		\ali{
	\sqcup \colon \Topo \times \Topo &\to \Topo \\
	(f\colon W\to X, g\colon Y\to Z)&\mapsto f\sqcup g\colon W\sqcup Y \to X\sqcup Z
}
	where $f\sqcup g(w,y)=(f(w),g(y))$.
	\\
	(II)
	There exists a monoidal category
	\[
	(\Topo,\;\sqcup \;,\; \emptyset \;,\; \alpha_{X,Y,Z}^{T}\colon (X\sqcup Y)\sqcup Z \to X\sqcup (Y\sqcup Z)\;, \; \lambda_X^{T}\colon \emptyset \sqcup X\to X \;,\; \rho_X^{T} \colon X\sqcup \emptyset\to X)
	\]	
	where the associators and unitors are the  obvious isomorphisms. \qed
}

The following is a special case of the monoidal structure on the category of modules over a ring, which is one of the examples in \cite[Sec.VII.1]{maclane}.

\prop{\label{pr:Vectkmon}
(I) There exists a bifunctor
\[
\Vectotimes\colon \Vectk \times \Vectk \to \Vectk 
\]
defined as follows.
Let $V$ and $W$ be vector spaces, then $V\Vectotimes W=V\times X/\sim$ where $\sim$ is the closure to an equivalence of the relations: for all $k\in \mathbbm{k}$, $v,v'\in V$ and $w,w'\in W$
\begin{itemize}
	\item $(kv,w)\sim k(v,w) \sim (w,kw)$,
	\item $(v+v',w)\sim (v,w)+(v',w)$,
	\item $(v,w+w')\sim (v,w)+(v,w')$.
\end{itemize} 
Given any $v\in V$ and $w\in W$, we use $v\Vectotimes w$ to denote the equivalence class $[(v,w)]$.
For linear maps $S\colon V\to X$ and $T\colon W\to Y$ we define
\ali{
S\Vectotimes T\colon V\Vectotimes W &\to X\Vectotimes Y \\
(v\Vectotimes w) &\mapsto S(v)\Vectotimes T(w).
}
(II) There exists a monoidal category 
\[
(\Vectk,\Vectotimes,\mathbbm{k}, \alpha_{V,W,X}^\mathbbm{k},\lambda_{V}^\mathbbm{k},\rho_{V}^\mathbbm{k}).
\]
where for all $v\in V$, $w\in W$, $x\in X$ and $k\in \mathbbm{k}$, $\alpha_{V,W,X}^\mathbbm{k}((v\Vectotimes w)\Vectotimes x)= (v\Vectotimes (w\Vectotimes x))$, $\lambda_{V}^\mathbbm{k}(v\otimes k)=kv$ and $\rho_V^\mathbbm{k}(k\otimes v)=kv$. \qed
}

\rem{Using the relations in $V\Vectotimes W$, it is not hard to show that a basis for $V\Vectotimes W$ is given by elements of the form $v\Vectotimes w$ where $v\in V$ and $w\in W$ are basis elements. }

\defn{Let $(\CC,\otimes, \mathbbm{1},\alpha_{-,-,-},\lambda_{-},\rho_{-})$ and $(\DD,\otimes', \mathbbm{1}',\alpha'_{-,-,-},\lambda'_{-},\rho'_{-})$ be monoidal categories.
	 A {\em monoidal functor} is a functor $F\colon \CC\to \DD$ endowed with a morphism $F_0\colon \mathbbm{1'}\to F(\mathbbm{1})$ in $\DD$ and with a natural transformation
 \[
 F_2=\left\{F_2(X,Y)\colon F(X)\otimes' F(Y) \to F(X\otimes Y)\right\}_{X,Y\in Ob(\CC)}
 \]
 between the functors $F\otimes' F=\otimes' \circ (F\times F)\colon \CC\times \CC \to \DD $ and $F\circ \otimes \colon \CC \times \CC \to \DD$ such that for all $X,Y,Z\in Ob(\CC)$ the following three diagrams commute.
 \[
 \begin{tikzcd}[column sep=5.5em]
 (F(X)\otimes'F(Y))\otimes' F(Z) \ar[d,"{F_2(X,Y)\otimes'1_{F(Z)}}"'] \ar[r,"{\alpha'_{F(X),F(Y),F(Z)}}"] 
 \&
 F(X)\otimes'(F(Y)\otimes'F(Z)) \ar[d,"{1_{F(X)}\otimes'F_2(Y,Z)}"] 
 \\
 F(X\otimes Y)\otimes'F(Z) 
 \ar[d,"{F_2(X\otimes Y ,Z)}"']
 \& F(X) \otimes' F(Y\otimes Z) 
 \ar[d,"{F_2(X,Y \otimes Z)}"] 
 \\
 F((X\otimes Y)\otimes Z) \ar[r,"{F(\alpha_{X,Y,Z})}"'] 
 \& F(X\otimes (Y\otimes Z))
 \end{tikzcd}
 \]
 \[
  \begin{tikzcd}
  \mathbbm{1}'\otimes' F(X)\ar[r,"{\lambda'_{F(X)}}"]\ar[d,"{F_0\otimes' 1_{F(X)}}"'] \& F(X) \\
  F(\mathbbm{1})\otimes'F(X)\ar[r,"{F_2(\mathbbm{1},X)}"']\& F(\mathbbm{1}\otimes X)\ar[u,"{F(\lambda_X)}"']
  \end{tikzcd}
 \]
 \[
 \begin{tikzcd}
 F(X)\otimes' \mathbbm{1} \ar[d,"{1_{F(X)}\otimes' F_0}"']\ar[r,"{\rho_{F_{(X)}}'}"]\& F(X)\\
 F(X)\otimes F(\mathbbm{1})\ar[r,"{F_2(X,\mathbbm{1})}"'] \& F(X\otimes \mathbbm{1})\ar[u,"{F(\rho_{X})}"']
 \end{tikzcd}
 \]
}
\defn{A {\em strong monoidal functor} is a monoidal functor $F$ where $F_0$ and all maps in $F_2$ are isomorphisms.}

It is straightforward to prove that functor composition extends to a composition of monoidal functors.
\prop{There is an associative composition of monoidal functors which sends a pair $F\colon \CC \to \DD$ and $	G\colon \DD \to \mathcal{E}$ to the monoidal functor $G\circ F\colon \CC \to \mathcal{E}$ with
\ali{
(G\circ F)_0=G(F_0)G_0 \;\;\;\text{ and } \;\;\;  (G\circ F)_2(X,Y)=G(F_2(X,Y))\circ G_2(F(X),F(Y))
}
for all $X,Y\in Ob(\CC)$. \qed
}

\defn{A {\em braided monoidal category} is a six-tuple 
	\[
	(\CC ,\otimes, \mathbbm{1},\alpha_{-,-,-},\lambda_{-},\rho_{-},\beta_{-,-})
	\]
consisting of a monoidal category $(\CC ,\otimes, \mathbbm{1},\alpha_{-,-,-},\lambda_{-},\rho_{-})$ and, 
for each pair $X,Y\in Ob(\CC)$, a family of natural isomorphisms 
\[
\beta_{X,Y}\colon X\otimes Y\to Y \otimes X
\]
 such that
\ali{
\beta_{X,Y\otimes Z}& =(1_Y\otimes \beta_{X,Z})*_\CC (\beta_{X,Y}\otimes 1_{Z}) \\
\beta_{X\otimes Y,Z}&=(\beta_{X,Z}\otimes 1_Y)*_\CC (1_X\otimes \beta_{Y,Z}) .
}
Naturality means that for any morphisms $f\colon X\to X'$ and $g\colon Y\to Y'$ the following diagram commutes.
\[
\begin{tikzcd}
X\otimes Y \ar[r,"f\otimes g"] \ar[d,"\beta_{X,Y}"'] \& X'\otimes Y'\ar[d,"\beta_{X',Y'}"] \\
Y\otimes X \ar[r,"g\otimes f"'] \& Y'\otimes X'
\end{tikzcd}
\]
Such a family of natural isomorphisms is called a \textit{braiding} on $(\CC ,\otimes, \mathbbm{1},\alpha_{-,-,-},\lambda_{-},\rho_{-})$.\\
A braiding $\beta$ on a monoidal category $(\CC,\otimes, \mathbbm{1},\alpha_{-,-,-},\lambda_{-},\rho_{-})$ is called {\em symmetric} if, for all pairs $X,Y \in Ob(\CC)$, we have
\[
\beta_{Y,X}*_\CC\beta_{X,Y}= 1_{X\otimes Y}\colon X\otimes Y \to X\otimes Y.
\] 
A {\em symmetric monoidal category} is a braided monoidal category $(\CC ,\otimes, \mathbbm{1},\alpha_{-,-,-},\lambda_{-},\rho_{-},\beta_{-,-})$ such that $\beta$ is symmetric.
}

When speaking about (braided) monoidal categories we may drop entries of the tuple corresponding to the natural isomorphisms, or even refer to a braided monoidal category as just $\CC$ where $\CC$ is the notation of the underlying category.
We note however, that there will often be many (braided) monoidal categories with the same underlying category, monoidal product and monoidal unit.

\prop{\label{pr:Top_braid}\label{pr:Vect_braid}
	It is straightforward to check that the monoidal structures on $\Topo$ and $\Vect_{\mathbbm{k}}$ given in Propositions~\ref{pr:Top_mon} and \ref{pr:Vectkmon} extend to braided monoidal categories with the braidings given by the morphisms $\beta_{X,Y}\colon X\sqcup Y\to Y\sqcup X$, $(x,y)\mapsto (y,x)$ in $\Topo$ and $\beta_{V,W}\colon V\otimes_\mathbbm{k} W\to W\otimes_\mathbbm{k} V$, $v\otimes_\mathbbm{k} w\mapsto w\otimes_\mathbbm{k} v$ in $\Vect_{\mathbbm{k}}$.\qed
}

\defn{A {\em monoidal subcategory} of a monoidal category $(\CC,\otimes, \mathbbm{1},\alpha_{-,-,-},\lambda_{-},\rho_{-})$ is a pentuple $(\DD,\otimes, \mathbbm{1},\alpha_{-,-,-},\lambda_{-},\rho_{-})$ such that 
	\begin{itemize}
		\item $\DD$ is a subcategory of $\CC$,
		\item $\otimes$ restricts to a closed composition on $\DD$,
		\item  $\mathbbm{1}\in \DD$, and
		\item for all $X,Y,Z\in Ob(\DD)$ we have $\alpha_{X,Y,Z},\lambda_{X},\rho_{X}$ are in $\DD$.
\end{itemize}
	\sloppypar{A {\em braided monoidal subcategory} of a braided monoidal category is defined analogously with the extra condition that 
 for all $X,Y\in Ob(\DD)$, $\beta_{X,Y}\in \DD$.}}

\defn{A {\em braided monoidal functor} between braided categories $(\CC,\beta_{-,-})$ and $(\CC',\beta'_{-,-})$ is a monoidal functor $F\colon \CC\to \CC'$ such that for all $X,Y\in Ob(\CC)$,
	\[
	F_2(Y,X)*_{\CC'}\beta'_{F(X),F(Y)}=F(\beta_{X,Y})*_{\CC'} F_2(X,Y).
	\] 
	A {\em symmetric monoidal functor} is a braided monoidal functor between symmetric monoidal categories. }
\subsection{Cofibrations in \texorpdfstring{$\Topo$}{Topo} 
}
\label{sec:cofibs+vk}

In Section~\ref{sec:homcobs} we define a magmoid whose morphisms are {\it cofibrant cospans}, and quotient by a congruence defined in terms of \textit{cofibre homotopy equivalence}.
 Then, in Section~\ref{sec:tqft}, our TQFT construction  will rely on a version of the van Kampen Theorem for cofibrations. 
Here we recall the results we will need. 
 More detail on cofibrations can be found in \cite[Ch.~5]{dieck} or \cite[Ch.~7]{brownt+g}, definitions and results on cofibre homotopy equivalence follow \cite[Ch.~6]{may}, and the version of the van Kampen theorem reproduced here can be found in \cite[Thm.~9.1.2]{brownt+g}.

\subsubsection{Cofibrations}

A cofibration can be thought of as a homotopically well behaved embedding. 
Specifically, an embedding $i\colon A \to X$ is a cofibration if we have both that
there is an open neighbourhood of the image which strongly deformation retracts onto $i(A)$, and that this retraction can be extended to a homotopy on the whole of $X$.
In other words there is an open neighbourhood of the image which, up to a homotopy of $X$, is equivalent to the image. This characterisation of a cofibration is not immediately obvious from the below definition but we will see it is equivalent in Theorem~\ref{th:hom_for_closed_cofibred_pair}.

One can think of the aforementioned characterisation of a cofibration as a version of the Collar Neighbourhood Theorem  of a boundary of a manifold (\cite{brown62}).
To construct cobordism categories (see \cite{milnor}) the collar neighbourhood is required to prove the identity axiom. 
The cofibration condition we impose will play a similar role in our category construction, we will see this in Theorem~\ref{th:CofCsp}.

\medskip
The following definition is from \cite[Sec.~5.1]{dieck}.

\begin{definition}\label{de:homotopy_extension}
	Let $A$ and $X$ be spaces.
	A map $i\colon A \to X$ has the \emph{homotopy extension property}, with respect to the space $Y$, if for any pair of a homotopy $h\colon A \times \II \to Y$ and a map $f\colon X \to Y$ satisfying $(f\circ i)(a)=h(a,0)$, there exists a homotopy $H:X \times \II \to Y$, extending $h$, with $H(x,0)=f(x)$ and $H(i(a),t)=h(a,t)$.
	This is illustrated by the following diagram.
	\[
	\begin{tikzcd}[cramped,column sep=small, ampersand replacement=\&] 
		\& X\ar[dr,"\iota_0^X"']\ar[drr, bend left=20,"f" ] \&[-0.7em]  \&[1.2em] \\
		A \ar[ur,"i"] \ar[dr,"\iota_0^A"'] \& \&  X\times \II\ar[r,dashrightarrow,"\exists H"] \&  Y \\
		\& A\times \II \ar[ur,"i\times \id_\II"]\ar[urr,bend right=20,"h"']\& \&
	\end{tikzcd}
	\]
	(Where for any space $X$, $\iota^X_0\colon X \to X \times \II$ is the map $x\mapsto (x,0)$.)
\end{definition}

\begin{definition}Let $A$ and $X$ be spaces.
We say that $i\colon A \to X$ is a \emph{cofibration} if $i$ satisfies the homotopy extension property for all spaces $Y$. A \emph{closed cofibration} is a cofibration with image a closed set. 
If $A\subseteq X$ is a subspace and the inclusion $\iota \colon A\to X$ is a cofibration, we say $(X,A)$ is a \emph{cofibred pair}.
\end{definition}

The following useful well known results are completely straightforward to prove. 

 \begin{lemma}\label{le:cofib_comp}
The composition of two cofibrations is a cofibration. \qed
\end{lemma}

\lemm{\label{le:homeo_cofib}
Every homeomorphism is a cofibration. \qed
}

\begin{lemma} \label{cofib_copr}
	Let $X$ and $Y$ be topological spaces. The map $\iota_1 \colon X \to X\sqcup Y$, $x\mapsto (x,1)$ is a cofibration. \qed
\end{lemma}

\begin{proposition}\label{pr:cofibdpair}
	Let $A$ and $X$ be spaces and $i\colon A \to X$ a map which is a homeomorphism onto its image. Then $i$ is a cofibration if and only if $(X,i(A))$ is a cofibred pair.
\end{proposition}
\begin{proof} 
	Suppose $i\colon A \to X$ is a cofibration and $K$ is any space.
	Consider a map $f\colon X\to K$ and a homotopy $h\colon i(A)\times \II\to K$ such that for all $x\in i(A)$,
	$h(x,0)=f(x)$.
	Applying the homotopy extension property to $f$ and $h\circ (i\times \id)$ gives a homotopy $H\colon X\times \II \to K$. 
	This same $H$ is also a homotopy extending $h$, and hence $(X,i(A))$ is a cofibred pair.

Suppose $i\colon A\to X$ a map which is a homeomorphism onto its image and  $(X,i(A))$ is a cofibred pair.
From Lemma~\ref{le:homeo_cofib} we have that the homeomorphism $ A\to i(A)$, $a\mapsto i(a)$ is a cofibration and 
the inclusion $\iota\colon i(A)\to X$ is a cofibration by assumption.
Hence $i\colon A\to X$ is a composition of cofibrations, and so a cofibration by Lemma~\ref{le:cofib_comp}.
\end{proof}	

\begin{theorem}\label{th:homeo_onto_im}
	(See for example \cite[Th.~1]{strom1})
	Let $A$ and $X$ be spaces. If $i\colon A \to X$ is a cofibration then $i$ is a homeomorphism onto $i(A)$ with the subspace topology (i.e. $i$ is an embedding).
	\qed
\end{theorem}
To prove specific pairs are cofibred we have the following two classical results.
 \begin{proposition} \label{retract}
(See for example \cite[Prop.~5.1.2]{dieck}) Let $A$ be a closed subspace of $X$. The pair $(X,A)$ is cofibred if and only if $X\times \{0\}\, \cup \, A \times \II$ is a retract of $X \times \II$.\\
(Recall $N\subset M$ is a retract of $M$ if there is a continuous map $r\colon M\to N$ such that $r(n)=(n)$ for all $n\in N$.) \qed
 \end{proposition}
 The following lemma characterises cofibrations as inclusions such that there is a neighbourhood of the image that deformation retracts onto, and hence is homotopy equivalent to the image.
 This is reminiscent of the Collar Neighbourhood Theorem for manifolds.

\begin{theorem}\label{th:hom_for_closed_cofibred_pair} \cite[Th.~2]{strom1} 
	Let $A$ be a closed subspace of space $X$. Then $(X,A)$ is a cofibred pair if and only if there exists 
	\begin{enumerate}[label={\textit{\roman*})}]
		\item a neighbourhood 
		 $U\subseteq X$ of $A$ (not necessarily open) and a homotopy $H \colon U\times \II \to X$  such that for all $t\in \II$, $x\in U$ and $a\in A$, we have $H(x,0)=x$, $H(a,t)=a$ and $H(x,1)\in A$, and
		\item a map $\phi\colon X \to \II $ such that $A=\phi^{-1}(0)$ and $\phi(x)=1$ for all $x\in X-U$. \qed
	\end{enumerate}
\end{theorem}
There is also more general version of the previous theorem (see \cite[Lem.~4]{strom2}) which characterises cofibred pairs in a similar way without restricting to closed subspaces. 
 	However we will only need the case of closed subspaces and this one will be easier to work with.

\subsubsection*{Examples of cofibrations}
These key examples will be useful to demonstrate the flexibility of our construction in Section~\ref{sec:homcobs}.
\exa{
For any space $X$, the pairs $(X,X)$ and $(X,\emptyset)$ are cofibred. 
}
\begin{example}\label{ex:I,{0,1}}
The pair $(\II,\{0,1\})$ is cofibred. \\
Consider $\II\times \II$ as a subset of $\mathbb{R}^2$ and let $z=(\frac{1}{2},\frac{3}{2})\in \mathbb{R}^2$.
For any $x\in \II\times \II$, let $x'$ be the unique point of $\II \times 0 \, \cup\, \{0,1\} \times \II$ such that $z,x,x'$ are colinear.
Then $\rho \colon x \mapsto x'$ is a retraction 
$\II\times \II$ to $\II \times 0 \,\cup\, \{0,1\} \times \II$.
\end{example}

Define topological spaces $S^n=\{x\in \R^{n+1} \;\vert \; |x|=1\}$ and $D^n=\{x\in \R^n \;\vert \; |x|\leq 1\}$, each topologised with the subset topology.
\begin{example}The pair $(D^{n},S^{n-1})$ is cofibred. \\
A retraction $r\colon D^{n} \to S^{n-1} \times \II \; \cup \; D^{n} \times 0$ can be constructed in 
a similar way to the previous example, see
\cite[Ex.~2.3.5]{dieck}.
\end{example}
The previous two examples are special cases of the following proposition for manifolds.
\begin{proposition}\label{pr:mfld_bdy}
 Let $M$ be a smooth compact manifold with boundary. The inclusion $i\colon \partial M \to M$ is a cofibration.
\end{proposition}
\begin{proof}
	The Collar Neighbourhood Theorem for smooth manifolds 
	\cite[Cor.~3.5]{milnor} 
	says that there exists an open neighbourhood $N'$ of $\partial M$ such that there exists a diffeomorphism $f\colon N' \to \partial M \times [0,1)$ satisfying $f(\partial M)=\partial M \times \{0\}$.
	It follows that we can choose a closed neighbourhood $N$ of $\partial M$ such that the we can identify $N$ with $\partial M \times \II$, where $\partial M \cong \partial M \times \{0\}$. 
	Then the function
	\ali{
		H\colon (\partial M\times [0,1])\times \II &\to M\\
		((n,s),t)&\mapsto 
		(n,s(1-t))}
	is a homotopy 
	satisfying condition (i) of Theorem~\ref{th:hom_for_closed_cofibred_pair}.
	Define a map $\phi \colon M \to \II$ as follows.
	\[
	\phi(m) =
	\begin{cases}
	s& 
	\text{if $m=(n,s) \in \partial M \times [0,1]$} \\
	1& 
	\text{if $m\in M\setminus (\partial M \times [0,1))$}
	\end{cases}
	\]
	Notice that these definitions agree on the overlap so $\phi$ is continuous.
	Hence by Theorem~\ref{th:hom_for_closed_cofibred_pair}  the inclusion $i\colon \partial M\to M$ is a cofibration.
\end{proof}

Recall that a smooth submanifold of $M$ is a  subset $N\subseteq M$ such that the identity inclusion $\iota\colon N \to M$ is a diffeomorphism onto its image and a topological embedding, as in \cite[Ch.~5]{Lee}.
A submanifold $N\subseteq M$ is {\em neatly embedded} if $\partial N\subseteq \partial M$.

We have the following stronger proposition. This is a slight generalisation of \cite[Ex.2.3.11]{Munson}.
\begin{proposition}\label{pr:sbmfld_cofib}
Let $N$ be a closed smooth submanifold of a smooth manifold $M$ which is neatly embedded. 
Then the inclusion $\iota\colon N \to M$ is a cofibration.
\end{proposition}
\begin{proof}
	By the tubular neighbourhood theorem \cite[Ch.4,Th.6.3]{hirsch}, there exists 
an open neighbourhood $U\subset M$ of $N$  which can be identified with the normal bundle of $\mathrm{v}(N,M)$ of $N$ in $M$, such that $N$ is identified with the zero section of the normal bundle.
This allows us to identify points in $u\in U$ with pairs $(n,v)$ where $n\in N$ and $v$ is a vector in the tangent space.
Further, a choice of Riemannian metric on $M$ allows us to define $V=\mathrm{D}_1(\mathrm{v}(N,M))$, which has as underlying set pairs $(n,v)$ such that $||v||\leq 1$. Label the corresponding subset in $M$ by $\bar{V}$.
Similarly let $V$ be the open set obtained from taking the set with pairs $(n,v)$ where $||v||<1$. 

Now, as in the previous proposition, we can define  a homotopy 
\ali{
H \colon \bar{V}\times \II &\to M \\
((n,v),t)&\mapsto
(n,(1-t)v) 
}
and a map $\phi\colon M \to \II$ with
\begin{align*}
\phi(m) =
\begin{cases}
 {\vert \vert v \vert \vert} & \text{ if }m=(n,v)\in \bar{V} \\
 1 & \text{ if }m\in M\setminus V. 
\end{cases}
& \qedhere
\end{align*}
\end{proof}

In many cases 
it will be easier to find a CW complex structure than to prove we have a smooth manifold submanifold pair.
In this case we have the following proposition.
\begin{proposition}\label{pr:CWcofib}
Let $X$ be a CW complex and let $A$ be a subcomplex of $X$. Then the inclusion $i\colon A \to X$ is a closed cofibration.
\end{proposition}
\begin{proof}
If $(X,A)$ is a CW pair, then $X\times \{0\}\cup A\times \II$ is a deformation retract of $X\times \II$ \cite[Prop.0.16]{hatcher}.
\end{proof}

\subsubsection*{Cofibre homotopy equivalence}
To construct a category of {\it cofibrant cospans}, we will will require a notion of homotopy equivalence of spaces relative to maps from a shared space.

\begin{definition}
	A {\em space under $A$} is a map $i \colon A \to X$.
	A {\em map of spaces under $A$} from $i \colon A \to X$ to $j\colon A\to Y$ is a map $f\colon X\to Y$ such that we have a commuting diagram 
	\[
	\begin{tikzcd}[ampersand replacement = \&]
		\& A \ar[dl,"i"'] \ar[dr,"j"] \& \\
		X \ar[rr,"f"']\& \& Y.
	\end{tikzcd}
	\]
	Suppose $f,f'\colon X\to Y$ are maps under $A$ from $i \colon A \to X$ to $j\colon A\to Y$.
	A {\em homotopy under $A$} from $f$ to $f'$ is a continuous map $H\colon X\times \II \to Y$ such that
	\begin{itemize}
		\item for all $a \in A$ and $t \in \II$, $H(i(a),t)=j(a)$,
		\item for all $x\in X$,  $H(x,0)=f(x)$,
		\item for all $x\in X$,  $H(x,1)=f'(x)$.
	\end{itemize}
\end{definition}

The proof of the following Proposition proceeds exactly as for the usual notion of homotopy equivalence.
\prop{
	\label{pr:cofib_hom_equiv} Define a relation on spaces under $A$ as follows. We have $(i\colon A\to X)\sim (j\colon A\to Y)$ if there exists maps of spaces under $A$, $f\colon X\to Y$ from $i\colon A\to X$ to $j\colon A\to Y$,  and $f'\colon Y\to X$ from $j\colon A\to Y$ to $i\colon A\to X$, such that there exists a homotopy under $A$ from $f\circ f'$ to $\id_Y$ and from $f'\circ f$ to $\id_X$.
	Note that this is well defined since $f\circ f'\circ j(a)=f\circ i(a)=j(a)$, $f'\circ f\circ i(a)=f'\circ j(a) = i(a)$.\\
	This is an equivalence relation.}

\defn{Given a space $A$, the equivalence relation described in Proposition~\ref{pr:cofib_hom_equiv} is called {\em cofibre homotopy equivalence}.}

The following technical result, which justifies the choice of name, will be crucial for our construction. Proofs are present in \cite{may,brownt+g,strom}.

\begin{theorem}\label{th:cofib_equiv}
	(\cite[Thm.~7.2.8]{brownt+g} for example.)
	Let $A$, $X$ and $Y$ be spaces.
	Let $i\colon A \to X$ and $j\colon A \to Y$ be cofibrations and let $f\colon X \to Y$ be a map such that $f\circ i=j$.
	Suppose that $f$ is a homotopy equivalence from $X$ to $Y$, then $f$ is a cofibre homotopy equivalence
	$i\colon A \to X$ to $j\colon A \to Y$. \qed
\end{theorem}

\subsubsection{A van Kampen Theorem for cofibrations}
Here we give a generalisation of the van Kampen Theorem, due to Brown \cite{brownt+g}, which says
that pushouts are preserved by the functor $\pi$ if at least one of the maps we take the pushout over is a cofibration.
Hence, in this case, we can obtain the fundamental groupoid of a pushout in $\Topo$ as a pushout of fundamental groupoids.

Suppose we have spaces $X_0$, $X_1$ and $X_2$ and maps 
$f\colon X_0 \to X_1$ and $g\colon X_0 \to X_2$.
Consider the pushout square: 
\begin{align}\label{po_vk}
	\begin{tikzcd}[ampersand replacement=\&]
		X_0 \ar[r,"f"]\ar[d,"g"'] \& X_1\ar[d,"p_1"] \\
		X_2\ar[r,"p_2"'] \& X_1\sqcup_{X_0} X_2
	\end{tikzcd}
\end{align}
Now let $A$, $B$ and $C$ be representative subsets of $X_0$, $X_1$ and $X_2$ respectively, 
with $f(A)= B\cap f(X_0)$ and $g(A) \subseteq C$. 
Let $D=B \sqcup_A C$, the pushout of $f|_A\colon A\to B$ and $g|_A\colon A\to C$.

\begin{lemma} 
	Under the above conditions conditions $D$ is representative (Def~\ref{de:representative}) in $ X_1\sqcup_{X_0} X_2$.
\end{lemma}
\begin{proof}
	It is clear that $B\sqcup C$ is representative in $X_1\sqcup X_2$.
	By Lemma~\ref{le:surjection_of_representative_set} surjections send representative subsets to representative subsets, hence we have the result by considering the surjection $\lan p_1,p_2\ran \colon X_1\sqcup X_2\to X_1\sqcup_{X_0}X_2$.
\end{proof}
\begin{theorem}\label{th:vk}
	(See \cite[Thm.~9.1.2]{brownt+g}.)
	Now suppose in addition to the above conditions we take $X_0\subseteq X_1$ and
	$f=\iota\colon X_0 \to X_1$ the inclusion map in \eqref{po_vk}.
	Then the following diagram is a pushout if $(X_1,X_0)$ is a cofibred pair.
	\[
	\begin{tikzcd}[ampersand replacement=\&]
	\pi(X_0,X_0\cap B) \ar[r,"\pi(\iota
	)"]\ar[d,"\pi(g)"'] \& \pi(X_1,B)\ar[d,"\pi(p_1)"] \\
	\pi(X_2,C)\ar[r,"\pi(p_2)"'] \& \pi(X_1\sqcup_{X_0} X_2,D)
	\end{tikzcd}\vspace*{-1.5\baselineskip}
	\]
	\qed
\end{theorem}
\begin{corollary}\label{co:vk}
	Now suppose in addition to the above conditions, we instead consider $f=i\colon X_0 \to X_1$ any cofibration in \eqref{po_vk}.
	Then the following square is a pushout.
	\[
	\begin{tikzcd}[ampersand replacement=\&]
	\pi(X_0,A) \ar[r,"\pi(i)"]\ar[d,"\pi(g)"'] \& \pi(X_1,B)\ar[d,"\pi(p_1)"] \\
	\pi(X_2,C)\ar[r,"\pi(p_2)"'] \& \pi(X_1\sqcup_{X_0} X_2,D)
	\end{tikzcd}
	\]
	\end{corollary}
\begin{proof}
	Using Proposition \ref{pr:cofibdpair} and Theorem \ref{th:homeo_onto_im} we can separate the cofibration $i$ into two maps, a homeomorphism $\tilde{i}\colon X_0 \to i(X_0)$
	and a cofibred inclusion $\iota \colon i(X_0)\to X_1$.
	Consider the following commuting pushouts.
	\[ 
	\begin{tikzcd}[ampersand replacement=\&]
		i(X_0)\ar[rr,"\iota"]\ar[dd,"g\circ\tilde{i}^{-1}"']\&[-25pt] \& X_1\ar[dd,"p_1"] \\[-1.5ex]
		\& X_0\ar[dl,"g"]\ar[ur,"i"']\ar[ul,"\tilde{i}"] \& \\
		X_2\ar[rr,"p_2"'] \& \& X_1 \sqcup_{X_0} X_2
	\end{tikzcd}
	\]	
	Observe that the pushout of $g$ and $i$ is also the pushout of $g\circ \tilde{i}^{-1}$ and $\iota$, since $\tilde{i}$ is a homeomorphism. 
	Choosing the subset $i(A)$ in $i(X_0)$ and keeping all other subsets as in the statement of the corollary, the outer pushout is preserved by the fundamental groupoid functor by Theorem \ref{th:vk}.
	Hence, by functoriality, so is the inner pushout.
\end{proof}

\section{Homotopy cobordisms}\label{sec:homcobs}

In this section the main result is Theorem~\ref{th:HomCob} which says that we have a symmetric monoidal category, $\HomCob$, of \textit{homotopy cobordisms}.

We proceed by first constructing a magmoid whose morphisms are  \textit{concrete cofibrant cospans}, which compose via pushouts. 
We then quotient by a congruence, defined in terms of a homotopy equivalence, to obtain the category $\CofCsp$ (Theorem~\ref{th:CofCsp}) and show that there exists a symmetric monoidal structure on $\CofCsp$ (Theorem~\ref{th:CofCsp_symm}).
We add a finiteness condition to the topological spaces in the cospans to arrive at the category $\HomCob$ as a symmetric monoidal subcategory of $\CofCsp$ (Theorems~\ref{HomCob} and \ref{th:HomCob}).

We note that the choices we make in this section are made with the type of functor into $\VectC$ that we will construct already in mind. 
We have already said that we are interested in functors which depend on the homotopy of the spaces, and Corollary~\ref{co:vk} tells us that cofibrations are the maps that behave well with respect to the fundamental groupoid.
Further the congruence is defined in terms of a suitable version of homotopy equivalence --
the aim is that the category $\HomCob$ contains interesting subgroupoids which are finitely generated, and thus manageable structures to work with, but we don't want to make any morphisms equivalent that might have been mapped to different linear maps in $\Vect$.

We note that our cospan categories deviate from those of e.g. \cite{fong}, in our choice of identity.
For Fong, the category identity at an object $X$ is the equivalence class of the cospan $\cmor{X}{\id_X}{X}{\id_X}{X}$.
In a topological quantum field theory we require that any arbitrary time evolution of a state $X$ is evaluated as the identity if the state $X$ does not change. Hence we insist identities are the equivalence classes of cospans of the form $\cmor{X\times \II}{\iota_0^X}{X}{\iota_1^X}{X}$ (where $\iota^X_a\colon X \to X \times \II$ is the map $x\mapsto (x,a)$).
As a result more work is required to prove that this is, in fact, an identity; see Theorem~\ref{th:CofCsp}.

\subsection{Magmoid of concrete cofibrant cospans \texorpdfstring{$\cCofCsp$}{CofCsp}}
Here we define \textit{concrete cofibrant cospans}, construct a composition and organise them into a magmoid. 

\begin{definition}
Let $X$, $Y$ and $M$ be spaces.
  A {\em \ccc{}} from $X$ to $Y$ 
  is a diagram $\cmor{M}{i}{X}{j}{Y}$ such that $\langle i, j\rangle \colon X \sqcup Y \to M$ is a closed cofibration. (The map $\langle i,j\rangle$ is obtained via the universal property of the coproduct, see Diagram \eqref{eq:coprod}.) \\
  For spaces $X,Y\in \Topo$, we define the set of all concrete cofibrant cospans from $X$ to $Y$
  \[
  \cCofCsp(X,Y)=\left\{ \cmortikz{M}{i}{X}{j}{Y}\; \middle\vert \; \lan i , j \ran \text{ is a closed cofibration} \right\}.
  \]
\end{definition}

\rem{The previous definition forces the images of $i$ and $j$ to be disjoint since a cofibration is a homeomorphism onto its image (Theorem~\ref{th:homeo_onto_im}).}

\begin{figure}
	\centering
	\def\svgwidth{0.5\columnwidth}
\begingroup%
  \makeatletter%
  \providecommand\color[2][]{%
    \errmessage{(Inkscape) Color is used for the text in Inkscape, but the package 'color.sty' is not loaded}%
    \renewcommand\color[2][]{}%
  }%
  \providecommand\transparent[1]{%
    \errmessage{(Inkscape) Transparency is used (non-zero) for the text in Inkscape, but the package 'transparent.sty' is not loaded}%
    \renewcommand\transparent[1]{}%
  }%
  \providecommand\rotatebox[2]{#2}%
  \newcommand*\fsize{\dimexpr\f@size pt\relax}%
  \newcommand*\lineheight[1]{\fontsize{\fsize}{#1\fsize}\selectfont}%
  \ifx\svgwidth\undefined%
    \setlength{\unitlength}{488.21681074bp}%
    \ifx\svgscale\undefined%
      \relax%
    \else%
      \setlength{\unitlength}{\unitlength * \real{\svgscale}}%
    \fi%
  \else%
    \setlength{\unitlength}{\svgwidth}%
  \fi%
  \global\let\svgwidth\undefined%
  \global\let\svgscale\undefined%
  \makeatother%
  \begin{picture}(1,0.40483431)%
    \lineheight{1}%
    \setlength\tabcolsep{0pt}%
    \put(0,0){\includegraphics[width=\unitlength,page=1]{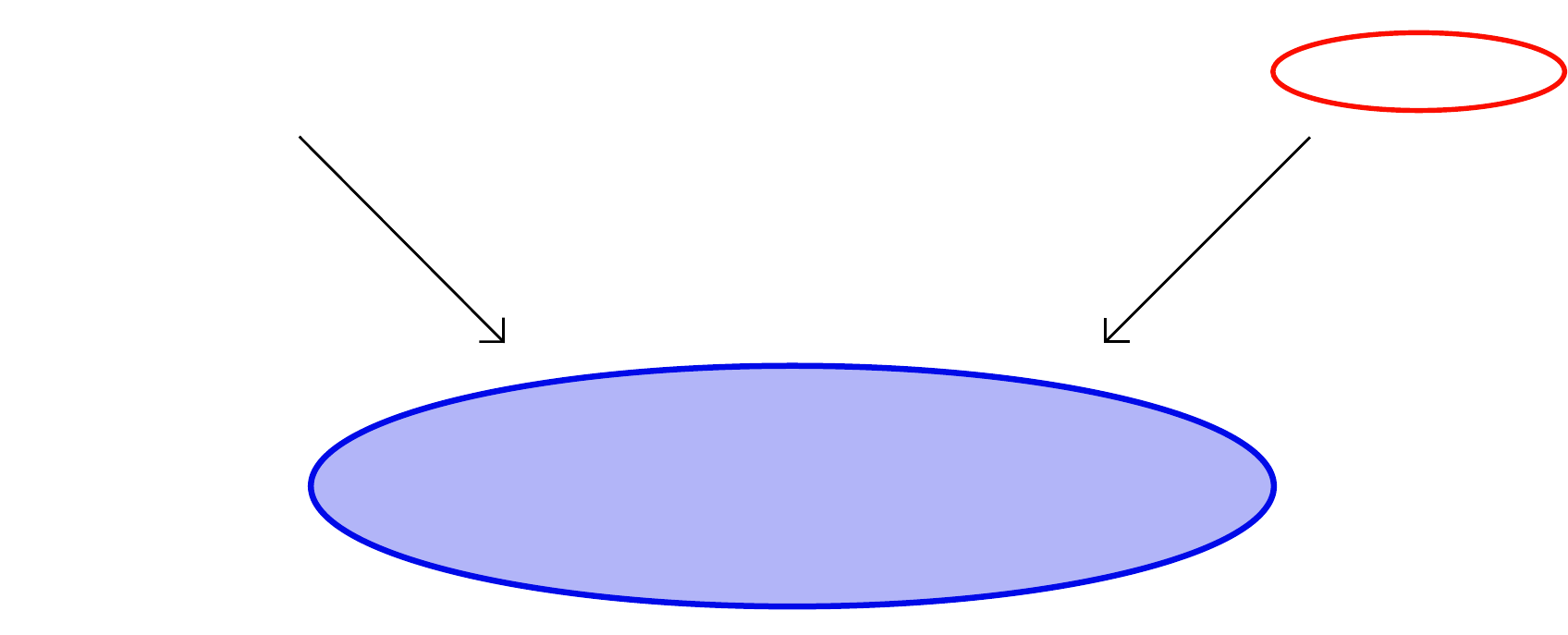}}%
    \put(0.93570791,0.29099505){\makebox(0,0)[lt]{\lineheight{1.25}\smash{\begin{tabular}[t]{l}$S^1$\end{tabular}}}}%
    \put(0.7206395,0.00240832){\makebox(0,0)[lt]{\lineheight{1.25}\smash{\begin{tabular}[t]{l}$D^2$\end{tabular}}}}%
    \put(0.78208758,0.21747674){\makebox(0,0)[lt]{\lineheight{1.25}\smash{\begin{tabular}[t]{l}$j$\end{tabular}}}}%
    \put(0.18296851,0.23283882){\makebox(0,0)[lt]{\lineheight{1.25}\smash{\begin{tabular}[t]{l}$i$\end{tabular}}}}%
    \put(0,0){\includegraphics[width=\unitlength,page=2]{ccc_S1,D2.pdf}}%
    \put(0.93570791,0.29099505){\makebox(0,0)[lt]{\lineheight{1.25}\smash{\begin{tabular}[t]{l}$S^1$\end{tabular}}}}%
    \put(0.93570791,0.29099505){\makebox(0,0)[lt]{\lineheight{1.25}\smash{\begin{tabular}[t]{l}$S^1$\end{tabular}}}}%
    \put(0.03426862,0.28633963){\makebox(0,0)[lt]{\lineheight{1.25}\smash{\begin{tabular}[t]{l}$S^1$\end{tabular}}}}%
  \end{picture}%
\endgroup%

	\caption[Example of a \ccc{} from $S^1$ to $S^1$]{Here $i$ is a diffeomorphism from $S^1$ to the boundary of $D^2$, and $j$ is a smooth embedding of $S^1$ into the interior of the disk $D^2$. We have that  $\cmor{D^2}{i}{S^1}{j}{S^1}$ is a \ccc{} (Proposition~\ref{disk}).} 
	\label{fig:ccc_S1,D2}
\end{figure}

\begin{example}
	Let $X$ be a space.
	The cospan $\cmor{X}{\id_X}{X}{\id_X}{X}$ is not a \ccc{}, unless $X=\emptyset$. 
	This is clear from the previous remark.
\end{example}

Physically we expect the identity cospan to be (the equivalence class of)
$\cmor{X\times \II}{\iota^X_0}{X}{\iota^X_1}{X}$.

\begin{proposition}\label{pr:identity_ccc}
	For $X$ a topological space, the cospan
	$\cmor{X\times \II}{\iota^X_0}{X}{\iota^X_1}{X}$ is a \ccc{} (where $\iota^X_a\colon X \to X \times \II$ is the map $x\mapsto (x,a)$).
\end{proposition}

\begin{proof}
	The complement of the image of $\langle \iota^X_0 , \iota^X_1 \rangle \colon X\sqcup X \to X\times \II$ is $X\times(0,1)$ which is open, so the image is a closed set.
	
	We now show $\langle \iota^X_0 , \iota^X_1 \rangle \colon X\sqcup X \to X\times \II$ is a cofibration. Let $K$ be any space and suppose we have a homotopy $h\colon (X\sqcup X)\times \II \to K$.
	By Lemma~\ref{le:producthom} and Theorem~\ref{th:lapcl}, the product with $\II$ preserves colimits, using this together with the universal property of the coproduct, the map $h\colon (X\sqcup X)\times \II \to K$ is uniquely defined by a pair of maps $h_0\colon X\times \II\to K$ and $h_1\colon X\times \II\to K$. 
	Now suppose we have a map $f\colon X \times \II \to K$
	such that for all $x\in X$ we have $h_0(x,0)=f(x,0)$ and $h_1(x,0)=f(x,1)$.
	(Notice this implies $h(\tilde{x},0)=f(\lan\iota_0^X,\iota_0^Y\ran(\tilde{x}))$ for $\tilde{x}\in X\sqcup X$.)
	
	We can construct a homotopy $H\colon (X \times \II) \times \II \to K$ which commutes with $h$ and $f$ as follows.
	Let $L = \left\{ 0,1 \right\} \times \II \cup \II \times \left\{ 0\right\}$ be the subset of the unit square consisting of the two vertical edges and the bottom horizontal edge.
	Let $\Gamma \colon \II \times \II \to L$ be a retraction sending the unit square to the subset $L$, see Example~\ref{ex:I,{0,1}}.
	We denote elements of $X\times L\subset (X\times \II)\times \II$ as triples $(x,s,t)$ and define $g\colon X\times L\to K$ as
	\[
	g(x,s,t)=\begin{cases}
		f(x,s) & t=0,\\
		h_0(x,t) & s=0,\\
		h_1(x,t) & s=1.
	\end{cases}
	\]
	By assumption these agree on the overlap and so $g$ is continuous.
	Now define $H\colon (X\times \II) \times \II \to K$
	by
	$g(x,\Gamma(s,t))$.
\end{proof}

\begin{proposition}\label{disk}
	(See Figure~\ref{fig:ccc_S1,D2}.)
	Consider $S^1$ and $D^2$ as smooth manifolds.
	There is a \ccc{}
	$\cmor{D^2}{i}{S^1}{j}{S^1}$
	where $i$ is any diffeomorphism sending $S^1$ to the boundary of $D^2$, and $j$ is any smooth embedding of $S^1$ into the interior of $D^2$.
\end{proposition}
\begin{proof} 
	The map $\lan i,j\ran\colon S^1\sqcup S^1\to D^2$ is the composition of a homeomorphism from $S^1\sqcup S^1$ to $i(S^1)\sqcup j(S^1)$, and an inclusion $\iota\colon i(S^1)\sqcup j(S^1) \to D^2$.
	Proposition~\ref{pr:sbmfld_cofib} gives that $\iota$ is a cofibration, by Proposition~\ref{pr:cofibdpair} the homeomorphism is a cofibration, and by Proposition~\ref{le:cofib_comp} the composition is a cofibration.
\end{proof}

The following definition can be found in e.g. \cite{lurie}, where it is referred to as a bordism.
\begin{definition}\label{de:cobord}
	An $n$-dimensional 
	{\em concrete cobordism} from an $(n-1)$-dimensional smooth oriented manifold $X$ to an $(n-1)$-dimensional smooth oriented manifold $Y$, is an $n$-dimensional smooth oriented manifold $M$ equipped with an orientation preserving  diffeomorphism 
	$\phi \colon \bar{X}\sqcup Y \to \partial M$ (where the bar denotes the opposite orientation).
\end{definition}

\begin{proposition}\label{pr:concob_to_concsp}
	There is a canonical way to map a concrete cofibration to a \ccc{}. Precisely,
	let $X$, $Y$ and $M$ be smooth oriented manifolds, and 
	let $M$ be a concrete cobordism from $X$ to $Y$. Hence there exists a diffeomorphism $\phi\colon \bar{X}\sqcup Y\to \partial M$.
	Define maps $i(x)=\phi(x,0)$ and $j(y)=\phi(y,1)$.
	Then, using $X$, $Y$ and $M$ to denote the underlying topological spaces, $\cmor{M}{i}{X}{j}{Y}$ is a \ccc{}.
\end{proposition}
\begin{proof}
	The pair $(M,\partial M)$ is cofibred by Proposition \ref{pr:mfld_bdy}.
	The map $\copr{i}{j}$ is a homeomorphism onto its image $\partial M$ as $\phi$ is a diffeomorphism, hence, using Proposition \ref{pr:cofibdpair}, $\copr{i}{j}$ is a cofibration.
	The boundary $\partial M$ is closed so $\copr{i}{j}$ a closed cofibration. 
\end{proof}

	\begin{figure}
	\centering
	\def\svgwidth{0.5\columnwidth}
\begingroup%
  \makeatletter%
  \providecommand\color[2][]{%
    \errmessage{(Inkscape) Color is used for the text in Inkscape, but the package 'color.sty' is not loaded}%
    \renewcommand\color[2][]{}%
  }%
  \providecommand\transparent[1]{%
    \errmessage{(Inkscape) Transparency is used (non-zero) for the text in Inkscape, but the package 'transparent.sty' is not loaded}%
    \renewcommand\transparent[1]{}%
  }%
  \providecommand\rotatebox[2]{#2}%
  \newcommand*\fsize{\dimexpr\f@size pt\relax}%
  \newcommand*\lineheight[1]{\fontsize{\fsize}{#1\fsize}\selectfont}%
  \ifx\svgwidth\undefined%
    \setlength{\unitlength}{533.05866408bp}%
    \ifx\svgscale\undefined%
      \relax%
    \else%
      \setlength{\unitlength}{\unitlength * \real{\svgscale}}%
    \fi%
  \else%
    \setlength{\unitlength}{\svgwidth}%
  \fi%
  \global\let\svgwidth\undefined%
  \global\let\svgscale\undefined%
  \makeatother%
  \begin{picture}(1,0.66574144)%
    \lineheight{1}%
    \setlength\tabcolsep{0pt}%
    \put(0,0){\includegraphics[width=\unitlength,page=1]{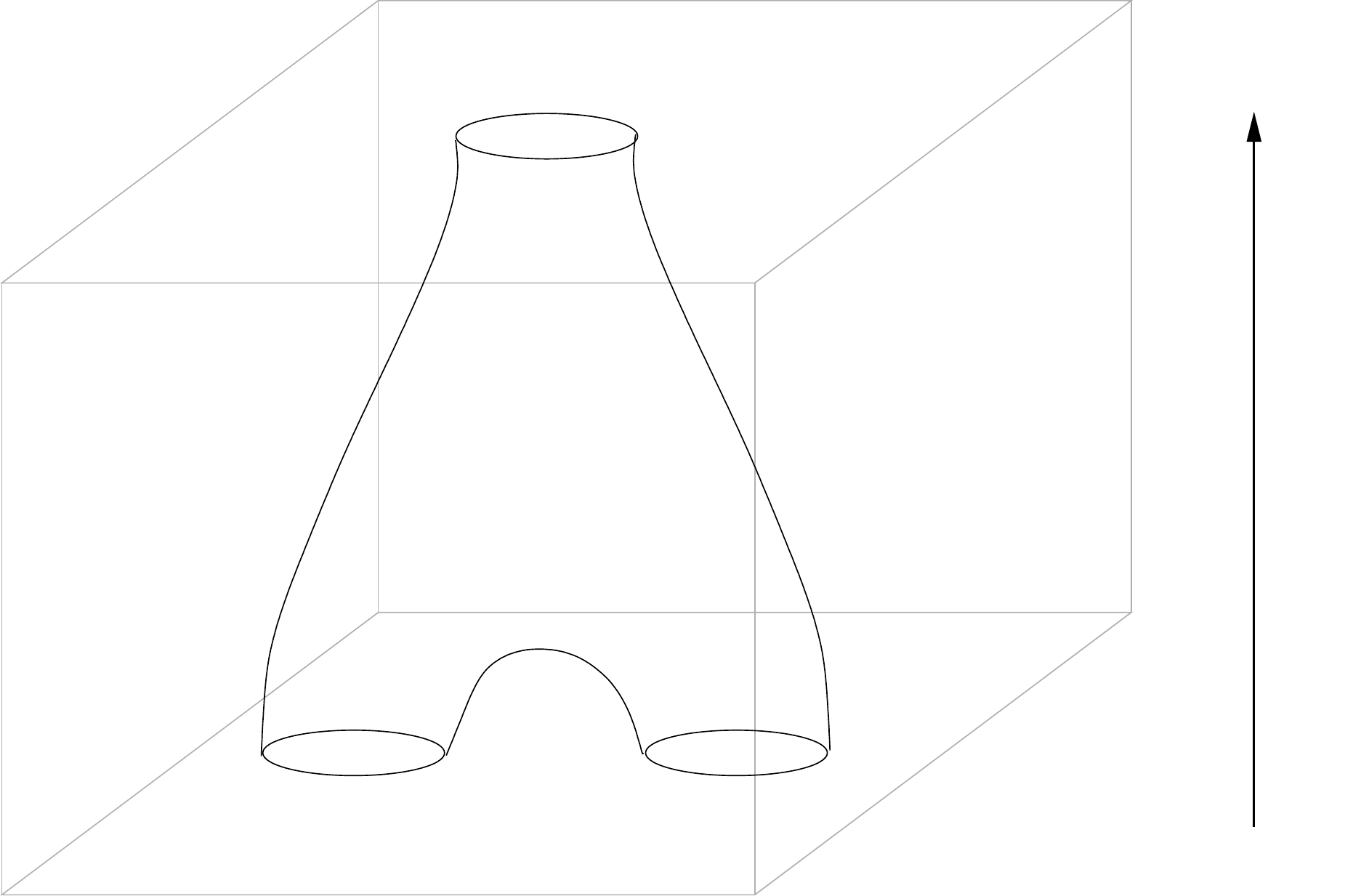}}%
    \put(0.95125213,0.03007156){\color[rgb]{0,0,0}\makebox(0,0)[lt]{\lineheight{1.25}\smash{\begin{tabular}[t]{l}$X$\end{tabular}}}}%
    \put(0.95125213,0.57068379){\color[rgb]{0,0,0}\makebox(0,0)[lt]{\lineheight{1.25}\smash{\begin{tabular}[t]{l}$Y$\end{tabular}}}}%
    \put(0.97006093,0.31062678){\color[rgb]{0,0,0}\makebox(0,0)[lt]{\lineheight{1.25}\smash{\begin{tabular}[t]{l}$M$\end{tabular}}}}%
  \end{picture}%
\endgroup%

	\caption[Example of a \ccc{} obtained from an embedded submanifold]{Here the grey lines represent the corners of the manifold $\II^3$, and the black lines represent an embedded submanifold $M'\subset \II^3$. 
		Let $X$, be the complement of $M'$ in the bottom boundary, $\II^2\times \{0\}$, $Y$ the complement in the top boundary, $\II^2\times \{1\}$, and $M$ the complement in $\II^3$.
		Then there is a \ccc{} $\cmor{M}{i}{X}{j}{Y}$ where $i$ and $j$ are subspace inclusions.} 
	\label{fig:embmer}
\end{figure}

\exa{\label{ex:embmer} Consider the manifold (with corners) $\II^3$, and let $M'$ be an embedded submanifold as illustrated by the black part of Figure~\ref{fig:embmer}. 
	Let $M=\II^3\setminus M'$, $X=(\II^2\times \{0\})\setminus (M\cap (\II^2\times \{0\}))$ and $Y=(\II^2\times \{1\})\setminus (M\cap (\II^2\times \{1\}))$, i.e. $X$ is the complement of $M'$ in bottom boundary in the figure and $Y$ the top boundary.
	There is a \ccc{} $\cmor{M}{i}{X}{j}{Y}$ where  $i$ and $j$ are subspace inclusions.
	We can see this by noticing that there are non-intersecting neighbourhoods of the top and bottom boundary of $M$ that are homeomorphic to $X\times [0,\epsilon]$ and $Y\times [0,\epsilon']$ with $\epsilon,\epsilon'\in \R^{+}$. Thus an $H$ and $\phi$ satisfying the conditions of Theorem~\ref{th:hom_for_closed_cofibred_pair} can be constructed as in the proof of Proposition~\ref{pr:mfld_bdy}. 
	}

\exa{\label{ex:mer}
	Figure~\ref{fig:mer} represents a \ccc{}. Proposition~\ref{pr:mfld_bdy} gives that $\lan i,j\ran$ is a cofibration. Notice also that the boundary is a closed subset of $M$.}

	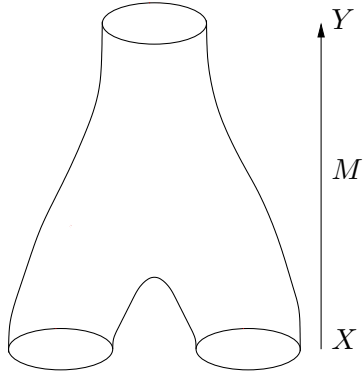
\begin{figure}
	\centering
	\def\svgwidth{0.3\columnwidth}
\begingroup%
  \makeatletter%
  \providecommand\color[2][]{%
    \errmessage{(Inkscape) Color is used for the text in Inkscape, but the package 'color.sty' is not loaded}%
    \renewcommand\color[2][]{}%
  }%
  \providecommand\transparent[1]{%
    \errmessage{(Inkscape) Transparency is used (non-zero) for the text in Inkscape, but the package 'transparent.sty' is not loaded}%
    \renewcommand\transparent[1]{}%
  }%
  \providecommand\rotatebox[2]{#2}%
  \newcommand*\fsize{\dimexpr\f@size pt\relax}%
  \newcommand*\lineheight[1]{\fontsize{\fsize}{#1\fsize}\selectfont}%
  \ifx\svgwidth\undefined%
    \setlength{\unitlength}{295.76913441bp}%
    \ifx\svgscale\undefined%
      \relax%
    \else%
      \setlength{\unitlength}{\unitlength * \real{\svgscale}}%
    \fi%
  \else%
    \setlength{\unitlength}{\svgwidth}%
  \fi%
  \global\let\svgwidth\undefined%
  \global\let\svgscale\undefined%
  \makeatother%
  \begin{picture}(1,1.08463313)%
    \lineheight{1}%
    \setlength\tabcolsep{0pt}%
    \put(0,0){\includegraphics[width=\unitlength,page=1]{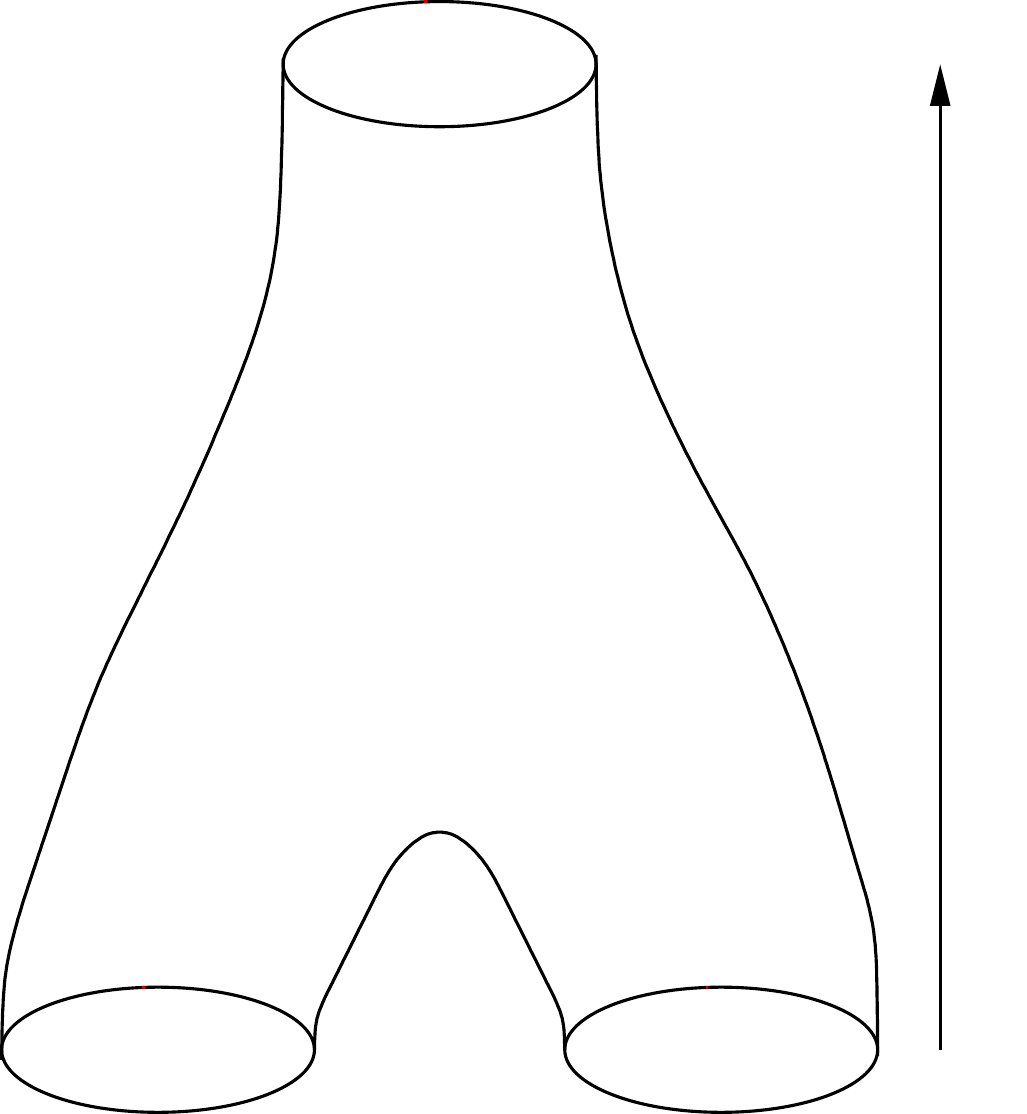}}%
    \put(0.94600726,0.06245781){\color[rgb]{0,0,0}\makebox(0,0)[lt]{\lineheight{1.25}\smash{\begin{tabular}[t]{l}$X$\end{tabular}}}}%
    \put(0.94600726,1.00694171){\color[rgb]{0,0,0}\makebox(0,0)[lt]{\lineheight{1.25}\smash{\begin{tabular}[t]{l}$Y$\end{tabular}}}}%
    \put(0.94600726,0.56516698){\color[rgb]{0,0,0}\makebox(0,0)[lt]{\lineheight{1.25}\smash{\begin{tabular}[t]{l}$M$\end{tabular}}}}%
  \end{picture}%
\endgroup%

	\caption[Example of a \ccc{} from $S^1\sqcup S^1$ to $S^1$]{Let $M$ be a manifold represented by the figure, with boundary homeomorphic to $S^1\cup S^1$ at the bottom, and $S^1$ at the top. Let $X$ be the bottom boundary and $Y$ the top boundary.
		Then there is a \ccc{} $\cmor{M}{i}{X}{j}{Y}$ where $i$ and $j$ are subspace inclusions.} 
	\label{fig:mer}
\end{figure}

\begin{lemma}\label{le:cofcsp_rev}
	For any pair $X,Y\in Ob(\Topo)$ there is a bijection
\ali{
\rev_{X,Y}\; \colon \; \;\cCofCsp(X,Y) \;\;\; \; &\to \;\;\;\; \cCofCsp(Y,X)\\
\cmortikz{M}{i}{X}{j}{Y} \;\;\;\;&\mapsto \;\;\;\; \cmortikz{M.}{j}{Y}{i}{X}}
\end{lemma}

\begin{proof}
	We first check $\rev$ is well defined.
	The image of $\lan j,i\ran $ is the same as the image of $\lan i,j \ran $ so it is a closed.
	We show it is a cofibration.
	Suppose we have a space $K$, and maps $h\colon (Y\sqcup X)\times \II\to K$
	and $f\colon M\to K$ satisfying the commutativity conditions of Definition~\ref{de:homotopy_extension}.
	The map $h$ canonically determines a map $h'\colon (X\sqcup Y)\times \II\to K$.
	The map $\lan i,j\ran $ is a cofibration so we can apply the homotopy extension property to give a map $H\colon M\times \II\to K$ which extends $f$ and $h'$.
	This $H$ also commutes with $f$ and $h$.

	It is clear that $\rev$ is its own inverse, thus it is a bijection.
\end{proof}

\lemm{\label{ccc_implies_cofibs}
	If $\cmor{M}{i}{X}{j}{Y}$ is a \ccc{}, then $i\colon X\to M$ and $j\colon {Y}\to {M}$ are closed cofibrations.}
\begin{proof}
	The map
$i\colon X\to M$ is equal to the composition $X\xrightarrow{\iota_1} X\sqcup Y \xrightarrow{\lan i,j\ran} M$.
The map $\iota_1$ is a cofibration by Lemma~\ref{cofib_copr} and the composition of cofibrations is a cofibration by Lemma~\ref{le:cofib_comp}, hence $i$ is a cofibration.	

We now prove that the image of $X$ under the composition is closed in $M$. Here we use primes to denote images of $\lan i,j\ran$.
The map $\lan i,j\ran$ is an embedding by Theorem~\ref{th:homeo_onto_im}, hence a homeomorphism onto its image, and it is straightforward to see that $\iota_1(X)$ is closed in $X\sqcup Y$.
Hence
there exists an open $U\subseteq M$ with $U\cap (X\sqcup Y)'= (X\sqcup Y)'\setminus \iota_1(X)'$.
The image of $X\sqcup Y$ is closed since $\lan i,j\ran$ is a closed cofibration, so $M\setminus (X\sqcup Y)'$ is an open set.
Thus there is an open set $M\setminus (X\sqcup Y)'\cup U=M\setminus \iota_1(X)'$, hence the image of $X$ under $\lan i,j \ran \circ \iota_1$ is closed.

The same argument gives that $j$ is a closed cofibration.
\end{proof}

\lemm{\label{le:comp_ccc}
$(I)$ For any spaces $X,Y$ and $Z$ in $Ob(\Topo)$ there is a composition of \ccc{}s
\ali{
	\sbullet\;\colon \cCofCsp(X,Y)\times \cCofCsp(Y,Z)&\to \cCofCsp(X,Z) \\
\left({\cmortikz{M}{i}{X}{j}{Y}}\;\raisebox{-1em}{,}\;{\cmortikz{N}{k}{Y}{l}{Z}}\right)&\mapsto \cmortikz{M \sqcup_Y N}{\tilde{i}}{X}{\tilde{l}}{Z}
}
where $\tilde{i}=p_M\circ i$ and $\tilde{l}=p_N\circ l$ are obtained via the following diagram 
\[
\begin{tikzcd}[ampersand replacement= \&,column sep=tiny]
	\makebox*{$M\sqcup_Y N$}{$X$} \ar[dr,"i"']\& \& Y\ar[dl,"j"]\ar[dr,"k"'] \& \& \makebox*{$M\sqcup_Y N$}{$Z$}\ar[dl,"l"] \\
	\&  M\ar[dr,"p_M"'] \& \& N \ar[dl,"p_N"]\& \\
	\& \&  M\sqcup_{Y} N, \& \& 
\end{tikzcd}
\]
the middle square of which is the
pushout of $j\colon M \leftarrow Y \to N \colon k$ in \textbf{Top}.\\
$(II)$ Hence there is a magmoid 
\[
\cCofCsp=(Ob(\Topo),\cCofCsp(-,-),\sbullet).
\] \\
We will also use notation ${\cmortikz{N}{k}{Y}{l}{Z}} \; \sbullet \; {\cmortikz{M}{i}{X}{j}{Y}}= \; \sbullet \;\left({\cmortikz{M}{i}{X}{j}{Y}}\;\raisebox{-1em}{,}\;{\cmortikz{N}{k}{Y}{l}{Z}}\right)$ for composition.
}

\begin{proof}
We need to prove that $\langle \tilde{i}, \tilde{l} \rangle\; \colon X\sqcup Z \to M\push{jk}{Y}N$ is a closed cofibration.
We first check the map is closed.
The image of $\langle \tilde{i}, \tilde{l} \rangle$ is equal to 
$p_M(i(X))\cup p_N(l(Y))$.
Sets in $M\sqcup_{Y} N$ are closed if the preimage under $p_{M}$ and $p_N$ is closed in $M$ and $N$ respectively. 
By Proposition~\ref{pr:cofibdpair}, $\lan i,j\ran$ is a homeomorphism onto its image, hence 
we have $i(X)\cap j(Y)=\emptyset$. 
This implies $p_N^{-1}(p_M(i(X)))=\emptyset$, which is closed, and $p_M^{-1}(p_M(i(X)))=i(X)$ is closed by Lemma~\ref{ccc_implies_cofibs}. Hence $p_M(i(X))$ is closed in $M\sqcup_Y N$.
Similarly $p_N(l(Y))$ is closed.

We now check $\langle \tilde{i}, \tilde{l} \rangle$ is cofibration.
Define $J$ to be the map obtained by taking either route around the pushout square:
\[
\begin{tikzcd}[ampersand replacement= \&,column sep =tiny]
 \& Y \ar[dl,"j"']  \ar[dr,"k"] \ar[dd,"J"]\& \\
M \ar[dr,"p_M"'] \& \& N \ar[dl,"p_N"] \\
\& M\sqcup_Y N. \&
\end{tikzcd}
\]
We will prove that we have a cofibration 
$\big\lan\lan \tilde{i},J\ran,\tilde{l} \big\ran \colon (X\sqcup Y) \sqcup Z \to M \sqcup_Y N $, then by Lemma~\ref{le:cofib_comp} and a straightforward extension of Lemma~\ref{cofib_copr} we then have that the composition
\[
\begin{tikzcd}[column sep=large, ampersand replacement= \&]
X \sqcup Z \ar[r] \& (X \sqcup Y) \sqcup Z \ar[r,"{\langle \lan \tilde{i},J\ran ,\tilde{l}\rangle}"] \& M\sqcup_Y N 
\end{tikzcd},
\] 
which is equal to $\langle \tilde{i}, \tilde{l} \rangle$, is a cofibration.
Let $K$ be a space and suppose we have maps 
$f \colon M\sqcup_Y N \to K$ and $h \colon ((X\sqcup Y) \sqcup Z)\times \II \to  K$ satisfying the commutation conditions of Definition~\ref{de:homotopy_extension}. 
We construct a map 
$H \colon (M\sqcup_Y N) \times \II \to K$ extending $f$ and $h$ as follows.
First note that by Lemma~\ref{le:producthom} and Theorem~\ref{th:lapcl}, the product with $\II$ preserves coproducts and thus we have canonical isomorphisms, $((X\sqcup Y) \sqcup Z)\times \II \cong ((X \times \II) \sqcup (Y \times \II)) \sqcup (Z \times \II)$ and $(M\sqcup_Y N) \times \II \cong M \times \II \sqcup_{Y\times \II} N\times \II$.
Thus, by the universal property of the coproduct we have that the map $h$ is in one to one correspondence with a triple of maps 
$h_X \colon X \times \II \to K$,
$h_Y \colon Y \times \II \to K$ and
$h_Z \colon Z \times \II \to K$. 
Now using the homotopy extension property of $\lan i,j\ran$ on the maps $\langle h_X,h_Y \rangle$ and the restriction of $f$ to $M$, we obtain a map $\mathcal{H}_L \colon M \times \II \to K$.
Similarly we obtain a map $\mathcal{H}_R \colon N \times \II \to K$.
These two homotopies agree on the images of $Y \times \II$ by construction so we can use the universal property of the pushout to obtain a map
$\langle \mathcal{H}_L,\mathcal{H}_R \rangle \colon  M \times \II \sqcup_{Y\times \II} N\times \II \to K$ which, precomposed with the canonical isomorphism $(M\sqcup_Y N) \times \II \cong M \times \II \sqcup_{Y\times \II} N\times \II$,
is a homotopy extending $h$.
\end{proof}

\prop{The family of bijections $\rev_{X,Y}\colon \cCofCsp(X,Y)\to \cCofCsp(Y,X)$ defined in Proposition~\ref{le:cofcsp_rev} show that $\cCofCsp$ is reversible. \qed}

\subsection{Category of cofibrant cospans \texorpdfstring{$\CofCsp$}{CofCsp}}
Notice that the composition in $\cCofCsp$ is not strictly associative.
Here we impose a congruence on \ccc{}s such that we obtain a category.

One option would be \textit{cospan isomorphism}, by which we mean
$\cmor{M}{i}{X}{j}{Y}$ is equivalent to $\cmor{N}{i'}{X}{j'}{Y}$ if there exists a homeomorphism $M\to N$ which commutes with the cospans.
This is a direct analogue of the equivalence usually used for smooth manifold cobordisms in e.g. \cite{lurie}.
This equivalence would be sufficient to give an associative composition. However it will not be sufficient to ensure the cospan $\cmor{X\times \II}{\iota^X_0}{X}{\iota^X_1}{X}$ behaves as an identity. (This is the image of a representative of the smooth manifold cobordism identity under the map described in Proposition~\ref{pr:concob_to_concsp}.)
One way to see this is by thinking about the cospan in Example~\ref{disk}: 
taking a pushout over the map from $S^1$ into the interior of the disk, and the map into $S^1 \times \II$ will not give a space homeomorphic to the disk.
Hence we use a stronger equivalence relation.

\defn{For each pair $X,Y\in Ob(\cCofCsp)$,
we define a relation on $\cCofCsp(X,Y)$ by
\[
\left(\cmortikz{M}{i}{X}{j}{Y}\right)\simche \left(\cmortikz{N}{i'}{X}{j'}{Y}\right)
\] 
if there exists a commuting diagram 
\[
\begin{tikzcd}[ampersand replacement=\&, sep =scriptsize]
	\& M \ar[dd,"\psi"]\& \\
	X \ar[ur,"i"] \ar[dr,"i'"'] \&  \& Y\ar[ul,"j"']\ar[dl,"j'"] \\
	\& M' \&
\end{tikzcd}
\]
where $\psi$ is a homotopy equivalence.}

\lemm{\label{le:che}
The relation $\simche$ an equivalence relation.}

\begin{proof}
	Using the universal property of the coproduct we can rewrite the relation in terms of a homotopy equivalence $\psi\colon M\to M'$ which is a map of spaces under $X\sqcup Y$ from $\lan i,j\ran \colon X\sqcup Y\to M$ to $\lan i',j'\ran \colon X\sqcup Y\to M'$. 
Then since the maps $X\sqcup Y\to M$ are defined to be cofibrations, Theorem~\ref{th:cofib_equiv} gives that this relation is precisely cofibre homotopy equivalence of spaces under $X\sqcup Y$, thus is an equivalence relation by Proposition~\ref{pr:cofib_hom_equiv}.
\end{proof}

\rem{The fact that, by Theorem~\ref{th:cofib_equiv}, \che{} is equivalent to cofibre homotopy equivalence of spaces under the disjoint union of the objects, will be used to obtain a congruence from \che{}. We could instead have {\it defined} \che{} to be cofibre homotopy of spaces under the disjoint union of the objects. Then we would use Theorem~\ref{th:cofib_equiv} in the proof of the identity axiom instead.}

We call the map $\psi$ a {\em \che}, and refer to an equivalence class of \ccc{}s as just a {\em \cc{}}, denoted $\classche{\cmor{M}{i}{X}{j}{Y}}$. We have 
\[
\cCofCsp/ \simche (X,Y)=\left\{\classche{\cmortikz{M}{i}{X}{j}{Y}} \; \middle\vert \; \lan i,j \ran  \text{ is a closed cofibration } \right\}.
\]

\lemm{\label{le:comp_cc}
	For each pair $X,Y\in \Topo$ the relations $(\cCofCsp(X,Y),\simche)$  are a congruence on $\cCofCsp$ and hence we have a magmoid 
	\[
	\CofCsp=\cCofCsp/\simche=(Ob(\Topo),\cCofCsp/\simche\; ,\; \sbullet).
	\]
}
\begin{proof}
	We have from Lemma~\ref{le:che} that the $\simche$ are equivalence relations for each pair $X,Y\in \Topo$, thus we only need to check that the relations respect composition.
	
	Let  $\cmor{M}{i}{X}{j}{Y}$ and $\cmor{M'}{i'}{X}{j'}{Y}$ be two representatives of the same \cc{} from $X$ to $Y$ and similarly let 
	$\cmor{N}{k}{Y}{l}{Z}$ and $\cmor{N'}{k'}{Y}{l'}{Z}$ be representatives of the same \cc{} from $Y$ to $Z$. 
	
	Using Theorem~\ref{th:cofib_equiv}, and the universal property of the coproduct, we have the following commuting diagram where $\phi,\phi',\psi$ and $\psi'$ are cofibre homotopy equivalences between spaces under of $X,Y$ and $Z$ as shown.
	\[
	\begin{tikzcd}[ampersand replacement=\&, column sep=small]
	\& \& \makebox*{$M'$}{$M \sqcup_{Y} N$} \& \& \\
	\& M \ar[ur,"p_{M}"]\ar[dd,"\phi"',bend right] \& \& N \ar[ul,"p_{N}"']\ar[dd,"\psi"',bend right] \& \\
	\makebox*{$M\sqcup_{Y} N$}{$\hspace*{1em}X$}\ar[ur,"i"]\ar[dr,"i'"'] \& \& Y\ar[ul,"j"']\ar[ur,"k"]\ar[dl,"j'"]\ar[dr,"k'"'] \& \& \makebox*{$M\sqcup_{Y} N$}{$Z\hspace*{1em}$}\ar[ul,"l"']\ar[dl,"l'"] \\
	\& M' \ar[dr,"p_{M'}"']\ar[uu,"\phi'"',bend right]  \& \& N'\ar[dl,"p_{N'}"]\ar[uu,"\psi'"',bend right] \& \\
	\& \& M'\sqcup_{Y} N' \& \& 
	\end{tikzcd}
	\]
	This means there exists a homotopy under $X\sqcup Y$, say $H_{\phi}\colon M\times \II \to M$, from $\phi'\circ \phi$ to the identity and a homotopy under $Y\sqcup Z$, say  $H_{\psi}\colon N\times \II\to N$, from $\psi'\circ \psi$ to the identity.
	In particular, for all $y\in Y$, we have $H_\phi(j(y),t)=j(y)$  and $H_\psi(k(y),t)=k(y)$.
	
	By the universal property of the pushout, the commuting pair $p_{M'}\circ \phi$ and $p_{N'}\circ \psi$ uniquely determine a map $F\colon M\sqcup_{Y}N \to M' \sqcup_{Y} N'$ making the diagram commute. We will show $F$ is a homotopy equivalence.
	
	We can similarly construct a map $F'\colon M' \sqcup_{Y} N' \to M\sqcup_{Y}N$ using the pair $ p_{M}\circ \psi'$ and $ p_{N}\circ \phi'$.
	Notice the maps $ p_M\circ H_\phi\circ(j\times \id_\II)\colon Y\times \II\to M\sqcup_{Y}N$ and $p_N\circ H_\psi\circ (k\times \id_\II )\colon Y\times \II\to M\sqcup_{Y}N$ commute
	using that for all $y\in Y$ we have $H_\psi(k(y),t)=k(y)$ and $H_\phi(j(y),t)=j(y)$, and the commutativity of the diagram.
	Taking the product with $\II$ of the pushout of $j$ and $k$ is still a pushout, by Lemma~\ref{le:producthom}.
	Using the universal property of this pushout on the maps $p_M\circ H_\phi$ and $p_N\circ H_\psi$ gives a map $(M\sqcup_Y N)\times \II\to M\sqcup_{Y} N$ which is a homotopy from $F'\circ F$ to the identity functor.
	
	In the same way we can construct a homotopy $F\circ F'$ to the identity.
\end{proof}

\begin{theorem}\label{th:CofCsp}
The quadruple 
\[
\CofCsp=\left(Ob(\Topo)\;,\;\cCofCsp(X,Y)/\simche\;,\;\sbullet\;,\;\classche{\cmortikz{X\times \II}{\iota_{0}^X}{X}{\iota_1^X}{X}}\;\right)
\]	
is a category.	
\end{theorem}
\begin{proof}
	Note that $\left(Ob(\Topo)\,,\;\cCofCsp(X,Y)/\simche\;,\;\sbullet\right)$ is a magmoid by Lemma~\ref{le:comp_cc}.\\
	$(\CC1)$ Note first that $\cmor{X\times \II}{\iota_{0}^X}{X}{\iota_1^X}{X}$ is a \ccc{} by Proposition~\ref{pr:identity_ccc}.
	Suppose we have a \cc{} represented by $\cmor{M}{i}{X}{j}{Y}$.
	We will show there is a \che{} from 
	$(\cmor{M}{i}{X}{j}{Y})\sbullet(\cmor{Y\times \II}{\iota^Y_0}{Y}{\iota^Y_1}{Y})$ to $\cmor{M}{i}{X}{j}{Y}$.
	Consider the following diagram.
	\[
	\begin{tikzcd}[ampersand replacement= \&,column sep=tiny]
		\makebox*{$M\sqcup_Y N$}{$X$} \ar[dr,"i"]\& \& Y\ar[dl,"j"']\ar[dr,"\iota^Y_0"] \& \& \makebox*{$M\sqcup_Y N$}{$Y$}\ar[dl,"\iota^Y_1"'] \\
		\&  M\ar[dr,"p_M"']\ar[ddr,"\id_M"',bend right] \& \& Y\times \II \ar[dl,"p_{Y\times \II}"]\ar[ddl,"{(y,t) \mapsto j(y)}",bend left]\& \\
		\& \&  M\sqcup_{Y} (Y\times \II)\ar[d,"\phi",dotted] \& \& \\[+2em]
		\& \& M \& \& 
	\end{tikzcd}
	\]
The map $\phi$ is constructed using the universal property of the pushout.
	By construction $\phi$ commutes with the cospans 
	$(\cmor{M}{i}{X}{j}{Y})\sbullet(\cmor{Y\times \II}{\iota^Y_0}{Y}{\iota^Y_1}{Y})$ and $\cmor{M}{i}{X}{j}{Y}$.
	We claim $\phi$ is a homotopy equivalence with homotopy inverse $p_M$.
	It is by construction that $\phi \circ p_M=\id_M$.
	
	We construct a homotopy $ p_M \circ \phi \to \id_{M\sqcup_Y(Y\times \II)}$ as follows.
	Since $M\sqcup_Y(Y\times \II)$ is a pushout, the map $p_M \circ \phi$ is uniquely determined by the pair of maps $M\to M\sqcup_Y(Y\times \II)$, $m\mapsto p_M(m)$ and $Y\times \II \to M\sqcup_Y(Y\times \II)$, $(y,t)\mapsto p_M(j(y))$, or equivalently $(y,t)\mapsto p_{Y\times \II}(\iota_{0}^Y(y))$.
	Similarly the identity is determined by the pair $M\to M\sqcup_Y(Y\times \II)$, $m\mapsto p_M(m)$ and $Y\times \II \to M\sqcup_Y(Y\times \II)$, $(y,t)\mapsto p_{Y\times \II}(y,t)$.
	The map $H_{Y\times \II} \colon (Y\times \II) \times \II \to M\sqcup_Y (Y \times \II)$, $((y,t),s) \mapsto p_{Y\times \II}(y,ts)$ is a homotopy between the two maps from $Y\times \II$.
	And for $M$ we can use the homotopy
	$H_M\colon M\times \II \to M\sqcup_Y (Y \times \II)$, $(m, t) \mapsto  p_M(m)$.
	
	By Lemma~\ref{le:producthom} the product with $\II$ preserves pushouts.
	Notice that 
	$H_M\circ (j\times \id)\colon Y\times \II\to M\sqcup_Y (Y \times \II)$ is $(y,t)\mapsto p_M(j(y))$ and 
	$H_{Y\times \II}\circ (\iota_0^Y\times \id)\colon Y\times \II\to M\sqcup_Y (Y \times \II)$
	is 
	$(y,s)\mapsto p_{Y\times \II} (\iota_0^Y(y))$,
	so we can use the universal property of the pushout of $j\times \id$ and $\iota_0^Y\times \id$ to obtain a homotopy 
	$\mathcal{H} \colon (M\sqcup_Y (Y \times \II)) \times \II \to M\sqcup_Y (Y \times \II)$ from $p_M \circ \phi$ to $\id_{M\sqcup_Y(Y\times \II)}$.
	
	We can similarly construct a \che{}  
	$(\cmor{Y\times \II}{\iota^Y_0}{Y}{\iota^Y_1}{Y})\sbullet(\cmor{M}{i}{X}{j}{Y})$ to $\cmor{M}{i}{X}{j}{Y}$.
	
	$(\CC2)$ We now check that the composition is associative. 
	Let $\cmor{M}{i}{W}{j}{X}$, $\cmor{N}{k}{X}{l}{Y}$ and $\cmor{O}{r}{Y}{s}{Z}$ be \ccc{}s.
	The two ways to compose these three cospans corresponds to taking a pushout first over $X$ or first over $Y$ as shown in the following diagram
	\[
	\begin{tikzcd}[ampersand replacement=\&, column sep =0]
	\makebox*{$\left(M\sqcup_X N\right) \sqcup_Y O$}{$W$}\ar[dr,"i"'] \&\&[-20pt]\&[-50pt] X\ar[dll,"j"]\ar[dr,"k"'] \&[+25pt] \&[+25pt]  Y\ar[dl,"l"]\ar[drr,"r"'] \&[-50pt] \&[-20pt]\& \makebox*{$\left(M\sqcup_X N\right) \sqcup_Y O$}{$Z$}\ar[dl,"s"] \\
	\& M\ar[dr]\ar[ddrrrr] \&\&\& N\ar[drr]\ar[dll,crossing over] \&\&\& O\ar[dl]\ar[ddllll,crossing over] \\
	\& \&  M \sqcup_X N\ar[dr] \&\&\&\& N\sqcup_Y O\ar[dl] \\
	\&\&\& |[xshift=-20pt]|\left(M\sqcup_X N\right) \sqcup_Y O \&\& |[xshift=20pt]|M\sqcup_X \left(N\sqcup_Y O \right). 
	\end{tikzcd}
	\] 
	We can use the universal property of the pushout on the pair of maps $M\to M\sqcup_X \left(N\sqcup_Y O \right) $ and $N\to N\sqcup_YO\to M\sqcup_X \left(N\sqcup_Y O \right)$ to obtain a map $ M\sqcup_X N\to  M\sqcup_X \left(N\sqcup_Y O \right)$.
	We can then apply the universal property again to this map $ M\sqcup_X N\to M\sqcup_X \left(N\sqcup_Y O \right)$ and the map $O\to N\sqcup_YO\to M\sqcup_X \left(N\sqcup_Y O \right)$ to obtain a map $\left(M\sqcup_X N\right) \sqcup_Y O \to M\sqcup_X \left(N\sqcup_Y O \right)$
	which commutes with the diagram.
	In a similar way we can obtain an inverse $M\sqcup_X \left(N\sqcup_Y O \right)\to \left(M\sqcup_X N\right) \sqcup_Y O$.
\end{proof}

Let $\cmor{M}{i}{X}{j}{Y}$ and $\cmor{N}{k}{Y}{l}{Z}$ be \ccc{}s. In an attempt to avoid excessive notation, from here we may use $i$ and $l$ to refer also to the maps $\tilde{i}=p_M\circ i$ and $\tilde{l}=p_N\circ l$ obtained in the composition.

\prop{The family of bijections $\rev_{X,Y}\colon \cCofCsp(X,Y)\to \cCofCsp(Y,X)$ from Lemma~\ref{le:cofcsp_rev} extend to an involutive functor
	\ali{
		\rev \; \colon \;\; \CofCsp&\to \CofCsp^{op}\\
		\classche{\cmortikz{M}{i}{X}{j}{Y}} &\mapsto \classche{\cmortikz{M.}{j}{Y}{i}{X}}
}}
\begin{proof}
	Lemma~\ref{le:cofcsp_rev} gives that $\rev$ is well defined, and that it is its own inverse.
	To show composition is preserved, let $\cmor{M}{i}{X}{j}{Y}$ and $\cmor{N}{k}{Y}{l}{Z}$ be \ccc{}s.
	Then the universal property of the pushout gives an isomorphism between $M\sqcup_Y N$ and $N\sqcup_Y M$, which gives a \che{} from $\rev\big((\cmor{N}{k}{Y}{l}{Z})\sbullet (\cmor{M}{i}{X}{j}{Y})\big) $ to $\rev (\cmor{N}{k}{Y}{l}{Z})\sbullet_{op} \rev(\cmor{M}{i}{X}{j}{Y}) = (\cmor{M}{j}{Y}{i}{X})\sbullet (\cmor{N}{l}{Z}{k}{Y})$.
\end{proof}

\subsubsection{Monoidal structure on \texorpdfstring{$\CofCsp$}{}}

We now construct a functor from  $\CofCsp\times \CofCsp$ to $\CofCsp$, and show that this leads to a symmetric monoidal category with underlying category $\CofCsp$.
\begin{lemma}\label{le:CofCsp_bifunctor}
	There is a functor
	\ali{
	\otimes \; \colon \; \CofCsp \times \CofCsp \; \;&\to\;\; \CofCsp\\
	\left(\classche{\cmortikz{M}{i}{W}{j}{X}}, \classche{\cmortikz{N}{k}{Y}{l}{Z}}\right)
	&\mapsto
	\classche{ \cmortikz{M\sqcup N}{i\sqcup k}{W\sqcup Y}{j\sqcup l}{X\sqcup Z}}
	}
	where $i\sqcup j$ is the image of a pair of maps under the monoidal product on $\Topo$ as in Proposition~\ref{pr:Top_mon}.
\end{lemma}
\begin{proof}
	We first check that 
	$\cmor{M\sqcup N}{i\sqcup k}{W\sqcup Y}{j\sqcup l}{X\sqcup Z}$
	is a \ccc{}.
	In particular we show that the map $\lan i\sqcup k,j\sqcup l\ran \colon (W\sqcup Y)\sqcup (X\sqcup Z)\to M\sqcup N$ is a closed cofibration.
	Let $K$ be a space and suppose we have maps $h\colon ((W\sqcup Y)\sqcup (X\sqcup Z))\times \II\to K$ and $f\colon M\sqcup N\to K$ satisfying the commutation conditions of Definition~\ref{de:homotopy_extension}.
	By Lemma~\ref{le:producthom}, the product with $\II$ preserves colimits so the map $h$ uniquely determines a pair $h'\colon (W\sqcup Y)\times \II\to K$ and $h''\colon (X\sqcup Z)\times \II\to K$.
	Similarly the map $f$ determines maps $f'\colon M\to K$ and $f''\colon N\to K$.
	We can use the homotopy extension property of $\lan i,j\ran$ on the pair $h'$ and $f'$ to obtain a map $H'\colon M\times \II \to K$ and similarly of $\lan k,l\ran$ on the pair $h''$ and $f''$ to obtain $H''\colon N\times \II \to K$. 
	Now using Lemma~\ref{le:producthom} again, $H'$ and $H''$ determine uniquely a map $H\colon (M\sqcup N)\times \II \to K$ extending $f$ and $h$.
	The image of $\lan i\sqcup k,j\sqcup l\ran$ is the union of the images of $\lan i,j\ran$ and $\lan k,l\ran$, thus is closed.
	
	We now check that the monoidal product respects the equivalence relation.
	Suppose we have a \ccc{} $\cmor{M'}{i'}{W}{j'}{X}$ which is cospan homotopy equivalent to $\cmor{M}{i}{W}{j}{X}$
	via some \che{} $\phi\colon M\to M'$ and similarly $\cmor{N'}{k'}{Y}{l'}{Z}$
	equivalent to 
	$\cmor{N}{k}{Y}{l}{Z}$ via $\psi\colon N\to N'$.
	Then there exist homotopy inverses $\phi'$ of $\phi$ and $\psi'$ of $\psi$.
	The following diagram commutes
	\[
	\begin{tikzcd}[ampersand replacement=\&]
	\& M\sqcup N \ar[dd,"\phi\sqcup \psi"'] \& \\
	W\sqcup Y\ar[ur,"i\sqcup k"]
	\ar[dr,"i'\sqcup k'"'] \& \&
	X\sqcup Z\ar[ul,"j\sqcup l"']\ar[dl,"j'\sqcup l'"]\\
	\& M'\sqcup N' \&
	\end{tikzcd}
	\]
	and using the universal property of the coproduct on the appropriate homotopies it is straightforward to check that $\phi\sqcup \psi$ is a homotopy equivalence with homotopy inverse $\phi'\sqcup \psi'$.
	
	We now check that $\otimes$ is a functor, starting with checking that $\otimes$ preserves identities. Let $X$ and $Y$ be any spaces. The canonical isomorphism  $(X\sqcup Y)\times \II \to (X\times \II )\sqcup (Y\times \II)$, which in particular is a homotopy equivalence, is sufficient to show that
	$\cmor{(X\times \II)\sqcup (Y\times \II)}{\iota_0^X\sqcup \iota^Y_0}{X\sqcup Y}{\iota_0^X\sqcup \iota_1^Y}{X\sqcup Y}$
	is cospan homotopy equivalent to $\cmor{(X\sqcup Y)\times \II}{\iota_0^{X\sqcup Y}}{X\sqcup Y}{\iota_1^{X\sqcup Y}}{X\sqcup Y}$.
	
	Finally we check that $\otimes$ preserves composition. Given two pairs of composable \ccc{}s, there are distinct cospans obtained from first appplying $\otimes$ and then composing and from composing and then applying $\otimes$.
	A commuting isomorphism is constructed between these cospans using the universal properties of the coproduct and the pushout.
\end{proof}	

To construct a monoidal structure on $\CofCsp$, we will give all associators and unitors in the form of the cospan in the following lemma.

\lemm{\label{le:monoidal_isos}
	Let $X$ and $X'$ be spaces and $f\colon X\to X'$ a homeomorphism.
	Then the cospan 
	\[
	\begin{tikzcd}[ampersand replacement=\&,column sep=small, row sep=small]
	X \ar[rd,"\iota_0^{X'}\circ f"',near start] \& \& X' \ar[dl,"\iota_1^{X'}",near start]\\
	\& X'\times \II 
	\end{tikzcd}
	\]
	is a \ccc{} and its \che{} class is an isomorphism in $\CofCsp$.\\
	(Recall that $\iota^X_a\colon X \to X \times \II$ is the map $x\mapsto (x,a)$.)
}
\begin{proof}
	We first prove that the cospan is a \ccc{}.
	Note that the map $\lan \iota_0^{X'}\circ f, \iota_1^{X'}\ran$ is equal to the composition
	\[
	X\sqcup X'
	\xrightarrow{\lan f,\id_X\ran}
	X'\sqcup X'
	\xrightarrow{\lan \iota_0^{X'}, \iota_1^{X'}\ran}
	(X'\sqcup X')\times \II.
	\]
	The first map is a homeomorphism; hence it is a cofibration by Lemma~\ref{le:homeo_cofib}. 
	We proved that the second map is a cofibration in  Lemma~\ref{pr:identity_ccc}.
	Hence, by Lemma~\ref{le:cofib_comp}, the composition is a cofibration. Since the first map is a homeomorphism, the image of the composition is equal to the image of the second map, so is closed by Lemma~\ref{pr:identity_ccc}.
	
	To see that the \che{} class is an isomorphism notice that the composition 
	\[
	\left[ \cmortikz{X'\times \II}{\iota_0^{X'}\circ f}{X}{\iota_1^{X'}}{X'} \right]
	\sbullet \left[ \cmortikz{X'\times \II}{\iota_0^{X'}}{X'}{\iota_1^{X'}\circ f	}{X}\right] 
	\]
	is equivalent to $\cmor{X'\times \II}{\iota_0^{X'}\circ f}{X}{\iota_1^{X'}\circ f}{X}$ 
	via the obvious isomorphism $X'\times \II\cong (X'\times \II)\sqcup_{X'}(X'\times \II)$, which is equivalent to $\cmor{X\times \II}{\iota_0^X}{X}{\iota_1^X}{X}$ via the homeomorphism $f\times \id_{\II}\colon X\times \II \to X'\times \II$.
\end{proof}

\begin{lemma}\label{le:CofCsp_mon}
	Recall from Proposition~\ref{pr:Top_mon} that $(\Topo,\sqcup,\emptyset,\alpha_{X,Y,Z}^T,\lambda_X^T,\rho_X^T)$ is a monoidal category.
	There is a monoidal category \[(\CofCsp\;,\;\otimes\;,\;\emptyset\;,\;\alpha_{X,Y,Z}\;,\;\lambda_{X}\;,\;\rho_{X})\] where $\otimes$ is as in Lemma~\ref{le:CofCsp_bifunctor},  
	\begin{itemize}
		\item for any spaces $X,Y,Z\in Ob(\CofCsp)$, $\alpha_{X,Y,Z}\colon (X\sqcup Y)\sqcup Z \to X \sqcup (Y\sqcup Z)$ is the \che{} of the cospan  
		\[
		\begin{tikzcd}[ampersand replacement=\&,column sep = -2em]
		(X\sqcup Y)\sqcup Z \ar[rd,"\iota_0^{X\sqcup (Y\sqcup Z)}\circ\alpha^{T}_{X,Y,Z}\hspace*{-1em}"',near start] \& \& X \sqcup (Y\sqcup Z) \ar[dl,"\iota_1^{X\sqcup (Y\sqcup Z)}",near start]\\
		\& (X\sqcup (Y\sqcup Z))\times \II,
		\end{tikzcd}
		\]
		\item for any space $X\in Ob(\CofCsp)$,
		$\lambda_X\colon \emptyset \sqcup X \to X$ is the \che{} class of the cospan
		\[
		\cmortikz{X\times \II,}{\iota_0^X\circ\lambda_X^T}{\emptyset\sqcup X}{\iota_1^X}{X}
		\]
		\item for any space $X\in Ob(\CofCsp)$,
		$\rho_X\colon X\sqcup \emptyset \to X$ is the \che{} class of the cospan 
		\[
		\cmortikz{X\times \II.}{\iota_0^X\circ\rho_X^T}{X\sqcup \emptyset}{\iota_1^X}{X}
		\]
	\end{itemize}
\end{lemma}
\begin{proof}
	First note that Lemma~\ref{le:monoidal_isos} gives that all associators and unitors are isomorphisms.\\
	The proofs of each of (M1-4)  
	are similar, so we only give the proof of (M2) here.\\	
	We must construct a \che{} from the composition
	\[
	\begin{tikzcd}[ampersand replacement=\&,column sep=scriptsize]
	(X\sqcup \emptyset)\sqcup Y
	\ar[dr,"\iota_0^{X\sqcup (\emptyset \sqcup Y)}\circ \alpha^T_{X,\emptyset,Y}"'] \&
	\&
	X\sqcup (\emptyset \sqcup Y)
	\ar[dl,"\iota_1^{X\sqcup (\emptyset \sqcup Y)}"]
	\ar[dr,"\iota_0^X\sqcup(\iota_0^Y\circ \lambda^T_Y)"'] \& \& \makebox*{$X\sqcup (\emptyset \sqcup Y)$}{$X\sqcup Y$}
	\ar[dl,"\iota_1^X\sqcup \iota_1^Y"] \\
	\& (X\sqcup (\emptyset \sqcup Y))\times \II\ar[dr] \& \& 
	(X\times \II )\sqcup (Y\times \II)\ar[dl]\\
	\& \& \makebox*{$X\sqcup (\emptyset \sqcup Y)$}{$((X\sqcup (\emptyset \sqcup Y))\times \II) \sqcup_{X\sqcup (\emptyset \sqcup Y)} ((X\times \II )\sqcup (Y\times \II))$,}
	\end{tikzcd}
	\]
	to the cospan
	\[
	\begin{tikzcd}[ampersand replacement=\&]
	(X\sqcup \emptyset)\sqcup Y \ar[dr,"(\iota_0^X\circ \rho^T_X)\sqcup \iota_0^Y"'] \& \& X\sqcup Y\ar[dl,"\iota_1^X\sqcup \iota_1^Y"] \\
	\& (X\times \II )\sqcup (Y\times \II ).
	\end{tikzcd}
	\]
	By the universal property of the coproduct and Lemma~\ref{le:producthom}, a map 
	$f\colon (X\times \II)\sqcup ((\emptyset \times \II)\sqcup (Y\times \II))\to (X\times \II )\sqcup (Y\times \II)$
	is uniquely determined by 
	\ali{
		f_X\colon X\times \II&\to (X\times \II )\sqcup (Y\times \II)\\
		(x,t)&\mapsto ((x,t/2),1)
	}
	and 
	\ali{
		f_Y\colon Y\times \II&\to (X\times \II )\sqcup (Y\times \II)\\
		(y,t)&\mapsto ((y,t/2),2).
	}
	Similarly a map $g\colon(X\times \II)\sqcup (Y\times \II)\to (X\times \II )\sqcup (Y\times \II)$ is determined by the pair
	\ali{
		g_X\colon X\times \II&\to (X\times \II )\sqcup (Y\times \II)\\
		(x,t)&\mapsto ((x,1/2(t+1)),1)
	}
	and 
	\ali{
		g_Y\colon Y\times \II&\to (X\times \II )\sqcup (Y\times \II)\\
		(y,t)&\mapsto ((y,1/2(t+1)),2).
	}
	We have that $f\circ \iota_1^{X\sqcup(\emptyset\sqcup Y)}= g\circ \iota_0^X\sqcup(\iota_0^Y\circ\lambda^T_Y)$ commute, so by the universal property of the pushout, these maps determine a map
	\[
	h\colon ((X\sqcup (\emptyset \sqcup Y))\times \II) \sqcup_{X\sqcup (\emptyset \sqcup Y)}( (X\times \II )\sqcup (Y\times \II))\to (X\times \II )\sqcup (Y\times \II)
	\]
	which is a homeomorphism, and it is straightforward to check this commutes with the cospans, hence is a \che{}.
\end{proof}

\begin{theorem}\label{th:CofCsp_symm}
	There is a symmetric monoidal category
	\[
	(\CofCsp\;,\;\otimes\;,\;\emptyset\;,\;\alpha_{X,Y,Z}\;,\;\lambda_{X}\;,\;\rho_{X}\;,\;\beta_{X,Y})
	\] 
	where $(\CofCsp,\otimes,\emptyset,\alpha_{X,Y,Z},\lambda_{X},\rho_{X})$ is as in Lemma~\ref{le:CofCsp_mon}, and for any spaces $X,Y\in Ob(\CofCsp)$,
	$\beta_{X,Y} \colon X\otimes Y \to Y\otimes X $ is the \che{} class of the cospan
	\[ 
	\begin{tikzcd}[ampersand replacement=\&]
	Y\sqcup X \ar[dr,"\iota_0^{Y\sqcup X}\circ \beta^T_{X,Y}"'] \& \& X\sqcup Y \ar[dl,"\iota_1^{Y\sqcup X}"] \\
	\& (Y\sqcup X)\times \II
	\end{tikzcd}
	\]
	where $\beta^T_{X,Y}$ is the braiding in $\Topo$ as in Proposition~\ref{pr:Top_braid}.\\
	By abuse of notation we will refer to this symmetric monoidal category as $\CofCsp$.
\end{theorem}
\begin{proof}
	As with the previous theorem, the proofs of all necessary identities are similar.
	Here we give the proof that $\beta$ is symmetric.
	
	We must construct a \che{} from the composition
	\[
	\begin{tikzcd}[ampersand replacement=\&]
	X\sqcup Y\ar[dr,"\iota_0^{Y\sqcup X}\circ \beta^T_{X,Y}"'] \& \& Y\sqcup X\ar[dl,"\iota_1^{Y\sqcup X}"]
	\ar[dr,"\iota_0^{X\sqcup Y}\circ \beta^T_{Y,X}"' {xshift=1em}] \& \& X\sqcup Y  \ar[dl,"\iota_1^{X\sqcup Y}"]\\
	\& (Y\sqcup X) \times \II \ar[dr]\& \&
	(X\sqcup Y) \times \II \ar[dl]\\
	\&\& \makebox*{$Y\sqcup X$}{$((Y\sqcup X) \times \II )\sqcup_{Y\times X}((X\sqcup Y) \times \II)$,}
	\end{tikzcd}	
	\]
	to the cospan 
	\[
	\begin{tikzcd}[ampersand replacement=\&, sep =small]
	X\sqcup Y\ar[dr,"\iota_0^{X\sqcup Y}"'] \& \& X\sqcup Y\ar[dl,"\iota_1^{X\sqcup Y}"] \\
	\& (X\sqcup Y)\times \II.
	\end{tikzcd}
	\]
	Define maps
	\ali{
		f_1\colon(Y\sqcup X) \times \II &\to (X\sqcup Y)\times \II\\
		(x,t)&\mapsto (\beta^T_{Y, X}(x),t/2)
	}
	and
	\ali{
		f_2\colon (X\sqcup Y)\times \II &\to (X\sqcup Y)\times \II \\
		(x,t)&\mapsto (x,1/2(t+1)) .
	}
Note that $f_1\circ \iota_1^{Y\sqcup X}=f_2\circ (\iota_0^{X\sqcup Y}\circ \beta^T_{Y,X})$, hence applying the universal property of the pushout 
	determines a map
	\[
	f\colon ((Y\sqcup X) \times \II )\sqcup_{Y\times X}((X\sqcup Y) \times \II)\to (X\sqcup Y)\times \II.
	\]
	Notice that $f$ is a homeomorphism, and it is straightforward to check that it commutes with the cospans, and so is a \che{}.
\end{proof}

\rem{Denote by $\Topo^h$ is the wide subcategory of $\Topo$ where all maps are homeomorphisms.
	There is a functor $\kappa\colon \Topo^h \to \CofCsp$, which sends a homeomorphism $f\colon X\to Y$ to the \che{} class of the cospan $\cmor{Y\times \II}{\iota_0^Y\circ f}{X}{\iota^Y_1}{Y}$.
	It then follows that (M1-2), and the identities required for braiding and symmetry, commute in $\CofCsp$ as they are precisely the images of the corresponding identities in $\Topo$.
	The same construction leads to a map from the classical mapping class group of a manifold into $\CofCsp$, this is why TQFTs give representations of mapping class groups.
	A related mapping leads to a functor from the mapping class groupoid of a space $X$ into $\CofCsp$. 
	See Section~\ref{sec:funcHomcob} for more.
}

\subsection{Category of homotopy cobordisms \texorpdfstring{$\HomCob$}{HomCob}}

Here we construct the symmetric monoidal category $\HomCob$ (Theorem~\ref{th:HomCob}), which we will use as the source category of the TQFT we construct in Section~\ref{sec:tqft}.
We obtain $\HomCob$ as a subcategory of $\CofCsp$ with a finiteness condition on spaces.

\begin{definition}\label{de:homfin}
	A space $X$ is called {\em \homfin{}} if $\pi(X,A)$ (Def.~\ref{de:Afundgrpd}) is finitely generated for all finite sets of basepoints $A$.\\
	Let $\chi$ denote the class of all \homfin{} spaces.
\end{definition}

\begin{lemma}\label{le:submgmoid_HomCob}
	There exists a submagmoid
	\[
	\cHomCob=(\chi,\cHomCob(-,-),\sbullet)
	\]
	of $\cCofCsp$
	where
	\[
	\cHomCob(X,Y)=\left\{ \cmortikz{M}{i}{X}{j}{Y}\; \middle\vert \;  \parbox{22em}{$\lan i , j \ran$ is a closed cofibration, and \\
	$X,Y$ and $M$ are \homfin{}}
\right\}.
	\]
	Morphisms in $\cHomCob$ are called {\em \chomcob{}s}.
\end{lemma}
\begin{proof}
	We check $\cHomCob$ is closed under composition. 
	Suppose $\cmor{M}{i}{X}{j}{Y}$ and $\cmor{N}{k}{Y}{l}{Z}$ are \chomcob{}s.
	Consider the pushout
	\[
	\begin{tikzcd}[ampersand replacement= \&,column sep=tiny]
	\& \& Y\ar[dl,"j"']\ar[dr,"k"] \& \&  \\
	\&  M\ar[dr,"p_M"'] \& \& N \ar[dl,"p_N"]\& \\
	\& \& M\push{$p_Mp_N$}{Y} N. \& \& 
	\end{tikzcd}
	\]
	We may choose finite representative subsets $Y_0\subseteq Y$, $M_0\subseteq M$ and $N_0\subseteq N$ such that $j(Y_0)= M_0\cap j(Y)$ and 
	$k(Y_0)= N_0\cap k(Y)$.
	Applying Corollary~\ref{co:vk}
	the following square is also a pushout.
	\[
	\begin{tikzcd}[ampersand replacement= \&,column sep=tiny]
	\& \& \pi(Y,Y_0)\ar[dl,"\pi(j)"']\ar[dr,"\pi(k)"] \& \&  \\
	\&  \pi(M,M_0)\ar[dr,"\pi(p_M)"'] \& \& \pi(N,N_0) \ar[dl,"\pi(p_N)"]\& \\
	\& \&  \pi(M\push{$p_Mp_N$}{Y} N,M_0\push{$p_Mp_N$}{Y_0} N_0) \& \& 
	\end{tikzcd}
	\]
	\sloppypar{\noindent
		We have, from  Theorem~\ref{th:grpd_pushoutfg}, that the pushout of finitely generated  groupoids is finitely generated, so that $\pi(M\push{$p_Mp_N$}{Y} N,M_0\push{$p_Mp_N$}{Y_0} N_0)$ is finitely generated follows from the fact that $\pi(M,M_0)$ and $\pi(N,N_0)$ are.
		Hence the composition is a \chomcob{}.}
\end{proof}

To give specific examples of \chomcob{} we will use the following result which says that, to check a space $X$ is \homfin{}, it will be sufficient to find a single representative subset $A\subseteq X$ such that  $\pi(X,A)$ is finitely generated.

\begin{lemma}\label{le:fg_all_sets}
	If $\pi(X,A)$ is finitely generated for some finite representative set $A$, then $\pi(X,A')$ is finitely generated for all finite representative sets $A'$. 
\end{lemma}
\begin{proof}
	Let $A=\{a_1,...,a_N\}\subset X$ be a finite representative subset, and suppose $\pi(X,A)$ is finitely generated. 
	Then there exists a finite set of generating morphisms. Let $S=\{s_1,...,s_K\}$ be a finite set of representative paths, such that taking path-equivalence classes of each path gives a set of generating morphisms for $\pi(X,A)$.
	
	Let $B=\{b_1,...,b_M\}\subset X$ be another finite representative subset. 
	For each pair $\{n,m\}$ such that $a_n$ and $b_m$ are in the same path connected component, choose a path $\gamma_{n,m} \colon a_n \to b_m$. We denote the set of all such paths by $\Gamma$.
	Note that this is finite since $A$ and $B$ are finite.
	We will show that $\pi(X,B)$ is generated by the set of path-equivalence classes of all paths of the form $\gamma_{n',m'} s\gamma_{n,m}^{-1}$, where $s\in S$, $s_0=a_n$ and $s_1=a_{n'}$. Note this is again finite. 
	
	Let $t:b_m\to b_{m'}$ be any path.
	Choose $n$ such that $a_n$ is in the same path connected component as $b_m$, then  $\tilde{t}=\gamma^{-1}_{n,m'}t \gamma^{\phantom{-1}}_{n,m}\colon a_n\to a_n$ is a path and 
	$\gamma^{\phantom{-1}}_{n,m'}\tilde{t} \gamma^{-1}_{n,m}\simp t$.
	Now we have $\tilde{t}\sim p_L\dots p_2p_1$, a finite sequence of $p_l \in S$, since the equivalence classes of the $s_k$ generate $\pi(X,A)$. Hence $t\simp\gamma_{n,m'}^{\phantom{-1}} p_L...p_2p_1\gamma_{n,m}^{-1}$.
	For each $p_l$, choose a path $\gamma_{p_l}\in \Gamma$ such that $(\gamma_{p_l})_0=(p_l)_1$, so  $\gamma^{-1}_{p_l}\gamma_{p_l}^{\phantom{-1}}\simp e_{(p_l)_1}$, the identity path. 
	Now $t\simp  
	\gamma_{n,m'}^{\phantom{-1}}p_L\gamma_{p_{L-1}}^{-1}...\gamma_{p_2}^{\phantom{-1}}p_2\gamma_{p_1}^{-1}\gamma_{p_1}^{\phantom{-1}}p_1\gamma_{n,m}^{-1}$, which is of the desired form.
\end{proof}

\exa{\label{ex:disk} The \ccc{} from Proposition~\ref{disk} is a \chomcob, as the fundamental groups of $D^2$ and of $S^1$ are finitely generated.}

\exa{\label{ex:mer_b} The \ccc{}, $\cmor{M}{i}{X}{j}{Y}$, in Example~\ref{ex:mer} is a \chomcob{}.
	We have $X\cong S^1\sqcup S^1$, hence, letting $X_0$ be a subset with a single point in each path connected component, $\pi(X,X_0)\cong \Z \sqcup \Z$.
	Similarly $Y\cong S^1$, so $\pi(Y,\{y\})\cong \Z$ for any $y\in Y$.
	The manifold $M$ is a homotopy equivalent to the twice punctured disk, hence has fundamental group $\pi(M,\{m\})\cong \Z*\Z$.
	Hence $X$, $Y$ and $M$ are \homfin{}.
}
\exa{\label{ex:embmer_b} 
	The \ccc{}, $\cmor{M}{i}{X}{j}{Y}$, in Example~\ref{ex:embmer} is a \chomcob{}. 
	The space $X$ is homotopy equivalent to the disjoint union of two copies of the open disk and a twice punctured disk.
	Thus, choosing $X_0\subset X$ with a point in each connected component, we have $\pi(X,X_0)$ is finitely generated.
The space $Y$ is homotopy equivalent to the disjoint union of the circle and the open disk thus, choosing $Y_0$ in the same way, we have $\pi(Y,Y_0)$ finitely generated.
	The space $M$ is the disjoint union of a contractible space, and a space which is homotopy equivalent to a $3$ punctured sphere 
	and thus via a stereographic projection, homotopy equivalent to the twice punctured disk. 	Hence $\pi(M,M_0)$ is finitely generated
	for any choice of $M_0\subset M$ finite. \\
	To see the homotopy equivalence to the punctured sphere, notice that a ball with the centre removed is equivalent to the sphere, and this can be seen by choosing a homotopy equivalence which sends each point in the ball to the closest point on the sphere. Now consider the ball with three line segments removed, each with one end point on the boundary and the other at the centre of the ball. The same homotopy equivalence sends this space to the $3$ punctured sphere. Notice that this extends to a homotopy gradually pushing points to the boundary, and at some point during this homotopy, the image will be a space homeomorphic to the space represented in Figure~\ref{fig:embmer}.
}

\begin{example}
	Let $\Gamma$ be a finite graph.
	Choose disjoint sets $V_1,V_2\subseteq V(\Gamma)$ of vertices.
	Then 
	$\cmor{\Gamma}{i}{V_1}{j}{V_2}$ is a concrete homotopy cobordism where $i$ and $j$ are inclusions.
	To see that $\lan i,j\ran $ is a closed cofibration, we can think of $\Gamma$ as a CW complex, and $V_1\cup V_2$ a subcomplex, and then apply Proposition~\ref{pr:CWcofib}.
	That the spaces are \homfin{} can be seen by taking the set of basepoints to be all vertices, and generating paths to be edges.
\end{example}

\begin{example}
	Let $M$ be a CW complex, and $X$ and $Y$ disjoint subcomplexes. Then $\cmor{M}{i}{X}{j}{Y}$, where $i$ and $j$ are inclusions, is a \chomcob{}.
	Proposition~\ref{pr:CWcofib} gives that $\lan i, j\ran$ is a closed cofibration. 
That finitely generated CW complexes have finite fundamental group follows from Proposition 1.26 of \cite{hatcher}.
\end{example}

\begin{definition}
	A \cc{} is called a {\em \homcob{}} if there exists a representative which is a \chomcob.\\
	For \homfin{} spaces $X,Y\in \Topo$ define
	\[
	\HomCob(X,Y)=\left\{\text{ }\classche{\cmortikz{M}{i}{X}{j}{Y}}
	\; \middle\vert \;\cmortikz{M}{i}{X}{j}{Y} \text{  is a \chomcob{}}   \right\}.
	\]
\end{definition}
Notice that if $\cmor{M}{i}{X}{j}{Y}$ is a \ccc{} with all spaces \homfin{}, then it is clear from the definition of \che{} that every cospan in the equivalence class also has all spaces \homfin{}. 

\begin{theorem}\label{HomCob}
	There is a subcategory of $\CofCsp$ (Thm~\ref{th:CofCsp})
	\[
	\HomCob=\left(\chi,\HomCob(X,Y),\; \sbullet \;, \; \classche{\cmortikz{X\times \II}{\iota_{0}^X}{X}{\iota_1^X}{X}}\;\right)
	\]
	with
	\begin{itemize}
		\item all \homfin{} spaces as objects;
		\item for spaces $X,Y\in Ob(\HomCob)$, 
		morphisms in $\HomCob(X,Y)$ are homotopy cobordisms i.e. cospan homotopy equivalence classes (Lem~\ref{le:che}) of cospans 
		\[
		\classche{\cmortikz{M}{i}{X}{j}{Y}}
		\]
		with all spaces \homfin{} and $\lan i, j \ran$ a closed cofibration;
		\item composition is as follows
		\ali{
			\sbullet\;\colon \HomCob(X,Y)\times \HomCob(Y,Z)&\to \HomCob(X,Z) \\
			\left({\classche{\cmortikz{M}{i}{X}{j}{Y}}}\;\raisebox{-1em}{,}\;{\classche{\cmortikz{N}{k}{Y}{l}{Z}}}\right)&\mapsto \classche{\cmortikz{M \sqcup_Y N}{\tilde{i}}{X}{\tilde{l}}{Z}}
		}
		where $\tilde{i}=p_Mi$ and $\tilde{l}=p_Nl$ are obtained via the pushout diagram
		\[
		\begin{tikzcd}[ampersand replacement= \&,column sep=tiny]
		\makebox*{$M\sqcup_Y N$}{$X$} \ar[dr,"i"']\& \& Y\ar[dl,"j"]\ar[dr,"k"'] \& \& \makebox*{$M\sqcup_Y N$}{$Z$}\ar[dl,"l"] \\
		\&  M\ar[dr,"p_M"'] \& \& N \ar[dl,"p_N"]\& \\
		\& \&  M\sqcup_{Y} N; \& \& 
		\end{tikzcd}
		\]
		\item for a space $X\in Ob(\HomCob)$ the identity morphism is the equivalence class of the cospan
		\[
		\cmortikz{X\times \II.}{\iota_{0}^X}{X}{\iota_1^X}{X}
		\] 
	\end{itemize}
\end{theorem}

\begin{proof}
	We have from Lemma~\ref{le:submgmoid_HomCob} that $\HomCob=(\chi, \HomCob(X,Y),\sbullet)$ is a magmoid, in particular the composition is closed.
	It remains only to prove that all identities are in $\HomCob$.
	
	Let $X$ be a \homfin{} space.
	Then $X\times \II$ is homotopy equivalent to $X$, and so $[\cmor{X\times \II}{\iota_0^X}{X}{\iota_1^X}{X}]$
	is a \homcob{}.	
\end{proof}

\begin{proposition}\label{pr:fun_cob}
	Let $\mathbf{Cob}(n)$ be the category where objects $(n-1)$-dimensional closed oriented smooth manifolds and morphisms are equivalence classes of concrete cobordisms (Definition~\ref{de:cobord}) as in \cite[Ch.~1]{lurie}.
	For all $n\in \N$ there is a functor \[\Cob{n}\colon\mathbf{Cob}(n) \to \HomCob
	\]
	 which maps objects to their underlying space and maps a morphism to the equivalence class of the \ccc{} which is the image of a representative cobordism under the mapping described in Proposition~\ref{pr:concob_to_concsp}.
\end{proposition}
\begin{proof}
	We first check that $\Cob{n}$ is well defined.
	Chapter 6 of \cite{hirsch} proves that compact smooth manifolds have the homotopy type of finite CW complexes (see in particular the start of Section~3 and Theorems~1.2 and 4.1).
	If we choose the set of basepoints to be the $0$-cells of the corresponding CW complex, then the generators of the fundamental groupoid are the $1$-cells and so the fundamental groupoids of smooth manifolds with a finite set of basepoints are finitely generated.
	If two concrete cobordisms are equivalent up to boundary preserving diffeomorphism then they are certainly equivalent up to cofibre homotopy equivalence using the same map. So we have that the functor is well defined.
	
	Let $X,Y,Z$ be a triple of objects in $\mathbf{Cob}(n)$ and  $M\colon X\to Y$, $M'\colon Y\to Z$ a pair of cobordisms. 
	Then we have maps $\phi\colon X\sqcup Y \to M$ and $\phi'\colon Y\sqcup Z\to M'$ between the underlying topological spaces.
	The image of the composition in $\mathbf{Cob}(n)$  is the cospan
	$\cmor{M\sqcup N/((y,0)\sim(y,1))}{i}{X}{j}{Z}$ 
	where $i(x)=\phi(x,0)$ and $j(y)=\phi'(z,1)$.
	This is precisely the composition of the images of $M\colon X\to Y$ and $M'\colon Y\to Z$ in $\HomCob$.
	
	The identity for a manifold $X$ in $\mathbf{Cob}(n)$ is represented by the cylinder $X\times \II$ with $\lan \iota^X_0 , \iota_1^X \ran \colon  \bar{X}\sqcup X\to X\times \II$, this clearly maps to a representative of the identity cospan of $X$.
\end{proof}

\subsubsection{Monoidal structure on \texorpdfstring{$\HomCob$}{}}\label{sec:monHomCob}
The category $\HomCob$ becomes a symmetric monoidal category, just like $\CofCsp$.

\medskip

We will need the following result about \homfin{} spaces to ensure $\otimes$ restricts to a closed composition in $\HomCob$.
\begin{lemma}\label{union_homfin_spaces}
	If $X$ and $Y$ are \homfin{} spaces, then $X\sqcup Y$ is \homfin{}.
\end{lemma}
\begin{proof}
	Suppose $X_0$ and $Y_0$ are finite representative subsets of $X$ and $Y$ respectively.
	The images of $X$ and $Y$ in $X\sqcup Y$, under the maps into the coproduct, are disjoint, hence there is an isomorphism $\pi(X\sqcup Y,X_0 \sqcup Y_0)\cong\pi(X,X_0)\amalg \pi(Y,Y_0)$ of groupoids given by sending a path equivalence class $[\gamma]$ to $([\gamma],1)$ if $\gamma$ is a path in $X$ and to $([\gamma],2)$ if $\gamma$ is a path in $Y$.
	By Theorem~\ref{th:grpd_pushoutfg} we have that $\pi(X,X_0)\amalg \pi(Y,Y_0)$ is finitely generated if and only if $\pi(X,X_0)$ and $\pi(Y,Y_0)$ are.
	Notice all finite representative subsets of $X\sqcup Y$ are of the form $X_0\sqcup Y_0$, so this is sufficient.
\end{proof}

\begin{theorem}\label{th:HomCob}
	There is a symmetric monoidal subcategory
	\[(\HomCob\;,\;\otimes\;,\;\emptyset\;,\;\alpha_{X,Y,Z}\;,\;\lambda_{X}\;,\;\rho_{X}\;,\;\beta_{X,Y})
	\]
	of $\CofCsp$.
	Here $\otimes$ is as in Lemma~\ref{le:CofCsp_bifunctor}, the $\alpha_{X,Y,Z}$, $\lambda_{X}$, $\rho_{X}$ are as in Lemma~\ref{le:CofCsp_mon} and the $\beta_{X,Y}$ are  as in Lemma~\ref{th:CofCsp_symm}.
\end{theorem}
\begin{proof}	
	The empty set is \homfin{}.
	For each pair of \homfin{} spaces, the disjoint union is \homfin{} by Lemma~\ref
	{union_homfin_spaces} so $\otimes$ sends a pair of \homcob{}s to a \homcob{}.
	
	Using again Lemma~\ref{union_homfin_spaces} along with the fact that for any space $X$ and finite representative $A\subseteq X$ we have $\pi(X,A)\cong\pi(X\times \II, A\times \{0\})$, and $A\times \{0\}$ is representative in $X\times \II$, thus the $\alpha_{X,Y,Z}$, $\lambda_{X}$, $\rho_{X}, \beta_{X,Y}$ are all in $\HomCob$.
\end{proof}

\begin{proposition}\label{cob_to_HomCob}
	\sloppypar{
	The functor $\Cob{n}\colon \mathbf{Cob}_n\to \HomCob$ as in Proposition~\ref{pr:fun_cob} is a symmetric strong monoidal functor with $(\Cob{n})_0=\classche{\cmor{\emptyset}{\emptyset}{\emptyset}{\emptyset}{\emptyset}}$ and $(\Cob{n})_2(X,Y)=\classche{\cmor{(X\sqcup Y)\times \II}{\iota_0^{X\sqcup Y}}{X\sqcup Y}{\iota_1^{X\sqcup Y}}{X\sqcup Y}}$. }
\end{proposition}
\begin{proof}
	The monoidal product $\otimes'$ in $\mathbf{Cob}(n)$ is given by disjoint union, thus we have $\otimes \circ(\Cob{n} \times \Cob{n})=\Cob{n} \circ \otimes'$ and thus the appropriate identity is the required natural transformation.
	It is straightforward to check all necessary identities as $(\Cob{n})_0$ and $(\Cob{n})_2$ are identity morphisms, and the functor $\Cob{n}$ maps all associators, unitors and braidings to exactly the corresponding associators, unitors and braidings in 
	$\HomCob$.
\end{proof}

\section{Motion groupoids, mapping class groupoids and homotopy cobordisms
}\label{sec:funcHomcob}

In this section we construct functors into $\HomCob$ from (subgroupoids of) the motion groupoid $\Mot{M}{A}$ and from the  mapping class groupoid $\mcg{M}{A}$ of a manifold $M$ and a subset $A\subset M$, as constructed in \cite{motion}.
Thus functors from $\HomCob$ into $\Vect_{\mathbb{C}}$ give representations of motion groupoids and mapping class groupoids.

 We briefly recall each of these constructions below.
 For a point of reference note that for 
 $K\subset D^2$ a finite subset in the $2$-disk, then  $\Mot{D^2}{\partial D^2}(K,K)\cong \mcg{D^2}{\partial D^2}(K,K)$ are isomorphic to the Artin braid group on $|K|$ strands \cite{artin}, and similarly considering $L\subset B^3$, a configuration of unknotted, unlinked circles in the $3$-ball, we have $\Mot{B^3}{\partial B^3}(L,L)\cong \mcg{B^3}{\partial B^3}(L,L)$ are isomorphic to the corresponding loop braid group \cite{damiani}. 

\medskip

Let $\Topo^h\subset \Topo$  denote the subcategory of homeomorphisms. 
Then let  $\TOPO^h(X,Y)$ denote the space with underlying set $\Topo^h(X,Y)$ and the compact open topology (see e.g. \cite[Sec.2.4]{dieck}).
Let $X$ be a space and $A\subset X$ a subset, denote by  $\TOPO^h_A(X,X)\subset \TOPO^h(X,X)$ the subspace containing maps $f\colon X\to X$ such that $f(a)=a$ for all $a\in A$.
	For $X$ a set, $\Power X$ denotes the power set of $X$.

\subsection{Functor from the motion groupoid, 	\texorpdfstring{$\Motfunctor{M}{A}\colon \hfMot{M}{A} \to \HomCob$}{Motfunctor}}
We briefly recall the construction of the motion groupoid from \cite{motion}.
Let $\II =[0,1]\subset\R$ with the usual topology.
	Fixing a \axiomM{} $M$, and a subset $A\subset M$  
	$$
	\premo{M}{A} \; = \; \{ f \in   \Topo^{}(\II,\TOPO^h_A(M,M)) \; | \; f_0 = \id_M \,  \}  
	$$
- a set of gradual deformations of $M$ fixing $A$,
over some standard unit of time. 
	For	$(f,g)\in \premo{M}{A}\times \premo{M}{A}$
	note that $g*f$ given by
	\begin{align*}
		(g*f)_t = \begin{cases}
			f_{2t} & 0\leq t\leq 1/2, \\
			g_{2(t-1/2)}\circ f_1 & 1/2\leq t \leq 1.
		\end{cases}
	\end{align*}
is in $\premo{M}{A}$, as is  $\bar{f}$  given by 
		$ 
		\bar{f}_t 
		\;=\;  f_{(1-t)}\circ f_1^{-1} .
		$ 
	\hspace{.1cm}
	Thus $(\premo{M}{A},*)$ is a magma.

		Let
		$\Mtcmag^A \; =\; ( \Power M , \Mtcmag^A(-,-)  , * )$ be the magmoid constructed as follows. The 
		object set is the power set $\Power M$, and
		morphisms are triples $(f,N,N')$ consisting of an $f\in \premo{M}{A}$ together with a pair of spaces $N,N'\subseteq M$ and  such that $f_1(N)=N'$. We denote triples $(f,N,N')\in \Mtcmag^A$ by $\mot{f}{}{N}{N'}$.
		A morphism in $\Mtcmag^A$ is a gradual deformation of 
		$M$ that takes the initial object subset to the final object subset, hence called a {\em motion}.
				The {\em worldline} of a motion $\mot{f}{}{N}{N'}$ in a manifold $M$ is
				\[\W\left(\mot{f}{}{N}{N'}\right)=  \bigcup_{t \in [0,1]} f_t(N) \times \{t\} \subseteq M\times \II.
				\]
				There is a partial composition of motions given by $\mot{g}{}{N'}{N''}*\mot{f}{}{N}{N'}= \mot{g*f}{}{N}{N''}$.
	For $f,g\in \premo{M}{A}$ and $N,N'$ subsets,
	let
	\begin{multline*}
		\Topo(\II^2,\TOPO_A^h(M,M))\rel{f}{N}{N'}{g} 
		\;=\; 
		\{ H \in \Topo(\II^2,\TOPO_A^h(M,M)) \; | \; \forall t \; H(t,0)=f_t, \\
		H(t,1) = g_t, \forall s \; H(0,s)=\Id_M, H(1,s)(N)=N'=f_1(N) \}
	\end{multline*}

	\begin{theorem} \label{th:mg} (Thms.~3.49 and 4.6 of \cite{motion})
		Let $M$ be a \axiomM{}.
		The relation 
		$f\simrp g$ if 
		$\Topo(\II^2,\TOPO_A^h(M,M))\rel{f}{N}{N'}{g} \neq\emptyset$ 
		gives a congruence on the 
		magmoid $\Mtcmag^A$. 
		The quotient is a groupoid --- the {\em motion groupoid} of 
		$M$: 
		$$
		\Mot{M}{A}\;\; 
		=\;\;\; \Mtcmag^A/\simrp \;\;\;\; 
		=\;\; \;
		(\Power M,\; \Mtcmag^A(N,N')/\simrp,*,
		\classrp{\Id_M}, \; 
		\classrp{f} \mapsto \classrp{\bar{f}} )  .
		$$	
	\end{theorem}

\begin{definition}
Let $\hfMot{M}{A}\subset \Mot{M}{A}$ be the full subgroupoid with $N\in Ob(\hfMot{M}{A})$ if $M\setminus N$ is \homfin{}.
\end{definition}

\begin{definition}
Let $\Wc(\mot{f}{}{N}{N'})=(M\times \II)\setminus \left(\W\left(\mot{f}{}{N}{N'}\right)\right).
$
\end{definition}

Notice that $\Wc(\mot{f}{}{N}{N'})\cap (M\times \{t\}) =(M\times \{t\})\setminus(f_t(N)\times \{t\} )\cong M\setminus f_t(N)$, thus there is a well-defined, continuous map $\mott{f}\colon M\setminus f_t(N)\to \Wc(\mot{f}{}{N}{N'}) $, $m\mapsto (m,t)$.

\begin{theorem}\label{th:functorMot_HomCob}
	Let $M$ be a manifold. There is a well-defined functor
	\[
	\Motfunctor{M}{A}\colon \hfMot{M}{A} \to \HomCob\]
	which sends an object $N \in Ob(\hfMot{M}{A})$ to $M\setminus N$,
	and which sends a morphism $\classrp{\mot{f}{}{N}{N'}}$ to the \che{}
	class of 
	\[\cmortikz{\Wc(\mot{f}{}{N}{N'}) }{\motzero{f}}{M\setminus N}{\motone{f}}{M\setminus N'}
	\]
	where the maps $\mott{f}$  are as explained before the lemma.
	\end{theorem}
\begin{proof}
	We first check that $\Motfunctor{M}{A}$ is well defined.
	
	From \cite[Thm.3.72]{motion}, for any motion $\mot{f}{}{N}{N'}$, there exists a homeomorphism $\Theta(f)\colon M\times \II\to M\times \II$ given by $\Theta(f)(m,t)=(f_t(m),t)$, and notice that $\Theta(f)(N\times \II)=(\W(\mot{f}{}{N}{N'}))$, and hence the restriction of $\Theta(f)$ gives a homeomorphism $(M\setminus N)\times \II\cong \Wc(\mot{f}{}{N}{N'})$.
	
	Now suppose $\mot{f}{}{N}{N'}$ is a motion representing a morphism in $\hfMot{M}{A}$, then, by construction, $M\setminus N$ and $M\setminus N'$ are \homfin{}. Let $X\subset M\setminus N$ be a finite representative subset, then $X\times \{0\}$ is representative in $(M\setminus N)\times \II$.
	Then we have $\pi((M\setminus N)\times \II,X\times \{0\})\cong \pi(M\setminus N, X)\times \pi(\II, \{0\})\cong \pi((M\setminus N),X)$, where the first isomorphism follows from the fact that $\pi$ preserves products (\cite[6.4.4]{brownt+g}), and hence  $(M\setminus N)\times \II$ is \homfin{}.
Using that  $(M\setminus N)\times \II\cong \Wc(\mot{f}{}{N}{N'})$, it follows that $\Wc(\mot{f}{}{N}{N'})$ is \homfin{}. 
	
	We can see the map $\lan \motzero{f}, \motone{f} \ran \colon (M\setminus N)\sqcup (M\setminus N')\to \Wc(\mot{f}{}{N}{N'})$ is a cofibration as follows.
	The map $\lan \motzero{f}, \motone{f} \ran$ is equal to the 	 composition
	\begin{multline*}
	(M\setminus N)\sqcup (M\setminus N')\xrightarrow{\id_{M\setminus N}\sqcup f_1^{-1}} (M\setminus N)\sqcup (M\setminus N) \\ \xrightarrow{\lan \iota_0^{M\setminus N}, \iota_1^{M\setminus N}\ran} (M\setminus N)\times \II \xrightarrow{\Theta(f)|_{(M\setminus N)\times \II}}
	\Wc(\mot{f}{}{N}{N'}).
	\end{multline*}
	Now the first and last maps are homeomorphisms, hence cofibrations by Lemma~\ref{le:homeo_cofib}.
	That the middle map is a cofibration is proved in Proposition~\ref{pr:identity_ccc}. Finally, by 
	Lemma~\ref{le:cofib_comp}, the composition of cofibrations is a cofibration.

	Suppose $\mot{f,f'}{}{N}{N'}$ are representatives of the same morphism in $\hfMot{M}{A}(N,N')$.
	By Theorem~4.18 of \cite{motion} 
	this implies there is a {\it level preserving ambient isotopy} of $H\colon (M\times \II)\times \II\to M\times \II$ taking $\W(\mot{f}{}{N}{N'})$ to $\W(\mot{f'}{}{N}{N'})$ which fixes $(M\times \{0,1\})\cup (A\times \II)$ pointwise.
	The full definition of level preserving ambient isotopy can be found in Definition 4.12 of \cite{motion}, the property we will need is that $(m,t)\mapsto H(m,t,1)$ defines a homeomorphism of $M\times \II$, and in particular it is a homotopy equivalence.
	It follows that the restriction to $J\colon \Wc (\mot{f}{}{N}{N'})\to \Wc(\mot{f'}{}{N}{N'})$, defined by $(m,t)\mapsto H(m,t,1)$ is a well defined homeomorphism, and in particular is a homotopy equivalence, which fixes $\Wc (\mot{f}{}{N}{N'})\cap (M\times \{0\})$ and $\Wc (\mot{f}{}{N}{N'})\cap (M\times \{1\})$. Hence $\Motfunctor{M}{A}\classrp{\mot{f}{}{N}{N'}}= \Motfunctor{M}{A}\classrp{\mot{f}{}{N}{N'}}$. 

\medskip

	We now check that $\Motfunctor{M}{A}$ preserves composition. Let $\mot{f}{}{N}{N'}$ and $\mot{g}{}{N}{N'}$ be motions in $M$ representing composable morphisms in $\hfMot{M}{A}$.
	Then \[\Motfunctor{M}{A}(\classrp{\mot{g*f}{}{N}{N''}})=\classche{\cmor{\Wc(\mot{g*f}{}{N}{N''})}{\iota_{f_0}}{M\setminus N}{\iota_{f_1}}{M\setminus N''}}.\]
	By Lemma~3.38 of \cite{motion} 
	\begin{multline*}
	\W\left(g*f\colon N \too N''\right)\hspace{-0.12pt}=\hspace{-0.12pt}
	\big\{(m,t/2) \hspace{-1pt}\mid\hspace{-1pt} (m,t) \in \W(f\colon N \too N')\big \} \\ \cup  \big\{(m,(t+1)/2) \hspace{-1pt}\mid \hspace{-1pt}(m,t) \in \W(g\colon N' \too N'')\big \}
	\end{multline*}
	and hence 
	\begin{multline*}
	\W'\left(g*f\colon N \too N''\right)\hspace{-0.12pt}=\hspace{-0.12pt}
	\big\{(m,t/2) \hspace{-1pt}\mid\hspace{-1pt} (m,t) \in \W'(f\colon N \too N')\big \} \\ \cup  \big\{(m,(t+1)/2) \hspace{-1pt}\mid \hspace{-1pt}(m,t) \in \W'(g\colon N' \too N'')\big \}.
	\end{multline*}
	We have that $\Motfunctor{M}{A}(\classrp{\mot{f}{}{N}{N'}})\sbullet \Motfunctor{M}{A}(\classrp{\mot{g}{}{N}{N'}})$
	\begin{align*}
		&=\classche{\cmortikz{\Wc(\mot{g}{}{N'}{N''})\sqcup_{M\setminus N'} \Wc(\mot{f}{}{N}{N''})}{\pomotzero{f}}{M\setminus N}{\pomotone{g}}{M\setminus N''}}
	\end{align*}
The pushout is given by the set of elements of the form $((m,t),1)$ and $((m,t),2)$ quotiented by the relation $((m,1),1)\sim ((m,0),2)$. 
There is a map $\W'\left(g*f\colon N \too N''\right)\to \Wc(\mot{g}{}{N'}{N''})\sqcup_{M\setminus N'} \Wc(\mot{f}{}{N}{N''}) $ given by sending $(m,t)$ to the  $[((m,2t),1)]$ if $t\in [0,1/2]$ and $(m,t)$ to  $[((m,2t-1),2)]$ if $t\in [1/2,1]$. It is straightforward to check that this map commutes appropriately with the relevant maps, and it is clearly a homotopy equivalence. Thus we have proved that composition is preserved. 
\end{proof}

\subsection{Functor from the mapping class groupoid, \texorpdfstring{$\mcgfunctor{M}{A} \colon \hfmcg{M}{A} \to \HomCob$}{mcgfunctor}}

We briefly recall the construction of the mapping class groupoid $\mcg{M}{A}$, of a manifold $M$ together with a fixed subset $A\subset M$, from \cite[Thm.6.10]{motion}. The object set is the power set $\Power M$.
Morphisms in $\mcg{M}{A}(N,N')$ are equivalence classes of triples, denoted $\shmor{f}{}{N}{N'}$, consisting of a pair of subsets $N,N'\subseteq M$ and a self-homeomorphism $\sh{f}\in \TOPO_A^h(M,M)$ such that $\sh{f}(N)=N'$. 
Triples $\shmor{f}{}{N}{N'}$ and $\shmor{g}{}{N}{N'}$ are related if  
\begin{multline*}
\pathphisotop{M}{\gammaf}{N}{N'}{\gammaff}{A}
= \{ H \in \Topo(\II, \TOPO_A^h(M, M) \; | \; 
H(0)=\gammaf, H(1)=\gammaff, \\
\forall t\colon H(t)(N)=N'
\}
\end{multline*}
is non-empty. 
We denote the equivalence class as $\classi{\shmor{f}{}{N}{N'}}$
The partial composition is $(\classi{\shmor{f}{}{N}{N'}},\classi{\shmor{g}{}{N'}{N''}})\mapsto(\classi{\shmor{g\circ f}{}{N'}{N''}})$ where $\circ$ denotes function composition.
The inverse of a class $\classi{\shmor{f}{}{N}{N'}}$ is given by $\classi{\shmor{f^{-1}}{}{N}{N'}}$.

\begin{definition}
	Let $\hfmcg{M}{A}\subset \mcg{M}{A}$ be the full subgroupoid with $N\in Ob(\hfmcg{M}{A})$ if $M\setminus N$ is \homfin{}.
\end{definition}

For a triple $\shmor{f}{}{N}{N'}$ representing a morphism in $\mcg{M}{A}$,
 define $\mcgim{f}\colon M\setminus N \to (M\setminus N') \times \II $ to be the composition $\iota_0^{M\setminus N'}\circ \sh{f}|_{M\setminus N}\colon M\setminus N\to (M\setminus N')\times \II$,  $m\mapsto (\sh{f}(m),0)$. The map $\mcgim{f}$ is well defined since $\sh{f}(N)=N'$ and $\sh{f}$ is a homeomorphism, thus $\sh{f}(M\setminus N)=M\setminus N'$.
 
\lemm{\label{le:functormcg_HomCob}
	Let $M$ be a manifold. There is a functor 
\[
\mcgfunctor{M}{A} \colon \hfmcg{M}{A} \to \HomCob
\]
which sends an object $N\in Ob(\hfmcg{M}{A})$ to $M\setminus N$,
and which sends a morphism $\classi{\shmor{f}{}{N}{N'}}$ to the \che{}
class of 
\[\cmortikz{(M\setminus N')\times \II .}{\mcgim{f}}{M\setminus N}{\iota_1^{M\setminus N'}}{M\setminus N'}\]
}
\begin{proof}
	We first check that $\mcgfunctor{M}{A}$ is well defined.
	
	Suppose $\shmor{f}{}{N}{N'}$ a represents a morphism in $\hfmcg{M}{A}$, then, by construction, $M\setminus N$ and $M\setminus N'$ are \homfin{}.
	
	We now check that the map $\lan \mcgim{f}, \iota_1^{M\setminus N'} \ran \colon (M\setminus N)\sqcup (M\setminus N') \to (M\setminus N')\times \II$ is a cofibration. 
	The map $\lan \mcgim{f}, \iota_1^{M\setminus N'} \ran$ is equal to the composition
	\[
	(M\setminus N)\sqcup (M\setminus N') \xrightarrow{\sh{f}_{M\setminus N}\; \sqcup \; \id_{M\setminus N'}} (M\setminus N')\sqcup (M\setminus N') \xrightarrow{\lan \iota_0^{M\setminus N'} , \iota_1^{M\setminus N'} \ran} (M\setminus N')\times \II .
	\]
	Now the first map is a homeomorphism, thus a cofibration by Lemma~\ref{le:homeo_cofib}. The second is a cofibration by Proposition~\ref{pr:identity_ccc}. Using Lemmas~\ref{le:homeo_cofib} and \ref{le:cofib_comp}
 it follows that the composition is a cofibration.	
 	
	Suppose that $\shmor{f}{}{N}{N'}$ and $\shmor{f'}{}{N}{N'}$ are representatives of the same morphism in $\hfmcg{M}{A}$.
	We show $\classche{\mcgfunctor{M}{A}(\shmor{f}{}{N}{N'})}=\classche{\mcgfunctor{M}{A}(\shmor{f'}{}{N}{N'})}$ by constructing a homeomorphism from $(M\setminus N')\times \II$ to itself, which commutes with the appropriate cospans.
		It follows from the assumption that there exists a continuous map  $H\colon \II \to \TOPO^h_A(M,M)$ from $\sh{f}$ to $\sh{f}'$, satisfying $H(t)(N)=N'$ for all $t\in \II$. We denote $H(t)$ by $H_t$. 
	Using the product-hom adjunction (Lemma~\ref{le:producthom}), $H$ corresponds to a continuous map $H'\colon M \times \II \to \II$, $(m,t)\mapsto H_t(m)$.
	Define a map $\hat{H}\colon M\times \II\to M\times \II$ by $(m,t)\mapsto (H_{1-t}\circ \sh{f}^{-1}(m),t)$. It is straightforward to see that this map is continuous. 
	We show that $\hat{H}$ is a homeomorphism by giving a continuous inverse.
	Since $\TOPO^h(M,M)$ is a topological group (original result \cite[Th.4]{Arens}, see also \cite[Th.2.11]{motion})
	the map from $\TOPO^h(M,M)$ to itself sending each homeomorphism to its inverse is continuous.
	Thus $J\colon \II\to \TOPO^h(M,M)$ given by $t\mapsto H_t^{-1}$ is a continuous map. Applying the product-hom adjunction again, there is a continuous map $J'\colon M\times \II\to M$ given by $(m,t)\mapsto H_t^{-1}(m)$.
	Hence there exists a continuous map $\hat{J}\colon M\times \II \to M\times \II$ given by $(m,t)\mapsto (f\circ J'(m,1-t),t)=(f\circ H_{1-t}^{-1}(m),t)$. It is straightforward to see $\hat{J}\circ \hat{H}=\id_{M\times \II}=\hat{H}\circ \hat{J}$.
	Next notice that $\hat{H}(N'\times \II)=\cup_{t\in\II}H_{1-t}(f^{-1}(N'))\times \{t\}=\cup_{t\in\II}H_{1-t}(N)\times \{t\}=\cup_{t\in\II}N'\times \{t\}=N'\times \II$, thus $\hat{H}$ restricts to a homeomorphism $(M\setminus N')\times \II$.
	Finally we check the commutation relations.
	We have $\hat{H}( \mcgim{f}(m))=\hat{H}(\sh{f},0)=(H_{1}\circ \sh{f}^{-1}\circ \sh{f} (m) ,0)=(\sh{f}'(m),0)=\mcgim{f'}(m)$ and $\hat{H}( \iota_1^{M\setminus N'}(m))=\hat{H}(m,1)=(H_0\circ\sh{f}^{-1}(m)),1)=(m,1)=\iota_1(m)$.
	Thus $\hat{H}$ is a \che{} from $\shmor{f}{}{N}{N'}$ to $\shmor{f'}{}{N}{N'}$.

	\medskip
	
	We now check that $\mcgfunctor{M}{A}$ preserves composition.
	Let $\shmor{f}{}{N}{N'}$ and $\shmor{g}{}{N'}{N''}$ be representatives of composable morphisms in $\hfmcg{M}{A}$.
	\[
	\mcgfunctor{M}{A}(\classi{\shmor{g\circ f}{}{N}{N''}})
	=\cmor{(M\setminus N'')\times \II}{\mcgim{(g\circ f)}}{M\setminus N}{\iota_1^{M\setminus N''}}{M\setminus N''}
	\]
	and
	we have that $\mcgfunctor{M}{A}(\classi{\shmor{g}{}{N'}{N''}})\sbullet \mcgfunctor{M}{A}(\classi{\shmor{f}{}{N}{N'}})$ 
	\begin{align*}
		&=\classche{\cmortikz{(M\setminus N'')\times \II}{\mcgim{g}}{M\setminus N'}{\iota_1^{M\setminus N''}}{M\setminus N''}} \sbullet \; \classche{\cmortikz{(M\setminus N')\times \II}{\mcgim{f}}{M\setminus N}{\iota_1^{M\setminus N'}}{M\setminus N'}}
		\\
	&= \classche{\cmortikz{((M\setminus N')\times \II)\sqcup_{M\setminus N'}(M\setminus N'')\times \II}{\mcgim{\tilde{f}}}{M\setminus N}{{\tilde{\iota}}_1^{M\setminus N''}}{M\setminus N''}}
	\end{align*}
	The pushout $((M\setminus N')\times \II)\sqcup_{M\setminus N'}((M\setminus N'')\times \II)$ is given by the set of elements of the form $((m,t),1)$ and $((m,t),2)$ quotiented by the relation $((m,1),1)\sim ((\sh{g}(m),0),2)$.
	To show that the two compositions are equivalent, we must construct a cospan homotopy equivalence $K\colon (M\setminus N'')\times \II \to ((M\setminus N')\times \II)\sqcup_{M\setminus N'}((M\setminus N'')\times \II)$.
	A suitable choice is given by 
	\[
	K(m,t)=\begin{cases}
		[((\sh{g}^{-1}(m),2t),1)], & 0\leq t \leq 1/2 \\
		[((m,2t-1),2)], & 1/2\leq t \leq 1
	\end{cases}
	\]
	Notice that $\sh{g}^{-1}(N'')=N'$, thus $\sh{g}^{-1}(M\setminus N'')=M\setminus N'$, and furthermore, the images are the same equivalence class at $t=1/2$, thus $K$ is well defined. It can be seen that $K$ is continuous by considering it as a composition of maps into the disjoint union and then projecting to the pushout.

	Each piece is clearly continuous, and defined on closed sets, thus $K$ is continuous.
	It is straightforward to see that $K$ commutes with the relevant maps in the cospans.
	We can see that $K$ is a homotopy equivalence by noting that it is, in particular a homeomorphism.
	A continuous inverse is given by the map 
	$L\colon ((M\setminus N')\times \II)\sqcup_{M\setminus N'}(M\setminus N'')\times \II \to (M\setminus N')\times \II$,
	\[
	L([((m,t),i)])=\begin{cases}
		 (m,t/2) & i=1 \\
		 (\sh{g}(m),(t+1)/2) & i=2. 
	\end{cases}
	\]
	It is straightforward to check that $L$ is well-defined and continuous.
\end{proof}

\subsection{Composition of functor \texorpdfstring{$\F\colon \Mot{M}{A}\to \mcg{M}{A}$}{FMottomcg} with \texorpdfstring{$\mcgfunctor{M}{A}$}{mcgfunctor} }

From \cite{motion} we have that, for any manifold $M$ and subset $A\subset M$, there is a functor $\F\colon \Mot{M}{A}\to \mcg{M}{A}$, which is the identity on objects and which sends the $\simrp$ equivalence class of a motion $\mot{f}{}{N}{N'}$ to the $\simi$ equivalence class of $\shmot{f_1}{}{N}{N'}$.
\begin{theorem}\label{th:comp_mcgfunctorF}
	Let $M$ be a manifold, with $A\subset M$ a subset, we have that 
	\[
	\Motfunctor{M}{A} = \mcgfunctor{M}{A}\circ\F .
	\]
\end{theorem}
\begin{proof}
	Suppose we have a morphism $\classrp{\mot{f}{}{N}{N'}}\in \hfMot{M}{A}(N,N')$, then 
	\begin{align*}\mcgfunctor{M}{A} \circ \F(\classm{\mot{f}{}{N}{N'}})&=\classi{\shmot{f_1}{}{N}{N'}} \\ &=\classche{\cmortikz{(M\setminus N')\times \II }{\mcgmotim{f_1}}{M\setminus N}{\iota_1^{M\setminus N'}}{M\setminus N'}}\end{align*}
	and
	\[\Motfunctor{M}{A}(\classm{\mot{f}{}{N}{N'}})=\cmortikz{\Wc(\mot{f}{}{N}{N'}) }{\motzero{f}}{M\setminus N}{\motone{f}}{M\setminus N'}.
	\]
	We must construct a homotopy equivalence $K\colon (M\setminus N')\times \II\to \Wc(\mot{f}{}{N}{N'})$ commuting with the cospans.
	From \cite[Thm.3.72]{motion}, for any motion $\mot{f}{}{N}{N'}$, there is homeomorphism of $\Theta(f)\colon M\times \II\to M\times \II$ given by $f(m,t)=(f_t(m),t)$.
		 
	 Let $K'\colon M\times \II\to M\times \II$ be the map $(m,t) \mapsto \Theta(f)(f_1^{-1}(m),t)$, and notice that this is a homeomorphism.
	 Notice also that $K'$ restricts to a well defined homeomorphism	 
	 $K\colon (M\setminus N')\times \II\to \Wc(\mot{f}{}{N}{N'})$, since $K'(N'\times \II)=\Theta(f)(N\times \II)=  \bigcup_{t \in [0,1]} f_t(N) \times \{t\} =\W(\mot{f}{}{N}{N'})$.
	
	We now check that it commutes with the cospans.
	Consider $m\in M\setminus N$. Then $K(\mcgmotim{f_1} (m))=K(f_1(m),0)=(\Theta(f)(f_1^{-1}(f_1(m)),0)=(f_0(m),0)=(m,0)=\motzero{f}(m)$, and $K(\iota_1^{M\setminus N'}(m))=K(m,1)=\Theta(f)(f_1^{-1}(m),1)=(f_1(f_1^{-1}(m)),1)=(m,1)$.
\end{proof}

\rem{The previous theorem says that any representation of $\hfMot{M}{A}$ obtained using the functor $\Motfunctor{M}{A}$, i.e. factoring through $\HomCob$, will map any pair of motion classes to the same linear map if their images under $\F$ are equivalent in $\mcg{M}{A}$.
For many cases we are interested in the functor $F$ is an isomorphism, but for some cases we need to change the category $\HomCob$ to see the information.}
\section{Construction of a functor \texorpdfstring{$\FG\colon \HomCob \to \VectC$}{ZG:HomCobto Vect}}\label{sec:tqft}

In this section we explicitly construct, for each finite group $G$, a functor 
\[
\tqft \colon \HomCob \to \Vect_{\mathbb{C}}.
\]
 We begin by defining a magmoid morphism from a modification of $\cHomCob$ whose objects are pairs of a space and a set of basepoints, into $\vVectC$.
 We then construct a colimit over all representative finite subsets $A\subseteq X$ of basepoints and maps $g\colon \pi(X,A)\to G$, and arrive at the functor $\FG$ in Theorem~\ref{th:FG_functor}.
 We then show, in Section~\ref{sec:local_equiv}, that $\FG$ can be calculated on objects by choosing a convenient fixed choice of finite representative subset $A\subseteq X$ (Theorem~\ref{natiso_to_colim}).
 Thus our construction is explicitly calculable, see Example~\ref{ex:ZG_surfacetang}.
We will ultimately prove, in Theorem~\ref{th:FG_pi(X)}, that $\FG$ maps an object space $X$ to 
the vector space with basis the set of equivalence classes of maps $f\colon \pi(X)\to G$, up to natural transformation.
In Section~\ref{sec:FG_mon} we prove that $\FG$ is a symmetric monoidal functor.

\subsection{Magmoid of based cospans, \texorpdfstring{$\bHomCob$}{bHomCob}}{}\label{sec:bHomCob}
Let $\chip$ denote the class of pairs of the form $(X,X_0)$ where $X$ is a \homfin{}  space (Def.~\ref{de:homfin}) and $X_0$ is a representative finite subset of $X$ (representative means one point in each path component).
We will refer to the set $X_0$ as a set of basepoints.

\begin{definition}
	Let $X,Y$ be spaces and $A \subseteq X$ and $ B \subseteq Y$ be subsets.
	A {\em map of pairs} $f\colon (X,A) \to (Y,B)$ is a map $f\colon X \to Y$ such that $f(A)\subseteq B$.
\end{definition}

\begin{definition}\label{de:cbcc}
	Let $(X,X_0)$, $(Y,Y_0)$ and $(M,M_0)$ be pairs in $\chip$.
	A diagram $\cbmor{M}{i}{X}{j}{Y}$
	is called a  {\em \cbhomcob{}} from $(X,X_0)$ to $(Y,Y_0)$  such that:
	\vspace{-\topsep}
	\begin{enumerate}[label=(\roman*)]
		\setlength{\parskip}{0pt}
		\setlength{\itemsep}{0pt plus 1pt}
		\item $i \colon X \to M \rightarrow Y \colon j$ is a \chomcob{}.
		\item $i$ and $j$ are maps of pairs.
		\item $M_0 \cap i(X) =i(X_0)$ and 
		$M_0 \cap j(Y) = j(Y_0)$.
	\end{enumerate} 
	\vspace{-\topsep} 
	For any pairs $(X,X_0),(Y,Y_0)$ in $\chip$,
	\[
	\bHomCob((X,X_0),(Y,Y_0))=\left\{ \text{ based homotopy cobordisms } \cbmortikz{M}{i}{X}{j}{Y}\right\}.
	\]
\end{definition}

\exa{Consider the \ccc{} described in Proposition~\ref{disk}, which is a \homcob{} (see Example~\ref{ex:disk}). We can add basepoints to obtain a  \bhomcob{} as shown in Figure~\ref{fig:bccc_S1,D2}.
 	\begin{figure}
		\centering
		\def\svgwidth{0.5\columnwidth}
\begingroup%
  \makeatletter%
  \providecommand\color[2][]{%
    \errmessage{(Inkscape) Color is used for the text in Inkscape, but the package 'color.sty' is not loaded}%
    \renewcommand\color[2][]{}%
  }%
  \providecommand\transparent[1]{%
    \errmessage{(Inkscape) Transparency is used (non-zero) for the text in Inkscape, but the package 'transparent.sty' is not loaded}%
    \renewcommand\transparent[1]{}%
  }%
  \providecommand\rotatebox[2]{#2}%
  \newcommand*\fsize{\dimexpr\f@size pt\relax}%
  \newcommand*\lineheight[1]{\fontsize{\fsize}{#1\fsize}\selectfont}%
  \ifx\svgwidth\undefined%
    \setlength{\unitlength}{567.05281787bp}%
    \ifx\svgscale\undefined%
      \relax%
    \else%
      \setlength{\unitlength}{\unitlength * \real{\svgscale}}%
    \fi%
  \else%
    \setlength{\unitlength}{\svgwidth}%
  \fi%
  \global\let\svgwidth\undefined%
  \global\let\svgscale\undefined%
  \makeatother%
  \begin{picture}(1,0.36208647)%
    \lineheight{1}%
    \setlength\tabcolsep{0pt}%
    \put(0,0){\includegraphics[width=\unitlength,page=1]{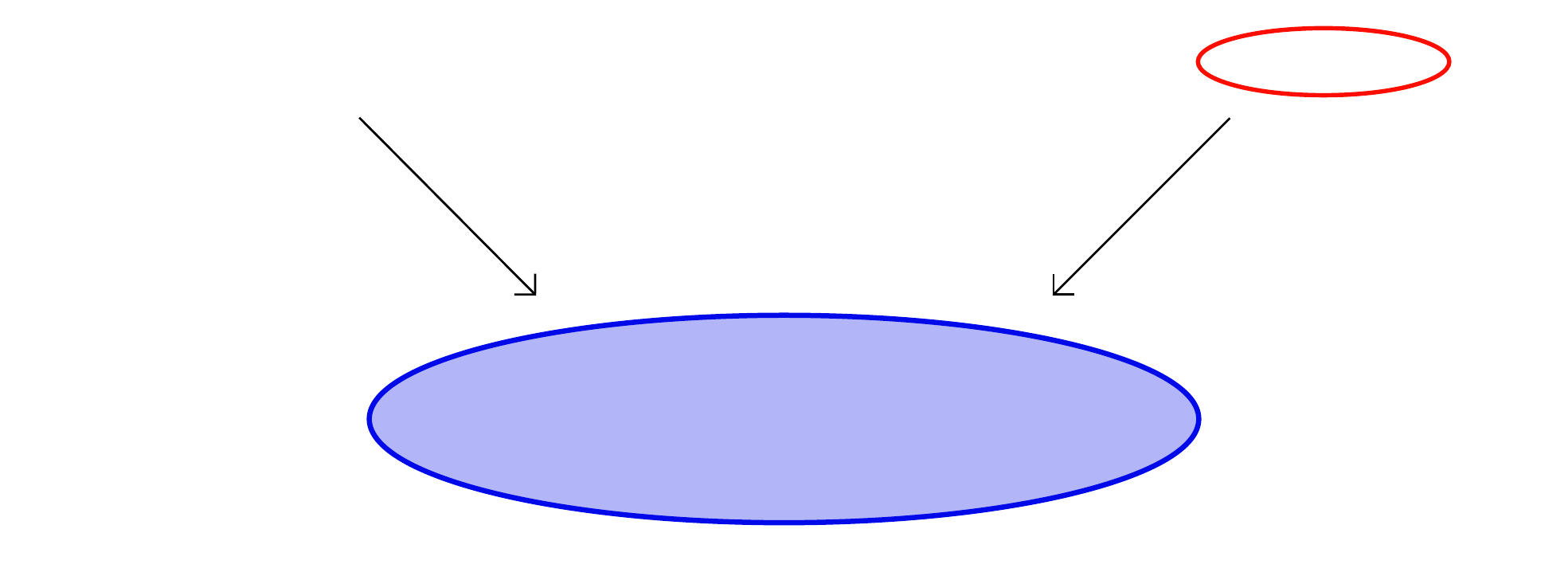}}%
    \put(0.68519346,0.00332725){\makebox(0,0)[lt]{\lineheight{1.25}\smash{\begin{tabular}[t]{l}$(D^2,D^2_0)$\end{tabular}}}}%
    \put(0.73809856,0.20077677){\makebox(0,0)[lt]{\lineheight{1.25}\smash{\begin{tabular}[t]{l}$j$\end{tabular}}}}%
    \put(0.22227359,0.2140031){\makebox(0,0)[lt]{\lineheight{1.25}\smash{\begin{tabular}[t]{l}$i$\end{tabular}}}}%
    \put(0,0){\includegraphics[width=\unitlength,page=2]{bccc_S1,D2.pdf}}%
    \put(0.87036142,0.25935034){\makebox(0,0)[lt]{\lineheight{1.25}\smash{\begin{tabular}[t]{l}$(S^1,{S^1_0}')$\end{tabular}}}}%
    \put(-0.00117108,0.24400533){\makebox(0,0)[lt]{\lineheight{1.25}\smash{\begin{tabular}[t]{l}$(S^1,S^1_0)$\end{tabular}}}}%
    \put(0,0){\includegraphics[width=\unitlength,page=3]{bccc_S1,D2.pdf}}%
  \end{picture}%
\endgroup%

		\caption[Example of a \bhomcob{}]{Here the dots represent basepoints. The three blue points in the leftmost copy of $S^1$ are $S^1_0$ and and two red basepoints in the rightmost copy of $S^1$ are ${S^1}_0'$.
		The set of basepoints $D^2_0$ then contains $i(S^1_0)\cup j({S^1}'_0)$ as well as two extra basepoints marked in green which do not intersect $i(S^1)\cup j(S^1)$. Then $\cmor{(D^2,D^2_0)}{i}{(S^1,S^1_0)}{j}{{S^1}'_0 }$ is a \cbhomcob{}.} 
		\label{fig:bccc_S1,D2}
\end{figure}}

\prop{Let $\cmor{M}{i}{X}{j}{Y}$ be a \chomcob{}, then there exists a \bhomcob{} $\cbmor{M}{i}{X}{j}{Y}$ for some representative finite subsets $X_0,Y_0$ and $M_0$ of $X$, $Y$ and $M$ respectively.}
\begin{proof}
	A suitable choice of $X_0$, $Y_0$ and $M_0$ is constructed as follows.
	Choose a point in each path-connected component of $X$, and let $X_0$ be the union of these points.
	Choose  $Y_0$ in the same way.
	Let $M_0$ be the union of $i(X_0)$ and $j(Y_0)$, together with a choice of point in each path-connected component not containing a point $i(X_0)\cup j(Y_0)$.
	
	Notice that $X_0,Y_0$ and $M_0$ are finite, 
	as $X$, $Y$ and $M$ are \homfin{}, and thus contain a finite number of path-connected components.
\end{proof}

Of course for any manifold $X$ consisting of a single path connected component, there is an uncountably infinite number of choices of $\{x\}$ such that $(X,\{x\})\in \chip$. Hence there will usually be many ways to obtain a \bhomcob{} from a \homcob{}.

\medskip 

The following Lemma says that the composition $\sbullet$ extends to a composition of \bhomcob{}s.
\lemm{\label{le:comp_bhomcob}
	$(I)$ For any spaces $X,Y$ and $Z$ in $Ob(\cHomCob)$ there is a composition
	\ali{
		\sbullet\;\colon \bHomCob((X,X_0),(Y,Y_0))\times \bHomCob((Y,Y_0),(Z,Z_0))&\to \bHomCob((X,X_0),(Z,Z_0)) \\
		\left({\cbmortikz{M}{i}{X}{j}{Y}}\;\raisebox{-1em}{,}\;{\cbmortikz{N}{k}{Y}{l}{Z}}\right)&\mapsto \scalebox{0.65}{\cmortikz{(M \sqcup_Y N,M_0\sqcup_{Y_0}N_0)}{\tilde{i}}{(X,X_0)}{\tilde{l}}{(Z,Z_0)}}
	}
	where $\cmor{M \sqcup_Y N}{\tilde{i}}{X}{\tilde{l}}{Z}$ is the composition 
	$(\cmor{M}{i}{X}{j}{Y})\sbullet(\cmor{N}{k}{Y}{l}{Z})$ with $\sbullet$ as in Lemma~\ref{le:comp_ccc}, and $M_0\sqcup_{Y_0}N_0$
	the set pushout of $M_0\xleftarrow{j}Y_0\xrightarrow{k}Z_0$ (where we use $j$ and $k$ also for the obvious restrictions).\\
	$(II)$ Hence there is a magmoid 
	\[
	\bHomCob=(\chip,\bHomCob(-,-),\sbullet).
	\]
}

\begin{proof}
	We check that the composition is well defined.
	We first show that $M_0\sqcup_{Y_0} N_0$ is representative in $M\sqcup_Y N$.	
	Recall that our chosen representatives of pushouts in $\Topo$ have the same underlying set and set maps as the representative of the corresponding pushout of the underlying set maps in $\Set$.
	Also $j$ and $k$ are homeomorphisms, by Lemma~\ref{le:homeo_cofib}.
	Thus $M_0\sqcup_{Y_0} N_0\subseteq M\sqcup_{Y} N$ and $M_0\sqcup_{Y_0} N_0=p_M(M_0)\cup p_N(N_0)$ where $p_M$ and $p_N$ are as in Proposition~\ref{le:comp_ccc}.
	Let $m\in M\sqcup_Y N$ be any point, then it has a preimage $p^{-1}(m)$ in $M$ or $N$, and thus there is a path in $M$ or $N$ connecting $p^{-1}(m)$ to a point in $M_0$ or $N_0$. The image of this path under $p_M$ or $p_N$ connects $m$ to a point in $M_0\sqcup_{Y_0} N_0$. 
	
	We have from Proposition~\ref{HomCob} that $\cmor{M\sqcup_{Y}N}{i}{X}{l}{Z}$ is a \chomcob{}.
	
	Since $M_0\sqcup_{Y_0} N_0=p_M(M_0)\cup p_N(N_0)$, and $i(X_0)\subseteq M_0$ and  $l(Z_0)\subseteq N_0$, we have $\tilde{i}(X_0)\subseteq M_0\sqcup_{Y_0}N_0$ and $\tilde{l}(Z_0)\subseteq M_0\sqcup_{Y_0}N_0$, thus $\tilde{i}$ and $\tilde{k}$ are maps of pairs.
	
	The map $\langle i,j \rangle \colon X\sqcup Y \to M$ is a cofibration, hence by Theorem~\ref{th:homeo_onto_im} it is a homeomorphism onto its image.
	This means $i(X)\cap j(Y)=\emptyset$ in $M$, and similarly $k(Y)\cap l(Z)=\emptyset$. Hence there is no equivalence on points in $X$ or $Z$ in the pushout. Thus $\left(M_0\sqcup_{Y_0} N_0\right) \cap \tilde{i}(X)=\tilde{i}(X_0)$ 
	follows directly from the fact that $\left(M_0\right) \cap i(X)=i(X_0)$ and similarly
	$\left(M_0 \sqcup_{Y_0}N_0\right) \cap \tilde{l}(Z) = \tilde{l}(Z_0)$.	
\end{proof}

\subsection{Magmoid morphism from \texorpdfstring{$\bHomCob$}{HomCob} to \texorpdfstring{$\vVectC$}{VectC}}

Here we construct a magmoid morphism 
$
\bFG \colon \bHomCob\to \vVectC
$.

Recall that there is a groupoid $\GG_G=(\{*\},\GG_G(*,*),\circ_G,e_G,g\mapsto g^{-1})$ obtained from any group $G$ with morphisms the elements of $G$, and composition and inverses given by group composition and inverse. Throughout this section, by abuse of notation we will use $G$ for $\GG_G$.
\defn{
Let $G$ be a group. 
For a pair $(X,X_0)\in \chip$, define
\[
\bFG(X,X_0) = \C \left( \Grpd \left( \pi(X,X_0) ,G \right) \right).
\]
That is, $\bFG(X,X_0)$ is the $\C$ vector space whose basis is the set of groupoid maps from the fundamental groupoid of $X$ with respect to $X_0$ into $G$.}

\exa{\label{ex:ObS1cupS1} Let $X=S^1\sqcup S^1$, and let $X_0\subset X$ contain two points, one in each copy of $S^1$.
	We have $\pi(X,X_0)\cong \Z\sqcup \Z$. Hence maps from $\pi(X,X_0)$ to $G$ are determined by pairs in $G\times G$, where the elements of $G$ denote the image of the generating elements of each copy of $\Z$. 
So we have $\bFG(X,X_0)\cong\C(G\times G)$. }

In the following example note that $X_0$ must be representative by the definition of $\chip$, therefore we choose basepoints even in path components that are homotopically trivial. 
This is necessary since, when considered as part of a cospan, these trivial components may have image in a homotopically non-trivial component.  

\exa{\label{ex:ObembS1cupS1}
	Let $X$ and $X_0$ be as explained in the caption to Figure~\ref{fig:emb_S1cupS1_based}.
	Then $\pi(X,X_0)\cong(\Z* \Z)\sqcup \{*\}\sqcup \{*\}$ where $\{*\}$ denotes the trivial group. Then maps from $\pi(X,X_0)$ to $G$ are determined by pairs in $G\times G$, whose elements respectively denote the images of the equivalence classes of the loops marked $x_1$ and $x_2$ in the figure, so we have $\bFG(X,X_0)\cong\C(G\times G)$.
	\begin{figure}
		\centering
		\def\svgwidth{0.5\columnwidth}
\begingroup%
  \makeatletter%
  \providecommand\color[2][]{%
    \errmessage{(Inkscape) Color is used for the text in Inkscape, but the package 'color.sty' is not loaded}%
    \renewcommand\color[2][]{}%
  }%
  \providecommand\transparent[1]{%
    \errmessage{(Inkscape) Transparency is used (non-zero) for the text in Inkscape, but the package 'transparent.sty' is not loaded}%
    \renewcommand\transparent[1]{}%
  }%
  \providecommand\rotatebox[2]{#2}%
  \newcommand*\fsize{\dimexpr\f@size pt\relax}%
  \newcommand*\lineheight[1]{\fontsize{\fsize}{#1\fsize}\selectfont}%
  \ifx\svgwidth\undefined%
    \setlength{\unitlength}{449.07728097bp}%
    \ifx\svgscale\undefined%
      \relax%
    \else%
      \setlength{\unitlength}{\unitlength * \real{\svgscale}}%
    \fi%
  \else%
    \setlength{\unitlength}{\svgwidth}%
  \fi%
  \global\let\svgwidth\undefined%
  \global\let\svgscale\undefined%
  \makeatother%
  \begin{picture}(1,0.43566711)%
    \lineheight{1}%
    \setlength\tabcolsep{0pt}%
    \put(0,0){\includegraphics[width=\unitlength,page=1]{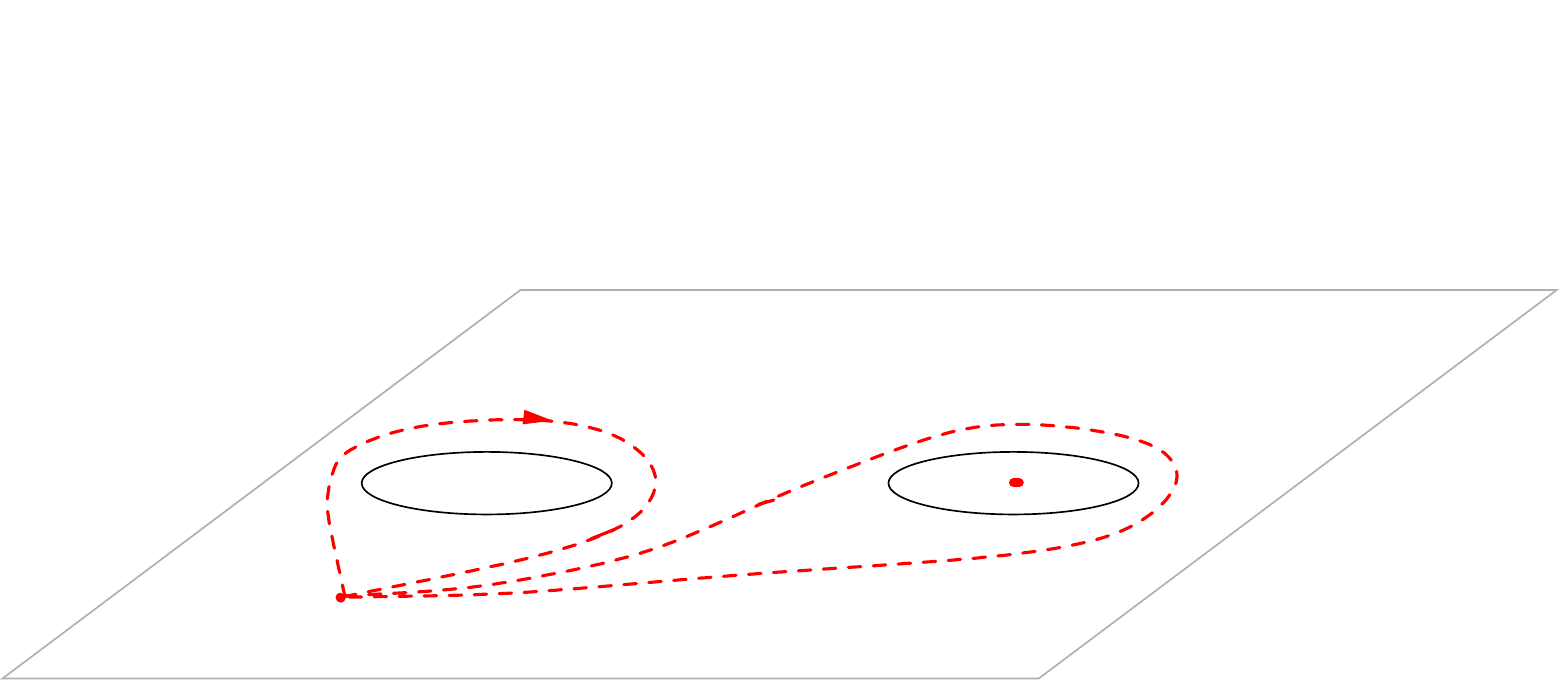}}%
    \put(0.42165773,0.14598936){\color[rgb]{0,0,0}\makebox(0,0)[lt]{\lineheight{1.25}\smash{\begin{tabular}[t]{l}$x_1$\end{tabular}}}}%
    \put(0,0){\includegraphics[width=\unitlength,page=2]{emb_S1cupS1_based.pdf}}%
    \put(0.71367689,0.16819658){\color[rgb]{0,0,0}\makebox(0,0)[lt]{\lineheight{1.25}\smash{\begin{tabular}[t]{l}$x_2$\end{tabular}}}}%
  \end{picture}%
\endgroup%

		\caption[Illustration of calculation of $\bFG$ on the complement of embeddings of $S^1$ in $\II^2$]{
		Let $X$ be the complement in $\II^2$ of the embedding of $S^1\sqcup S^1$ shown. Let $X_0$ be the three marked points shown.
		} 
		\label{fig:emb_S1cupS1_based}
\end{figure}}

For each pair $(X,X_0)$, the vector space $\bFG(X,X_0)$
has an intrinsic basis, the maps into $G$.
We will define a linear map assigned to each \bhomcob{} as a matrix indexed by these bases.

\defn{
Let $\cbmor{M}{i}{X}{j}{Y}$ be a \cbhomcob.
We define a matrix
\[
\bFG \left(\cbmortikz{M}{i}{X}{j}{Y}\right) : \bFG(X,X_0) \to \bFG(Y,Y_0)
\]
as follows. 
Let $f \in \bFG(X,X_0)$ and $g \in \bFG(Y,Y_0)$ be basis elements, then
\begin{align}
\label{eq:dirac_bFG}
\left\langle g \,\middle|\, \bFG \left(\cbmortikz{M}{i}{X}{j}{Y}\right) \middle| f \right\rangle
=\left|\left\{ h:\pi(M, M_0) \to G\hspace*{0.2em} \middle\vert \scalebox{.7}{
	\begin{tikzcd}[cramped, column sep=tiny, ampersand replacement=\&] 
	\pi(X, X_0) \ar[dr,"\pi(i)"'] \ar[ddr, bend right,"f"'] \& \&  \pi(Y,Y_0) \ar[dl,"\pi(j)"] \ar[ddl, bend left,"g"] \\
	\& \pi(M, M_0) \ar[d,"h"]\& \\
	\& G \&
	\end{tikzcd}}
\right\}\right|.
\end{align}
In other words, the right hand side is the cardinality of the set of maps $h$ making the diagram commute.
Here we are using Dirac notation: \eqref{eq:dirac_bFG} is the matrix element in the column corresponding to $f$ and the row corresponding to $g$.

\sloppypar{
When we have already specified the relevant cospan, we will often use $\bFG (M,M_0)$ 
for $\bFG \left(\cbmor{M}{i}{X}{j}{Y}\right)$,
and
 write the matrix elements as}
\[
\langle g \,|\, \bFG (M,M_0) \,|\, f \rangle = \left\vert \left\{h\colon  \pi(M,M_0) \to G  \;|\; h|_{\pi(X,X_0)}=f \wedge h|_{\pi(Y,Y_0)}=g \right\}\right\vert,
\]
	where by $h|_{\pi(X,X_0)}$ we really mean the restriction of the map $h$ to the image $\pi(i)(\pi(X,X_0))$.
}

\exa{\label{ex:bFG_mer}
	Let $\cbmor{M}{i}{X}{j}{Y}$ 
	be the \bhomcob{} represented in Figure~\ref{fig:mer_c}, and described in the caption. Note this is a \homcob{} as discussed in Examples~\ref{ex:mer} and \ref{ex:mer_b}. Note also that there is a finite number of marked points, and at least one in each connected component of $X$, $Y$ and $M$.
	
	Now $\pi(Y,Y_0)\cong \Z$, where the isomorphism is realised by mapping the loop labelled $y_1$ in the figure to $1$.
	Hence a map $g\colon \pi(Y,Y_0)\to G$ is uniquely determined by a choice of an element $g_1\in G$ with $f(\classp{y_1})=g_1$. Thus we have $\bFG(Y,Y_0)\cong\C(G)$.
	
	Recall from Example~\ref{ex:ObS1cupS1} that $\bFG(X,X_0)\cong\C(G\times G)$, where a pair $(g_1,g_2)$ denotes the map $(g_1,g_2)(\classp{x_1})=g_1$ and $(g_1,g_2)(\classp{x_2})=g_2$.
	
	Let $x$ be the basepoint which is in the loop labelled $x_1$. By Lemma~\ref{le:add_bps}, there is a bijection sending a map $h\in \Grpd(\pi(M,M_0), G)$ to a triple $(h',h(\gamma_1),h(\gamma_2)) \in\Grpd(\pi(M,\{x\}), G)\times G\times G$,
	where $h'$ agrees with $h$ on $\pi(M,\{x\})$.
	The space $M$
	 is equivalent to the twice punctured disk, which has fundamental group isomorphic to the free product $\Z*\Z$. This isomorphism can be realised by sending the element represented by $x_1$ to the $1$ in the first copy of $\Z$ and the element represented by $\gamma_2^{-1}x_2\gamma_2$ to the $1$ in the second copy of $\Z$.
	Thus we can label elements in $\Grpd(\pi(M,\{x\}), G)$ by elements of $G\times G$ where $a \in(a,b)$ corresponds to the image of $\classp{x_1}$, and $b$ the image of $\classp{\gamma_2^{-1}x_2\gamma_2}$. 
	Hence a map in $\Grpd(\pi(M,M_0), G)$ is determined by a quadruple $(a,b,c,d)\in G\times G\times G\times G$ where $a$ corresponds to the image of $x_1$, $b$ to the image of $\gamma_2^{-1}x_2\gamma_2$, and $c$ and $d$ correspond to the images of $\gamma_1$ and $\gamma_2$ respectively.
	
	Choosing basis elements $(f_1,f_2)\in \bFG(X,X_0)$ and $g_1\in \bFG(Y,Y_0)$  
	the commutation condition in \eqref{eq:dirac_bFG} gives conditions on allowed quadruples $(a,b,c,d)\in G\times G\times G\times G$.
	We have
	\ali{
		\lan(g_1) \; \vert\; \bFG(M,M_0)\; \vert\; (f_1,f_2)\ran &=
		\vert \left\{a,b,c,d\in G \; \vert \; (a,dbd^{-1})=(f_1,f_2), c^{-1}bac=g_1 \right\} \vert \\
		&= \vert \left\{b,c,d\in G \vert dbd^{-1}=f_2, c^{-1}bf_1c=g_1 \right\} \vert \\
		&= \vert \left\{d\in G \vert c^{-1} d^{-1}f_2d f_1c=g_1 \right\} \vert.
	}
}

	\begin{figure}
	\centering
	\def\svgwidth{0.3\columnwidth}
	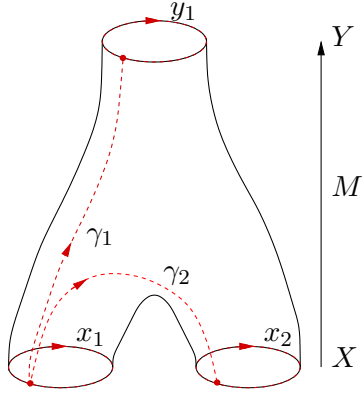
	\caption[Figure showing possible choice of basepoints added to Example~\ref{ex:mer}]{
		This figure represents the \ccc{} from Example~\ref{ex:mer}, so $M$ is the represented manifold, and $X\cong S^1 \sqcup S^1$ and $Y\cong S^1$ are the bottom and top boundary respectively, with the inclusion maps. The red points and dashed lines show a possible choice of basepoints $M_0$ and generating paths. Let $X_0$ and $Y_0$ be the intersection of $M_0$ with $X$ and $Y$ respectively.} 
	\label{fig:mer_c}
\end{figure}

\begin{lemma}\label{le:tqft_bFGmagmor}
	The map
	\[\bFG\colon \bHomCob\to \vVectC\]
	is a magmoid morphism.
\end{lemma}

\begin{proof}
	It is immediate from the construction that the map is well defined. Thus we only need to check that composition is preserved.
	Let $\cbmor{M}{i}{X}{j}{Y}$ and $\cbmor{N}{k}{Y}{l}{Z}$ be \cbhomcob{}s.
	Let $f\in\bFG(X,X_0)$ and $g\in\bFG(Z,Z_0)$ be basis elements. 
	The matrix element corresponding to these basis elements is given by counting maps $h$ in the following diagram.
	\[
	\begin{tikzcd}[cramped,sep=scriptsize, ampersand replacement=\&] 
	\pi(X, X_0) \ar[dr,"\pi(i)"'] \ar[dddrr, bend right=35,"f"'] \& \& \pi(Y,Y_0) \ar[dr,"\pi(k)"'] \ar[dl,"\pi(j)"] \& \& \pi(Z,Z_0) \ar[dl,"\pi(l)"] \ar[dddll, bend left=35,"g"] \\
	\& \pi(M,M_0) \ar[dr]  \& \& \pi(N,N_0)\ar[dl] \& \\
	\& \& \makebox*{$\pi(X, X_0)$}{$\pi(M \sqcup_Y N, M_0 \sqcup_{Y_0} N_0)$} \ar[d,"h"]\& \&\\[+15pt]
	\& \& G \& \&
	\end{tikzcd}
	\]
	From Definition~\ref{de:cbcc} and Lemma~\ref{ccc_implies_cofibs}, the pushout of $ (M,M_0)\xleftarrow{j}(Y,Y_0)\xrightarrow{k}(N,N_0)$ satisfies the conditions of  Corollary~\ref{co:vk}. Hence the middle square of this diagram is a pushout.
	Hence each $h$ is uniquely determined by a pair
	$h_1\colon \pi(M,M_0) \to G$ and 
	$h_2\colon \pi(N,N_0) \to G$ 
	such that the above diagram commutes.
	So we have
	\begin{align*}
	\langle g \,|\, \mathcal{F}_{G}^{!} (M \sqcup_{Y} N ,M_0 \sqcup_{Y_0} N_0) \,|\, f \rangle
	&= \left\vert \left\{  h_1, h_2 \,|\,  h_1 \circ \pi(j)= h_2 \circ \pi(k) \wedge h_1|_{\pi(X,X_0)} = f \right.\right. \\[-2pt] 
	& \hspace*{16em} \left.\left. \wedge h_2|_{\pi(Z,Z_0)} = g\right\} \right\vert  \\
	&=\sum _{\makebox[0pt]{\scriptsize $\theta:\pi(Y,Y_0) \to G$}} \left\vert\left\{h_1 \,|\, h_1|_{\pi(Y,Y_0)} = \theta \wedge h_1|_{\pi(X,X_0)} = f \right\} \right\vert \\[-12pt] 
	& \hspace*{8em} \left\vert\left\{ h_2 \,|\, h_2|_{\pi(Y,Y_0)}= \theta \wedge h_2|_{\pi(Z,Z_0)} = g\right\}\right\vert \\
	&=\sum _{\makebox[0pt]{\scriptsize $\theta:\pi(Y,Y_0) \to G$}}\langle g \vert \mathcal{F}_{G}^{!} (N,N_0) \vert \theta \rangle
	\langle \theta \vert \mathcal{F}_{G}^{!} (M,M_0) \vert f \rangle
	\end{align*}
	Now this is precisely the corresponding matrix element given by multiplying the matrices $\mathcal{F}_{G}^{!}(M,M_0)$ and $\mathcal{F}_{G}^{!}(N,N_0)$. 
\end{proof}

The following lemma says that $\bFG$ respects \che{}.
\begin{lemma}\label{bFG_well_defined}
	Suppose we have \chomcob{}s $\cmor{M}{i}{X}{j}{Y}$ and $\cmor{M'}{i'}{X}{j'}{Y}$ which are equivalent up to cospan homotopy equivalence (as defined in Lemma~\ref{le:che}). 
	Then (by Theorem~\ref{th:cofib_equiv}) we have homotopy equivalences 
	$\psi\colon M\to M'$ and $\psi'\colon M'\to M$ which commute with the cospans.
	Choosing sets of baspoints $X_0\subseteq X$, $Y_0\subseteq Y$, $M_0\subseteq M$ such that $\cbmor{M}{i}{X}{j}{Y}$ is a \bhomcob{}, we have
	\[
	\bFG\left(\cbmortikz{M}{i}{X}{j}{Y}\right) = \bFG\left(
	\scalebox{0.6}{\begin{tikzcd}[ampersand replacement= \&, cramped ,sep=tiny]
		(X, X_0) \ar[dr,"i'"' near start] \& \& \makebox*{MM}{$(Y,Y_0)$} \ar[dl,"j'" near start] \\
		\& (M',M_0'), \& \\
		\end{tikzcd}
		\hspace*{5pt}}
	\right)
	\]
	where $M_0'=\psi(M_0)$.
\end{lemma}
\begin{proof}
		Let $f\in \bFG(X,X_0)$ and $g\in \bFG(Y,Y_0) $ be basis elements and  $h\colon  \pi(M,M_0) \to G$ a map with $h|_{\pi(X,X_0)}=f$ and  $h|_{\pi(Y,Y_0)}=g$.
		On the level of fundamental groupoids $\psi$ and $\psi'$ become inverse group isomorphisms making the following diagram commute.

	\[
	\begin{tikzcd}[ column sep=tiny, ampersand replacement=\&] 
	\& \pi(M,M_0) \ar[dd,"\pi(\psi)"', bend right] \ar[ddd, bend left=45,  "h" near start]\& \\
	\pi(X, X_0) \ar[dr,"\pi(i')"'] \ar[ur,"\pi(i)"] \ar[ddr, bend right,"f"'] \& \&  \pi(Y,Y_0) \ar[ul,"\pi(j)"'] \ar[dl,"\pi(j')", crossing over] \ar[ddl, bend left,"g"] \\
	\& \pi(M', M'_0) \ar[d,dashed,"h'"] \ar[uu,"\pi(\psi')",bend right]\& \\
	\& G \&
	\end{tikzcd}
	\]
	For any such $h$ we can obtain a map $h'$ making the diagram commute
	by precomposing $h$ with $\pi(\psi')$.
	Thus we have a set map 
	\begin{multline*}\Psi\colon \left\{h\colon  \pi(M,M_0) \to G  \,|\, h|_{\pi(X,X_0)}=f \wedge h|_{\pi(Y,Y_0)}=g \right\} \to \\
		\left\{h'\colon  \pi(M',M'_0) \to G  \,|\, h'|_{\pi(X,X_0)}=f \wedge h'|_{\pi(Y,Y_0)}=g \right\},
		\end{multline*}
	which has inverse given by precomposing with $\pi(\psi)$.
	Thus we have that 
	$\langle g \,|\, \bFG (M,M_0) \,|\, f \rangle = \langle g \,|\, \bFG (M',M'_0) \,|\, f \rangle$ for all $f,g$.
\end{proof}

\subsection{Functor \texorpdfstring{$\FG\colon\HomCob\to \Vect_{\mathbb{C}}$}{HomCob}}\label{remove_bps}
The magmoid morphism $\bFG$ clearly depends on the sets of basepoints. 
Also notice that there are many ways we could obtain a based cospan from the cospan representing the identity in $\HomCob$, $\cmor{X\times \II }{\iota_0^X}{X}{\iota_1^X}{X}$, and in general these based cospans will not give the identity matrix under $\bFG$. 
In this section we add a normalisation factor, then define a map on the objects of $\HomCob$ by taking a colimit over $\bFG$ for all choices of basepoints, and  finally
redefine the map on morphisms using the universal property of the colimit. This leads to a functor $\FG\colon \HomCob\to \Vect_{\mathbb{C}}$ which does not depend on a choice of basepoints.
We will find that
removing the basepoint dependence will also solve the identity problem. 

\subsubsection*{Varying the set of basepoints}

Let $\cbmor{M}{i}{X}{j}{Y}$ be a \cbhomcob{}.
We first consider how changing the set of basepoints in the set $M_0$ changes $\bFG(M,M_0)$.

If such a point exists, choose a point $m\in M\setminus  M_0$ such that 
$\scbmor{M}{M_0\cup\{m\}}{i}{X}{X_0}{j}{Y}{Y_0}$ is also a \cbhomcob{}. 
By Lemma~\ref{le:add_bps}, the set of maps $h'\colon  \pi(M,M_0\cup \{m\}) \to G $ is in bijective correspondence with the set of pairs of a map $h\colon  \pi(M,M_0) \to G $ and an element of $G$, via a bijection $\Theta_{\gamma}$.
Let $f\in\bFG(X,X_0)$ and $g\in\bFG(Y,Y_0)$ be basis elements.
Note that for all $h\colon  \pi(M,M_0) \to G $ such that  $h|_{\pi(X,X_0)}=f $ and $ h|_{\pi(Y,Y_0)}=g$, the map $h'$ obtained from a pair $(h,x)$ using the map $\Theta^{-1}_\gamma$ as in the proof of Lemma~\ref{le:add_bps} also satisfies $h'|_{\pi(X,X_0)}=f $ and $ h'|_{\pi(Y,Y_0)}=g $.
Similarly, for a map $h'\colon  \pi(M,M_0\cup \{m\}) \to G $ satisfying $h'|_{\pi(X,X_0)}=f $ and $ h'|_{\pi(Y,Y_0)}=g $, $\Theta_{\gamma}(h')|_{\pi(X,X_0)}=f $ and $ \Theta_{\gamma}(h')|_{\pi(Y,Y_0)}=g$.
Hence  
for all pairs $f,g$ we have that 
$
\left\langle g \,\middle|\, \bFG (M,M_0 \cup \{m\}) \,\middle|\, f \right\rangle = |G| \left\langle g \,\middle|\, \bFG (M,M_0) \,\middle|\, f \right\rangle $,
and hence that 
\begin{align*}
 \bFG (M,M_0 \cup \left\{m \right\} )  &= |G| \, \bFG (M,M_0).
\end{align*}
It follows that for all $ M_0' \supseteq M_0$ we have
$
\bFG (M,M_0') = |G|^{(|M_0'| - |M_0|)} \bFG (M,M_0),
$ 
and hence
\begin{align*}
|G|^{-|M_0'|}\bFG (M,M_0') &= |G|^{ - |M_0|} \bFG (M,M_0). 
\end{align*}
Now suppose instead there are no containment conditions between $M_0'$ and $M_0$, then we can write 
\[\bFG (M,M_0' \cup M_0) = |G|^{(|M_0' \cup M_0| - |M_0|)} \bFG (M,M_0)
\]
and
\[\bFG (M,M_0' \cup M_0) = |G|^{(|M_0' \cup M_0| - |M_0'|)} \bFG (M,M_0')
\]
which together imply
\[|G|^{- |M_0|} \bFG (M,M_0) = |G|^{- |M_0'|} \bFG (M,M_0')
\]
and that
\[|G|^{- (|M_0|-|X_0|)} \bFG (M,M_0) = |G|^{- (|M_0'|-|X_0|)} \bFG (M,M_0').
\]
We have proven the following.
\lemm{\label{le:bbFG_mor_bps}
	The map $\bbFG$, assigning a linear map to a \cbhomcob{} as follows
	\[
	\bbFG \left(\cbmortikz{M}{i}{X}{j}{Y} \right) = |G|^{- (|M_0| - |X_0|)} \bFG \left(\cbmortikz{M}{i}{X}{j}{Y} \right),
	\]
	 does not depend on the subset $M_0\subseteq M$.
 \qed}
When the relevant cospan is clear, or has been given, we will refer to the image as $\bbFG(M,X_0,Y_0)$ to highlight the dependence on $X_0$ and $Y_0$.

In defining $\bbFG$ we have included a term counting the cardinality of $X_0$.
Some term counting basepoints in $X$ or $Y$ is necessary to ensure the new definition is still compatible with the composition; however we could have chosen $\nicefrac{1}{2}(|X_0|+|Y_0|)$ for example, as is the convention in \cite{yetter},
and avoided the asymmetry. 
The reason for our convention is that it allows us to work for longer in the basis set rather than moving to the $\C$ vector space, making calculation easier. 
We will highlight later where this becomes relevant (Remark~\ref{rem:Yetter}). 

\lemm{\label{le:tqft_compbbFG}
	Let $\cbmor{M}{i}{X}{j}{Y}$ and $\cbmor{N}{k}{Y}{l}{Z}$ be \cbhomcob{}s.
	Then
	\[
	\bbFG\left(\cbmortikz{N}{k}{Y}{l}{Z}\right)\bbFG\left(\cbmortikz{M}{i}{X}{j}{Y}\right)=\bbFG\left(\cbmortikz{M}{i}{X}{j}{Y}\sbullet\cbmortikz{N}{k}{Y}{l}{Z}\right)
	\]
	where concatenation denotes composition of linear maps, or equivalently matrix multiplication.
}
\begin{proof}
	We have
	\begin{align*}
	\bbFG (M \sqcup_Y N, X_0, Z_0) 
	&= |G|^{-(|M_0\sqcup_{Y_0} N_0| - |X_0|)}  \bFG (M \sqcup_{Y} N, M_0 \sqcup_{Y_0} N_0) \\
	&= |G|^{-(|M_0|+|N_0| - |Y_0| - |X_0|)}  \bFG (M , M_0 )   \bFG (N, N_0)\\
	&= |G|^{-(|M_0|-|X_0|)} \bFG (M , M_0 ) |G|^{-(|N_0| - |Y_0|)} \bbFG (N, N_0)  \\
	&= \bbFG (M , X_0, Y_0) \bbFG (N, Y_0, Z_0)
	& &\qedhere
	\end{align*}
using that, by Lemma~\ref{le:tqft_bFGmagmor}, $\bFG$ preserves composition. 
\end{proof}

\subsubsection*{Basepoint independent map from $ Ob(\HomCob)$ to  $Ob(\VectC)$}

We focus here on the sets of basepoints in $Ob(\bHomCob)$.
Here we will move to using Greek subscripts to indicate varying choices of subsets, so for a space $X$, objects in $\bHomCob$ are pairs of the form $(X,X_\alpha)$.
We will eventually show that we can choose just one subset to calculate the image and will then switch back to the original notation.

We proceed by constructing, for a space $X\in Ob(\HomCob)$, a colimit in $\VectC$ over a diagram with vertices $\bFG(X,X_\alpha)$ for all possible choices of finite, representative $X_\alpha\subset X$.

\prop{Let $X$ be a \homfin{} space.
	There is a subcategory of $\Set$,
	\[
	\FinSetX=(Ob(\FinSetX),\FinSetX(-,-),\circ,\id)
	\]
	where $Ob(\FinSetX)$ contains all $X_\alpha$ such that $(X,X_\alpha)\in \chip$ (i.e. $X_\alpha$ finite, representative) and 
	$\FinSetX(X_\alpha,X_\beta)$ contains the inclusion $\iota_{\alpha\beta}\colon X_\alpha \to X_\beta$ if $X_\alpha\subseteq X_\beta$, otherwise $\FinSetX(X_\alpha,X_\beta)=\emptyset$.
	}
\begin{proof}
	Note we have $\iota_{\alpha \alpha }=1_{X_{\alpha}}\colon X_\alpha \to X_\alpha$ in $\FinSetX(X_\alpha,X_\alpha)$.
	Suppose $X_\alpha,X_\beta,X_\gamma\in \FinSetX$, with $X_\alpha\subseteq X_\beta\subseteq X_\gamma$, then the composition of $ \iota_{\alpha\beta}\colon X_\alpha\to X_\beta$ and $ \iota_{\beta\gamma}\colon X_\beta\to X_\gamma$ is precisely the unique morphism in $\FinSetX(X_\alpha,X_\gamma)$. (This is the only case for which we have composable morphisms.)
\end{proof}

By abuse of notation, for an inclusion $\iota_{\alpha\beta}\colon X_\alpha\to X_\beta$ we will also write $\iota_{\alpha\beta}\colon\pi(X,X_\alpha)\to \pi(X,X_\beta)$ for the inclusion of groupoids.
\lemm{\label{le:functor_V}
There is a contravariant functor \[{\mathcal{V}_X:\FinSetX \to \mathbf{Set}}\] constructed as follows.
Let $X_\alpha, X_\beta \in Ob(\FinSetX)$ with $X_\beta \subseteq X_\alpha$.
Let $\mathcal{V}_X(X_\alpha)= \Grpd(\pi(X,X_\alpha),G)$.
For any $v_\alpha\in \mathcal{V}_X(X_\alpha)$ we have a commuting triangle
\[
\begin{tikzcd}[ampersand replacement =\&]
\pi(X,X_\beta) \ar[r,"\iota_{\beta\alpha}"] \ar[dr,"v_\alpha \circ \iota_{\beta\alpha}"',dashed] \& \pi(X,X_\alpha) \ar[d,"v_\alpha"] \\ 	
\& G.
\end{tikzcd}
\]
Now let $\mathcal{V}_X(\iota_{\beta\alpha}\colon X_\beta \to X_\alpha) = \phi_{\alpha\beta} $ where $\phi_{\alpha\beta}:\mathcal{V}_X(X_\alpha) \to \mathcal{V}_X(X_\beta) $, $v_\alpha \mapsto v_\alpha \circ \iota_{\alpha\beta}$.
}
\begin{proof}
	We have $\VV(1_{X_\alpha}\colon X_\alpha\to X_\alpha)=1_{\VV(X_\alpha)}\colon \VV(X_\alpha)\to \VV(X_\alpha)$.
		Suppose $X_\alpha,X_\beta,X_\gamma\in \FinSetX$, with $X_\gamma\subseteq X_\beta\subseteq X_\alpha$, then $\phi_{\beta\gamma}\circ \phi_{\alpha\beta} =\phi_{\alpha\gamma}$ since
		$(v_{\alpha}\circ \iota_{\beta\alpha})\circ \iota_{\gamma\beta}=v_\alpha \circ (\iota_{\alpha\gamma})$.
\end{proof}
\begin{lemma}\label{le:tqft_VVepi}
	For any space $X$ and $X_\beta , X_\alpha \in Ob(\FinSetX)$ with $X_\beta\subseteq X_\alpha$, $\phi_{\alpha\beta}$ is a surjection.
\end{lemma}
\begin{proof}
	By Lemma~\ref{le:extension} for any $v_\beta \in \mathcal{V}_X(X_\beta)$ we can extend to some $v_\alpha \in \mathcal{V}_X(X_\alpha)$ which is equal to $v_\beta$ on the image $\iota_{\beta\alpha}(\pi(X,X_\beta))$ in $\pi(X,X_\alpha)$.
\end{proof}

\noindent
The colimit over $\mathcal{V}_X$ consists of a family of commuting triangles diagrams of the form
\[
\begin{tikzcd}[ampersand replacement =\&, column sep=tiny]
\mathcal{V}_X(X_\alpha) \ar[rd,"\phi_\alpha"'] \ar[rr,"\phi_{\alpha\beta}"] \& \& \mathcal{V}_X(X_\beta)\ar[ld,"\phi_\beta"] \\
\& \colim (\mathcal{V}_X) \&
\end{tikzcd}
\]
for each pair $X_\beta \subseteq X_\alpha$.
By abuse of notation we will use $v_\alpha$ for both $v_\alpha \in \mathcal{V}_X(X_\alpha)$ and its image in $ \sqcup_{X_\alpha}\mathcal{V}_X(X_\alpha)$. 
Hence we have 
\[
\colim(\mathcal{V}_X)={}^{\sqcup_{X_\alpha}\mathcal{V}_X(X_\alpha)}/_ { \sim}
\]
where $\sim$ is the reflexive, symmetric and transitive closure of the relation $v_\alpha\sim v_\beta$ if $\phi_{\alpha\beta}(v_\alpha) = v_\beta$. See Section~\ref{sec:colim} for more on colimits in $\Set$.
We use $[v_\alpha]$ to denote the equivalence class of $v_\alpha$ in $\colim(\mathcal{V}_X)$.
Hence we have $\phi_\alpha\colon \VV_X(X_\alpha)\to \colim(\VV_X)$, $v_\alpha\mapsto [v_\alpha]$. 

Notice that this relation is certainly not itself an equivalence.
For example for any $X_\beta\subset X_\alpha$, and $v_\beta=v_\alpha \circ \iota_{\beta\alpha}\colon \pi(X,X_\beta) \to G$, then the relation says $v_\beta=\phi_{\alpha\beta}(v_\alpha)\sim v_\alpha$ but not $v_\alpha \sim v_\beta$ as there is no map $\phi_{\beta \alpha}$.

\begin{lemma}\label{le:surjective}
	Let $\VV_X\colon \FinSetX \to \Set$ be as in Lemma~\ref{le:functor_V}, then we have that all maps $\phi_\alpha\colon \VV_X(X_\alpha)\to \colim(\VV_X)$ are surjections.
\end{lemma}
\begin{proof}
	Fix some $\mathcal{V}_X(X_\alpha)$.
	We must show that every equivalence class $[v]\in \colim (\mathcal{V}_X)$ has a representative in $\mathcal{V}_X(X_\alpha)$.
	Certainly $[v]$ has a representative $v_\beta$ in some $\mathcal{V}_X(X_\beta)$. 
	Let $X_\gamma = X_\alpha \cup X_\beta$ and choose $v_\gamma \in \mathcal{V}_X(X_\gamma)$ with $\phi_{\gamma \beta}(v_\gamma)=v_\beta$, which is always possible since $\phi_{\gamma \beta}$ is a surjection by Lemma~\ref{le:tqft_VVepi}.
	Now $v_\alpha=\phi_{\gamma \alpha}(v_\gamma)$ is a representative for $[v]$ since $v_\gamma \sim v_\beta$ and $v_\gamma \sim v_\alpha$.
\end{proof}

\begin{lemma}\label{le:finite_colim}
	Let $\VV_X\colon \FinSetX \to \Set$ be as in Lemma~\ref{le:functor_V}. The set $\colim (\mathcal{V}_X)$ is finite.
\end{lemma}

\begin{proof}
	The groupoid $\pi(X,X_\alpha)$ is finitely generated since $X$ is a \homfin{} space. Also $G$ is finite, 
	hence the $\mathcal{V}_X(X_\alpha)$ is finite for all $X_\alpha$, so with Lemma~\ref{le:surjective} we have the result.
\end{proof}

It is well known that there is an adjunction between the categories $\VectC$ and $\Set$ \cite{maclane}.

	\exa{\label{ex:functor_VectSet}
	There is a forgetful functor $\mathrm{U_{V_\mathbbm{k}}}\colon \Vectk \to \Set$ which sends a vector space $V$ to the set of all vectors and a linear map to the corresponding set map determined by its action on vectors.}
\lemm{\label{le:left_adjoint_Vectforget}
	Consider the functor $\mathrm{U_{V_{\mathbbm{k}}}}\colon\Vectk\to \Set$ as in Example~\ref{ex:functor_VectSet}.\\
	(I) There is a functor $\mathrm{F_{V_{\mathbbm{k}}}}\colon \Set \to \Vectk$ which sends a set $X$ to the free vector space over $X$ and a function to the unique linear map between the corresponding free vector spaces which acts in the same way on basis elements.\\
	(II) The functor $\mathrm{F_{V_{\mathbbm{k}}}}$ is left adjoint to $\mathrm{U_{{V_\mathbbm{k}}}}\colon\Vectk\to \Set$.
}
\begin{proof}
	Straightforward.
\end{proof}
 
The previous lemma implies $F_{V_{\mathbb{C}}}$ preserves colimits (see Lemma~\ref{le:left_adjoint_Vectforget}).
This means we can equivalently define a map on objects in $\HomCob$ terms of the vector space with basis $\colim(\VV_X)$, or as the colimit of the image $F_{V_{\mathbb{C}}}\circ \VV_X$ in $\VectC$. In general the colimit in $\Set$ will be easier to work with.

\defn{ \label{de:FG_obj}
	For $X\in\chi $ define
\[
\FG(X)= \colim(\mathcal{V}_X')=\C (\colim(\mathcal{V}_X))
\]
where $\mathcal{V}_X' = F_{V_{\mathbb{C}}}\circ \mathcal{V}_X $ and  $\VV_X\colon \FinSetX \to \Set$ is as in Lemma~\ref{le:functor_V}.}

\subsubsection*{Magmoid morphism $\FG\colon\cHomCob\to \vVectC$}

The linear map $\bbFG$ assigned to a \bhomcob{} still depends on the basepoints in objects.
Here we define a map $\FG$ which assigns to a \bhomcob{}, $\cbmor{M}{i}{X}{j}{Y}$, a linear map $ \FG(X)\to \FG(Y)$. We will then show that this linear map is also be independent of $X_0$ and $Y_0$.

Let $(X,X_\alpha),(X,X_\beta)\in \chip$.
In the previous section we constructed
$\VV_X\colon \FinSetX \to \Set$ (Lemma~\ref{le:functor_V}) which sends inclusions $\iota_{\beta \alpha}\colon X_\beta \to X_\alpha$ to maps 
$\phi_{\alpha\beta}\colon \mathcal{V} \left(X_\alpha\right)\to\mathcal{V} \left(X_\beta\right)$.
Notice that $F_{V_{\mathbb{C}}}\circ\mathcal{V} (X_\alpha)=\mathcal{V}'(X_\alpha)=\bFG(X,X_\alpha)$,
so we have a map $F_{V_{\mathbb{C}}}(\phi_{\alpha\beta})\colon \bFG(X,X_\alpha)\to \bFG(X,X_\beta)$.
By abuse of notation we will also use $\phi_{\alpha\beta}$ to refer to the maps $F_{V_{\mathbb{C}}}(\phi_{\alpha\beta})$.
In this section we will need to vary the input space in the construction of $\VV_X$, thus we add a superscript denoting the space, so we have maps
\[\phi^X_{\alpha\beta}\colon\bFG\left(X,X_\alpha\right)\to\bFG\left(X,X_\beta\right).
\]

\begin{lemma}\label{le:bbFG_cocone}
	Let $\cmor{M}{i}{X}{j}{Y}$ be a \chomcob{}.
	Then for any pair $X_\alpha, X_\beta\subseteq X$ with $X_\beta \subseteq X_\alpha$, and \cbhomcob{}s $\scbmor{M}{M_{\alpha\alpha^{\lprime}}}{i}{X}{X_\alpha}{j}{Y}{Y_{\alpha^{\lprime}}}$ 
	and $\scbmor{M}{M_{\beta\alpha^{\lprime}}}{i}{X}{X_\beta}{j}{Y}{Y_{\alpha^{\lprime}}},$
	the following diagram commutes
	\[
	\begin{tikzcd}[ampersand replacement =\&,column sep=0, row sep=large]
	\bFG\left(X,X_\alpha\right)\ar[dr, "{\bbFG(M,X_\alpha,Y_{\alpha^{\lprime}})}"'] \ar[rr, "\phi_{\alpha\beta}^{X}"] \& \& \bFG\left(X,X_\beta\right) \ar[dl, "{\bbFG(M,X_\beta,Y_{\alpha^{\lprime}})}"]\\
	\& \bFG(Y,Y_{\alpha^{\lprime}}) \& 
	\end{tikzcd}.
	\]
	 That is, the maps $\bbFG$ form a cocone over the vector spaces $\bFG(X,X_\alpha)$ and the maps $\phi^{X}_{\alpha\beta}$.
	Hence 
	there is a unique map \[d^M_{\alpha^{\lprime}}\colon \FG(X) \to \bFG(Y,Y_{\alpha^{\lprime}}).\]
\end{lemma}

\begin{proof}
	First suppose $X_\alpha=X_\beta \cup \{x\}$ for some $x\notin X_\beta$.
	Let $f \in \bFG(X,X_\alpha)$ and $ g\in \bFG(Y,Y_{\alpha^{\lprime}})$ be basis elements. 
	We have 
	\[
	\left\langle g\middle| \bbFG(M,X_\alpha,Y_{\alpha^{\lprime}})\middle| f \right\rangle = 
	|G|^{-(|M_{\alpha\alpha^{\lprime}}|-|X_\alpha|)}
	\left\langle g \,\middle|\, \bFG (M,M_{\alpha\alpha^{\lprime}}) \,\middle|\, f \right\rangle
	\]
	and
	\[
	\left\langle g\middle| \bbFG(M,X_\beta,Y_{\alpha^{\lprime}})\middle|\phi_{\alpha\beta}^X(f)\right\rangle 
	= |G|^{-(|M_{\beta\alpha^{\lprime}}| - |X_\beta|)} \left\langle g \,\middle|\, \bFG(M,M_{\beta\alpha^{\lprime}}) \,\middle|\, \phi_{\alpha\beta}^{X}(f) \right\rangle
	\]
	for appropriate choices $M_{\alpha\alpha^{\lprime}}$ and $M_{\beta\alpha^{\lprime}}$.
	We may choose $M_{\alpha\alpha^{\lprime}}=M_{\beta\alpha^{\lprime}} \cup \{x\}$.
	There is a map from $\left\langle g \,\middle|\, \bFG (M,M_{\alpha\alpha^{\lprime}}) \,\middle|\, f \right\rangle$ to $\left\langle g \,\middle|\, \bFG(M,M_{\beta\alpha^{\lprime}}) \,\middle|\, \phi_{\alpha\beta}^{X}(f) \right\rangle$ given by taking the restriction of a map $h\colon \pi(M,M_{\alpha\alpha^{\lprime}}) \to G$ to $h'=h|_{\pi(M,M_{\beta\alpha^{\lprime}})}$.
	Note also that if $h|_{\pi(X,X_\alpha)}=f$ and $h|_{\pi(Y,Y_{\alpha'})}=g$, then also $h'|_{\pi(X,X_\beta)}=\phi_{\alpha\beta}^{X}(f)$ and $h'|_{\pi(Y,Y_{\alpha'})}=g$.

	This map has inverse given by extending any map $\tilde{h}:\pi(M,M_{\beta\alpha^{\lprime}}) \to G$ to map $\tilde{h'}:\pi(M,M_{\alpha\alpha'}) \to G$, which sends a path $\gamma:x \to x'$ in $\pi(X,X_\alpha)$, with $x'\in X_\beta$, to $f(\gamma)$, as in Lemma~\ref{le:extension}. If the map $\tilde{h}$ satisfies $\tilde{h}|_{\pi(X,X_\beta)}=\phi_{\alpha\beta}^{X}(f)$ and $\tilde{h}|_{\pi(Y,Y_{\alpha'})}=g$, then $\tilde{h}'|_{\pi(X,X_\alpha)}=f$ and $\tilde{h}'|_{\pi(Y,Y_{\alpha'})}=g$.
	Hence 
	\[
	\langle g \,|\, \mathcal{F}_{G}^{!} (M,M_{\alpha\alpha^{\lprime}}) \,|\, f \rangle=\langle g \,|\, \mathcal{F}_{G}^{!} (M,M_{\beta\alpha^{\lprime}}) \,|\, \phi_{\alpha\beta}^X(f) \rangle.
	\]
	Also $|M_{\beta\alpha^{\lprime}}| - |X_\beta| = |M_{\alpha\alpha^{\lprime}}+1| - |X_\alpha+1|=|M_{\alpha\alpha^{\lprime}}|-|X_\alpha|$.
	The set $X_\alpha\setminus X_\beta$  is finite, so we can repeat the same process for all $\{x_1,...,x_n\}\in X_\alpha\setminus X_\beta$.
\end{proof}

\rem{\label{rem:Yetter}Notice that, had we defined the normalisation to be $|G|^{-(M_0-\nicefrac{1}{2}(X_0+Y_0))}$, as is Yetter's convention, the triangle in the previous Lemma would not be commutative. The fix is to redefine the maps $\phi_{\alpha\beta}$ in such a way that they no longer send basis elements to basis elements. This complicates the picture slightly. It is straighforward to see that each choice leads to the same image on any cospan of the form $\emptyset \to M\leftarrow \emptyset$.}

\begin{lemma}\label{le:subset_ind}
	Let $\cmor{M}{i}{X}{j}{Y}$ be a \chomcob{}.
	Fix a choice of $Y_{\alpha'}\subseteq Y$ such that $(Y,Y_{\alpha'})\in \chip$. 
	For each pair $X_\alpha, X_\beta\subseteq X$ such that $(X,X_\alpha),(X,X_\beta)\in \chip$ we have the following diagram
	\begin{align}
	\begin{tikzcd}[ampersand replacement =\&,column sep=tiny]\label{cocone}
	\bFG\left(X,X_\alpha\right)\ar[ddr, bend right, "{\bbFG(M,X_\alpha,Y_{\alpha^{\lprime}})}"'] \ar[rr, "\phi_{\alpha\beta}^{X}"] \ar[dr, "\phi_\alpha^X"'] \& \& \bFG\left(X,X_\beta\right) \ar[ddl, bend left, "{\bbFG(M,X_\beta,Y_{\alpha^{\lprime}})}"] \ar[dl, "\phi_\beta^X"] \\
	\& \FG(X) \ar[d,"d^M_{\alpha^{\lprime}}"'] \& \\
	\& \bFG(Y,Y_{\alpha^{\lprime}}) \ar[dr,bend right, "\phi_{\alpha^{\lprime}}^Y"']  \& \\
	\& \&\FG(Y).
	\end{tikzcd}
	\end{align}
	The assignment 
	\[
	\FG\left(
	\cmortikz{M}{i}{X}{j}{Y}\right)=\phi_{\alpha^{\lprime}}^{Y}d^M_{\alpha^{\lprime}}
	\]
	does not depend on the choice of $Y_{\alpha'}$. \\
	As above, where we have given a cospan we will use the notation $\FG(M)$ for $\FG(\cmor{M}{i}{X}{j}{Y})$.
\end{lemma}
\begin{proof}
	We show that the following diagram commutes for any pair $Y_{\alpha^{\lprime}}, Y_{\beta^{\lprime}}$
	\[
	\begin{tikzcd}[ampersand replacement =\&,column sep = 0, row sep =large]
	\& \bFG(X,X_\alpha)\ar[dl, "{\bbFG(M,X_\alpha,Y_{\alpha^{\lprime}})}"'] \ar[dr, "{\bbFG(M,X_\alpha,Y_{\beta^{\lprime}})}"]\& \\
	\bFG\left(Y,Y_{\alpha^{\lprime}}\right) \ar[rr, "\phi_{\alpha^{\lprime}\beta^{\lprime}}^{Y}"'] \& \& \bFG\left(Y,Y_{\beta^{\lprime}}\right) 
	\end{tikzcd}
	\]
	This implies that $\phi_{\alpha^{\lprime}\beta^{\lprime}}^{Y}$ is a map of cocones and, by the universal property of the colimit, that $\phi_{\alpha^{\lprime}\beta^{\lprime}}^{Y}d_{\alpha^{\lprime}}^M=d^M_{\beta^{\lprime}}$
	and hence that
	$\phi^Y_{\alpha^{\lprime}} d_{\alpha^{\lprime}}^{M}=\phi^Y_{\beta^{\lprime}}\phi_{\alpha^{\lprime}\beta^{\lprime}}^{Y}d_{\alpha^{\lprime}}^M=\phi^Y_{\beta^{\lprime}} d^M_{\beta^{\lprime}}$.
	
	Suppose first that $Y_{\alpha^{\lprime}}=Y_{\beta^{\lprime}} \cup \{y\}$ for some $y\notin Y_{\beta^{\lprime}}$ and let $f\in\bFG(X,X_\alpha)$ and $g\in\bFG(Y,Y_{\beta^{\lprime}})$ be a basis elements.
	The map $\phi^{Y}_{\alpha^{\lprime},\beta^{\lprime}}\colon \mathcal{V}_Y(Y_{\alpha^{\lprime}})\to \mathcal{V}_Y(Y_{\beta^{\lprime}}) $ is surjective (by Lemma~\ref{le:tqft_VVepi}), so sends a subset of $\mathcal{V}_Y(Y_{\alpha^{\lprime}})$ to $g\in \mathcal{V}_Y(Y_{\beta^{\lprime}})$.
	Thus the matrix element $\lan f \vert \phi_{\alpha^{\lprime}\beta^{\lprime}}^{Y} \bbFG(M,X_\alpha,Y_{\alpha^{\lprime}})\vert g \ran$ is the sum of the matrix elements in $\bbFG(M,X_\alpha,Y_{\alpha^{\lprime}})$ corresponding to $f$ and each $g'$ in the preimage ${\phi_{\alpha^{\lprime}\beta^{\lprime}}^{Y}}^{-1}(g)$.
	Hence we have 
	\begin{align*}
	\left\langle g \,\middle|\, \phi_{\alpha^{\lprime}\beta^{\lprime}}^{Y} \bbFG (M,X_\alpha,Y_{\alpha^{\lprime}}) \,\middle|\, f \right\rangle
	&=\sum_{g' \in \phi_{\alpha^{\lprime}\beta^{\lprime}} ^{Y-1}(g)} \left\langle g' \,\middle|\, \bbFG (M,X_\alpha,Y_{\alpha^{\lprime}}) \,\middle|\, f \right\rangle\\
	&=|G| ^{-(|M_{\alpha\alpha^{\lprime}}| - |X_\alpha|)} \sum_{g' \in \phi_{\alpha\beta} ^{Y-1}} \left\langle g' \,\middle|\, \bFG (M,M_{\alpha\alpha^{\lprime}}) \,\middle|\, f \right\rangle 
	\end{align*}
	for an appropriate choice of $M_{\alpha\alpha^{\lprime}}$.
	Following the same argument as used in the previous lemma 
	we may choose a subset $M_{\alpha\beta^{\lprime}}$ with $M_{\alpha\alpha^{\lprime}}=M_{\alpha\beta^{\lprime}}\cup\{y\}$
	and then
	\[
	\left\langle g \,\middle|\, \bFG (M,M_{\alpha\beta^{\lprime}}) \,\middle|\, f \right\rangle
	=\left\langle g' \,\middle|\, \bFG(M,M_{\alpha\alpha^{\lprime}}) \,\middle|\, f \right\rangle .
	\]
	For every map $g:\pi(Y,Y_{\beta^{\lprime}}) \to G$, there will be precisely $G$ maps in the preimage under $\phi_{\alpha^{\lprime}\beta^{\lprime}}^Y$, one for each choice of an element of $G$. This can be seen by noting that $\phi_{\alpha^{\lprime}\beta^{\lprime}}^Y$ is the composition of the bijection $\Theta_\gamma^{-1}\colon\Grpd (\pi(X,Y_{\alpha'}),G)  \to \Grpd(\pi(Y,Y_{\beta'}),G)\times G$ in Lemma~\ref{le:add_bps} with the projection to the first coordinate, for some choice of $\gamma\colon y\to y'\in M_{\alpha\beta'}$.
	Hence we have 
	\begin{align*}
	\left\langle g \,\middle|\, \phi_{\alpha^{\lprime}\beta^{\lprime}}^{Y} \bbFG (M,X_\alpha,Y_{\alpha^{\lprime}}) \,\middle|\, f \right\rangle
	&=|G| ^{-(|M_{\alpha\alpha^{\lprime}}| - |X_\alpha|)} |G|	\left\langle g \,\middle|\, \bFG (M,M_{\alpha\alpha^{\lprime}}) \,\middle|\, f \right\rangle \\
	&=|G| ^{-(|M_{\alpha\beta^{\lprime}}| - |X_\alpha|)}	\left\langle g \,\middle|\, \bFG (M,M_{\alpha\beta^{\lprime}}) \,\middle|\, f \right\rangle\\
	&=	\left\langle g \middle| \bbFG(M,X_\alpha,Y_{\beta^{\lprime}})\middle| f\right\rangle.
	\end{align*}
	Now suppose $Y_{\alpha^{\lprime}} = Y_{\beta^{\lprime}} \cup \{y_1,...,y_n\}$, then we similarly acquire one factor of $\vert G\vert$ and one factor $\vert G \vert^{-1}$ for each new point, hence $\phi_{\alpha^{\lprime}\beta^{\lprime}}^{Y} \bbFG (M,X_\alpha,Y_{\alpha^{\lprime}})=\bbFG(M,X_\alpha,Y_{\beta^{\lprime}})$.
\end{proof}
\begin{lemma}\label{le:tqft_FGmagmor}
	We have a magmoid morphism 
	\[
	\FG\colon \cHomCob\to \vVectC
	\]
	where $\FG$ is given in Definition~\ref{de:FG_obj} and Lemma~\ref{le:subset_ind}.
\end{lemma}
\begin{proof}
	Lemmas~\ref{le:bbFG_mor_bps} and \ref{le:subset_ind} give that $\FG$ is well defined.
	
	We prove $\FG$ preserves composition.
	Suppose we have \chomcob{}s 
	$\cmor{M}{i}{X}{j}{Y}$ and $\cmor{N}{k}{Y}{l}{Z}$.
	Let $Y_0\subseteq Y$ and $Z_0\subseteq Z$ be fixed finite representative subsets.
	Notice that for any finite representative subset $X_0\subseteq X$, by Lemma~\ref{le:tqft_compbbFG}, we have  $\bbFG(M\sqcup N,X_0,Z_0)=\bbFG(N,Y_0,Z_0)\bbFG(M,X_0,Y_0)=d_0^N\phi^Y_0\bbFG(M,X_0,Y_0)$.
	Thus $d_0^N\phi^Y_0 d_0^M\colon \FG(X)\to \bFG(Z,Z_0)$ is a map commuting with the cocone given by the $\bbFG(M\sqcup N,X_0,Z_0)$, where $X_0$ is varying.
	Hence by the uniqueness of the map obtained from the universal property of the colimit, we have $d_0^N\phi^Y_0 d_0^M=d_0^{M\sqcup_{Y}N}$.
	Hence $\phi^Z_0d_0^N\phi^Y_0 d_0^M=\phi^Z_0d_0^{M\sqcup_{Y}N}$ and 
	$\FG(N)\FG(M)=\FG(M\sqcup_{Y}N)$.
\end{proof}	
\subsubsection*{The functor $\FG\colon \HomCob\to \VectC$}
The following theorem says that $\FG$ becomes a functor from the category $\HomCob$.
\begin{theorem}\label{th:FG_functor}
	There is a functor 
	\[\FG \colon \HomCob \to \VectC
	\]
	 defined as follows.
	 \begin{itemize}
	 	\item For a space $X\in Ob(\HomCob)$,
	 	\[
	 	\FG(X) = \mathbb{C}\left( \colim(\V_X)  \right)
	 	\]
	 	where $\V_X$ is the diagram in $\Set$
	 	with vertices
	 	$
	 	\V_X (X_\alpha) = \Grpd\left( \pi(X,X_\alpha) ,G \right) 
	 	$ for each finite representative subset $X_\alpha\subseteq X$
	 	and edges 
	 	$\phi_{\alpha\beta}\colon\V_X(X_\alpha) \to \V_X(X_\beta)$ whenever $X_\beta \subseteq X_\alpha$, sending each $f\in \V_X (X_\alpha)$ to $f\circ \iota_{\beta\alpha}$ where $\iota_{\beta\alpha}\colon \pi(X,X_\beta)\to \pi(X,X_\alpha)$ is the inclusion.
	 	\item For a \homcob{} $\classche{\cmor{M}{i}{X}{j}{Y}}$,
	 	\[
	 	\FG\left(\classche{\cmortikz{M}{i}{X}{j}{Y}}\right)=
	 	\FG\left(
	 	\cmortikz{M}{i}{X}{j}{Y}\right)=\phi_{\alpha^{\lprime}}^{Y}d^M_{\alpha^{\lprime}}\colon \FG(X)\to \FG(Y)
	 	\]
	 	where $Y_{\alpha'}\subseteq Y$ is some choice of finite representative subset and, $\phi_{\alpha^{\lprime}}^{Y}$ and $d^M_{\alpha^{\lprime}}$ are as in Lemma~\ref{le:subset_ind}.
	 \end{itemize}
\end{theorem}
\begin{proof}
	We have from Lemma~\ref{le:tqft_FGmagmor} that $\FG$ is a magmoid morphism so it remains only to check that $\FG$ does not depend on a choice of representative cospan and that it preserves identities.
	We will need a different interpretation of the colimit to prove that $\FG$ preserves identities, we do this in Lemma~\ref{le:identity_preserved}.
	
	In Lemma \ref{bFG_well_defined} we show that $\bFG$ does not depend on the representative \homcob{} we choose. It thus follows that $\bbFG$ and hence $\FG$ do not depend on a choice of representative cospan.
\end{proof}

The following Lemma gives an alternative description of the image of the linear map a cospan is sent to under $\FG$, in terms of a choice of based cospan.
\begin{lemma}\label{le:FG_sum}
	Let $\cmor{M}{i}{X}{j}{Y}$ be a \chomcob{},
	 $\cbmor{M}{i}{X}{j}{Y}$ a choice of \cbhomcob{},
	and $[f]\in \FG(X)$ and $[g]\in\FG(Y)$ be basis elements (so $[f]$, for example, is an equivalence class in $\colim(\VV_X)$), then
	\begin{align*}
	\langle [g] | \FG(M) | [f] \rangle&=
	|G|^{-(|M_0| - |X_0|)} \hspace*{-1em}\sum_{g \in \phi^{Y-1}_0\left([g]\right)}
	\hspace*{-1em}\left\vert\left\{  h\colon\pi(M,M_0) \to G  \,|\,h|_{\pi(X,X_0)}=f \wedge h|_{\pi(Y,Y_0)}=g \right\}\right\vert\\
	&=|G|^{-(|M_0|-|X_0|)}\sum_{g \in \phi^{Y-1}_0\left([g]\right)}\hspace*{-1em}
	\left\langle g\,\middle|\, \bbFG(M,M_0) \,\middle|\, f \right\rangle
	\end{align*}
	where $\phi_0^{Y}\colon \bFG(Y,Y_0)\to \FG(Y)$ is the map into $\colim(\VV'_Y)$; see Definition~\ref{de:FG_obj}.
\end{lemma}
\begin{proof}
	We will use notation as in \eqref{cocone}.
	Since each map $\phi_0^{Y}$ is surjective (Lemma~\ref{le:surjective}), we can find $d^M_{0}([f])$ by looking at $\bbFG(M,X_0,Y_0)(f)$.
	Hence we have
	\ali{
	d^M_0 ([f])
	&=
	\sum_{g\in \mathcal{V}_Y(Y_0)} 
	\left\langle g\,\middle|\, \bbFG(M,X_0,Y_0) \,\middle|\, f \right\rangle
	\left.\middle| g\right\rangle
	}
	and, choosing a basis element $[g]\in \FG(Y)$,
	\begin{align*}
	\langle [g] | \FG(M) | [f] \rangle
	&=\sum_{g \in \phi^{Y-1}_0\left([g]\right)}\hspace*{-1em}
	\left\langle g\,\middle|\, \bbFG(M,X_0,Y_0) \,\middle|\, f \right\rangle \\
	&= |G|^{-(|M_0|-|X_0|)}\sum_{g \in \phi^{Y-1}_0\left([g]\right)}\hspace*{-1em}
	\left\langle g\,\middle|\, \bFG(M,M_0) \,\middle|\, f \right\rangle \\
	\begin{split}
	&=|G|^{-(|M_0| - |X_0|)} \hspace*{-1em}\sum_{g \in \phi^{Y-1}_0\left([g]\right)}
	\hspace*{-1em}\left\vert\left\{  h\colon\pi(M,M_0) \to G  \,|\, h|_{\pi(X,X_0)}=f \; \wedge\;  h|_{\pi(Y,Y_0)}=g \right\}\right\vert. 
	\end{split}\\
	&&\qedhere	
	\end{align*}
\end{proof}
\begin{remark}\label{FG_equivalence}
	The set of maps $\phi_0^{-1}([g])$ contains all maps $g'\colon \pi(Y,Y_0)\to G$ such that $g'\sim g$ where $\sim$ is the equivalence relation defined by the colimit. 
	Since we are only counting the cardinality of maps $h$ we can rewrite the map on morphisms as 
	\begin{align}\label{eq:FG_sim}
	\langle [g] | \FG(M) | [f] \rangle \hspace*{-0.2em}=\hspace*{-0.2em} 
	|G|^{-(|M_0| - |X_0|)} \left\vert\left\{  h:\pi(M,M_0) \to G  \,|\, h|_{\pi(X,X_0)}=f \wedge  h|_{\pi(Y,Y_0)}\sim g \right\}\right\vert
	\end{align}
	where we have removed the sum and only insist maps $h$ are equivalent to $g$ on $Y$.
	In many cases, especially with the local equivalence obtained in following section, this will be the most useful formulation to use for calculations.
\end{remark}

\exa{\label{ex:FGmer}
	Let $\cmor{M}{i}{X}{j}{Y}$ be the homotopy cobordism shown in Figure~\ref{fig:mer}. Note this is a homotopy cobordism from Examples~\ref{ex:mer} and \ref{ex:mer_b}.
	Using \eqref{eq:FG_sim}, we may choose to calculate the image of  $\FG(\classche{\cmor{M}{i}{X}{j}{Y}})$ using the \bhomcob{} considered in Example~\ref{ex:bFG_mer}.
	Using the results and notation from Example~\ref{ex:bFG_mer}, we have 
	\ali{
	\lan[(g_1)] \; \vert\; \FG(M)\; \vert\; [(f_1,f_2)]\ran
	&=|G|^{-1}\lan(g_1) \; \vert\; \bFG(M,M_0)\; \vert\; (f_1,f_2)\ran \\
	&= |G|^{-1}\;\; \vert \left\{c,d\in G \vert c^{-1} d^{-1}f_2d f_1c\sim g_1 \right\} \vert. 
	}	
}

\subsection{Map \texorpdfstring{$\FG\colon Ob(\HomCob)\to Ob(\Vect_{\mathbb{C}})$}{ZG:HomCob to Vect} in terms of a local equivalence relation}\label{sec:local_equiv}

For a general \homfin{} space $X$ it is 
unlikely to be straightforward to calculate the colimit constructed in the previous section.
Usually there will be an uncountably infinite number of choices of finite representative subsets $X_\alpha\subseteq X$, and thus an uncountably infinite number of vertices in $\VV_X$.
Although, in Lemma~\ref{le:finite_colim}, we proved that $\FG(X)$ is finite dimensional for all $X$.
In this section we show that this global equivalence, given by taking the colimit over all choices of subsets, is the same as choosing a single subset and taking a local equivalence given by taking maps up to natural transformation. 
This will allow us to prove, in Lemma~\ref{le:identity_preserved},
that $\FG$ preserves the identity. We will also need this interpretation of $\FG$ to prove, in Section~\ref{sec:FG_mon}, that $\FG$ is a monoidal functor.

\medskip

Here we only need to work with a single space $X$, so with $\V_X$ as constructed in Lemma~\ref{le:functor_V}, we drop the subscript on $\V_X$, and the superscript on the $\phi^X$. Consider the commuting diagram,
\[
\begin{tikzcd}[column sep= small]
	{}^{\V(X_\alpha)}/_{\cong}\ar[drr,"\hat{\phi}_\alpha"',dashed, bend right=20] \&[+1.2em] \V(X_\alpha)\ar[rr,"\phi_{\alpha\beta}"]\ar[l,"p_\alpha"']\ar[dr,"\phi_\alpha"',bend right=15] \& \& \V(X_\beta)\ar[dl, "\phi_\beta",bend left=15] \\
	\& \& \colim(\V)
\end{tikzcd}
\]
where $\cong$ denotes the relation obtained by taking maps up to natural isomorphism (it is straightforward to check this is an equivalence relation). The set map $p_\alpha$ sends a groupoid map in $\Grpd(\pi(X,X_\alpha),G)$ to its equivalence class in ${}^{\V(X_\alpha)}/_{\cong}$.  
The map $\hat{\phi}_\alpha:{}^{\V(X_\alpha)}/_{\cong} \to \colim(\mathcal{V})$ is the canonical map sending an equivalence class to $\phi_\alpha$ of some representative (it remains to check this is well defined).
\begin{theorem}\label{natiso_to_colim}
	For  a space $X$, the map $\hat{\phi}_\alpha$ is an isomorphism.
	Hence, for a \homfin{} space $X$
	\[ \FG(X)=\C((\Grpd(\pi(X,X_0),G)/\cong),\] for any choice $X_0\subset X$ of finite representative subset, where $\cong$ denotes taking maps up to natural transformation.
\end{theorem}
\begin{proof}
	We prove $\hat{\phi}_\alpha$ is well defined and injective in Lemmas~\ref{well_defined} and \ref{injective} respectively. Surjectivity follows directly from Lemma~\ref{le:surjective}. 
\end{proof}

For a path $s\colon \II \to X$ in $X$, we will also use $s$ to denote its path equivalence class in $\pi(X)$ and $s\simp s'$ to mean that $s'\in\classp{s}$.

\begin{lemma} \label{well_defined}
	Let $v_\alpha,v_\alpha' \in \mathcal{V}(X_\alpha)$ be two groupoid maps such that $p_\alpha(v_\alpha) = p_\alpha(v_\alpha')$, then $\phi_\alpha(v_\alpha) = \phi_\alpha(v_\alpha')$.
\end{lemma}

\begin{proof}
	There exists a subset $X_{\tilde{\alpha}}\subseteq X_\alpha$ containing precisely one basepoint in each path-connected component, and maps $\tilde{v}_\alpha, \tilde{v}_\alpha'\colon\pi(X,X_{\tilde{\alpha}})\to G$ such that $\phi_{\alpha\tilde{\alpha}}(v_\alpha)=\tilde{v}_\alpha$ and $\phi_{\alpha\tilde{\alpha}}(v_\alpha')=\tilde{v}_\alpha'$.
	We will show that $\tilde{v}_\alpha$ and $\tilde{v}_\alpha'$ are equivalent in the colimit, implying $v_\alpha\sim v_\alpha'$.
	The idea of this proof is illustrated by the following diagram.
	\[
	\begin{tikzcd}[column sep = small]
	\& \pi(X,X_\gamma)\ar[ddr,"v_\gamma"'] \&\& \pi(X, X_\gamma)\ar[ddl,"v_\gamma'"] \& \\
	\pi(X,X_{\tilde{\alpha}})\ar[ur]\ar[drr,"\tilde{v}_{{\alpha}}"'] \&\& \pi(X,X_\beta)\ar[ul]\ar[ur]\ar[d] \&\& \pi(X,X_{\tilde{\alpha}})\ar[ul]\ar[dll,"\tilde{v}_{{\alpha}}'"] \\
	\& \& G \& \& 
	\end{tikzcd}
	\]
	We use the morphisms in the natural transformation connecting $v_\alpha$ and $v_\alpha'$ to extend the map $\tilde{v}_\alpha$ to a map from $v_\gamma\colon \pi(X,X_\gamma)\to G$, where $X_\gamma$ is a larger set of basepoints.
	We also trivially extend the map $\tilde{v}_\alpha'$  to $v_\gamma'\colon \pi(X,X_\gamma)\to G$, and
	show that these extensions have the same image under some $\phi_{\gamma\beta}$, and therefore are equivalent in the colimit.
	
	The set $X_{\tilde{\alpha}}$ is finite so we can write $X_{\tilde{\alpha}}=\{x_1,...,x_N\}$. 
	Recall that $X_{\tilde{\alpha}}$ contains one point in each path-component, thus all morphisms in $\pi(X,X_{\tilde{\alpha}})$ are represented by loops.
	Since $v_\alpha$ and $v_\alpha'$ are related by a natural transformation, for all points $x_n \in X_{\tilde{\alpha}}$ and for all equivalence classes of loops $s\colon x_n \to x_n$, the below square commutes.
	\[
	\begin{tikzcd}
		v_\alpha(x_n) \ar[r,"v_\alpha(s)"] \ar[d,"\eta_{x_n}"'] \& v_\alpha(x_n) \ar[d,"\eta_{x_n}"] \\
		v_\alpha'(x_n) \ar[r,"v_\alpha'(s)"'] \& v_\alpha'(x_n)
	\end{tikzcd}
	\]
	Recall that the image of $v_\alpha$ and $v_{\alpha'}$ is a groupoid with one object, so the image on points is always the same. Hence the two maps must be the same on any path-components that have no non-trivial paths.
	
	Choose another set of points $X_\beta=\{y_1,...,y_N\}$
	as follows.
	If there are no non-trivial loops based at $x_n$ then $y_n=x_n$,
	otherwise choose $y_n \neq x_n$ and, for each $n$, choose a path $t_n\colon x_n\to y_n$, with $t_n$ the constant path if $x_n=y_n$. This is always possible since a non-trivial loop based at $x_n$ must contain some $y_n\neq x_n$.
	
	Let $X_\gamma=X_{\tilde{\alpha}} \cup X_\beta$.
	We define a map $v_\gamma\colon \pi(X,X_\gamma) \to G$ as follows.
	Let $v_\gamma|_{\pi(X,X_{\tilde{\alpha}})}= \tilde{v}_{\alpha}$, and $v_\gamma(t_n) = \eta_{x_n}$ unless $t_n$ is the constant path, in which case $v_\gamma(t_n)=1_G$.
	By Lemma~\ref{le:extension} this completely defines $v_\gamma$.
	Notice $\phi_{\gamma\tilde{\alpha}}(v_\gamma)=\tilde{v}_{\alpha}$, hence $\tilde{v}_{\alpha} \sim v_\gamma$.
	
	Define another map $v_\gamma'\colon \pi(X,X_\gamma) \to G$ by $v_\gamma'|_{\pi(X,X_{\tilde{\alpha}})}=\tilde{v}_{\alpha}'$ and $v_\gamma'(t_n)=1_G$.
	We have $\phi_{\gamma\tilde{\alpha}}(v_\gamma')=\tilde{v}_{\alpha}'$ and so $\tilde{v}_\alpha' \sim v_\gamma'$.
	
	Now we check that $\phi_{\gamma\beta}(v_\gamma)=\phi_{\gamma\beta}(v'_\gamma)$, hence $v_\gamma\sim v_\gamma'$. Since $X_\beta$ has only one point in each path-connected component we only need to check that $v_\gamma$ and $v_{\gamma'}$ agree on loops.
	For any trivial  
	$s\colon x_n \to x_n$ with $y_n=x_n$, we have $v_\gamma(s)=1_G=\tilde{v}_\alpha(s)=\tilde{v}'_\alpha(s)$.
	
	Now suppose $s\colon y_n \to y_n$ is any class of loops with $y_n\neq x_n$,
	\[v_\gamma(s)=v_\gamma(t_n t_n^{-1} s t_n t_n^{-1}) = \eta_{x_n} \tilde{v}_{\alpha}(t_n^{-1} s t_n) 
	\eta_{x_n}^{-1}=\tilde{v}_\alpha'(t_n^{-1} s t_n)
	\]
	and similarly,
	\[v_\gamma'(s)=v_\gamma'(t_n t_n^{-1} s t_n t_n^{-1}) = v_\gamma'(t_n)v_\gamma'(t_n^{-1} s t_n) v_\gamma'(t_n^{-1}) = \tilde{v}_{\alpha}'(t_n^{-1} s t_n).
	\]
	Hence $\phi_{\gamma\beta}(v_\gamma)=\phi_{\gamma\beta}(v_\gamma')$ so $v_\gamma \sim v_\gamma'$ and $\tilde{v}_{\alpha} \sim \tilde{v}_{\alpha}'$. 
\end{proof}

\begin{lemma}\label{le:equivalence_pi(X)}
	For any finite representative subset $X_\alpha$ of a space $X$, $\pi(X,X_\alpha)$ and $\pi(X)$ are equivalent as categories.
\end{lemma}

\begin{proof}
	We have an inclusion
	$\iota_\alpha\colon \pi(X,X_\alpha)\to \pi(X)$. 
	We define explicitly a map $r_\alpha\colon \pi(X)\to \pi(X,X_\alpha)$ as follows.
	For each $x\in X\setminus X_\alpha$, choose a point $y_x \in X_\alpha$ in the same path-connected component as $x$, and a path $t_x\colon x \to y_x$. 
	If $x\in X_\alpha$ choose $y_x=x$ and $t_x$ the trivial path.
	Now define 
	\[
	r_\alpha(x)=	y_x 
	\]
	and for a path $s\colon x\to x'$ in $\pi(X)$
	\[
	r_\alpha(s)=t_{x'} s t_{x}^{-1}
	\]
	The composition $r_\alpha\iota_\alpha $ is equal to $1_{\pi(X,X_\alpha)}$ and a natural transformation $\eta\colon 1_{\pi(X)} \to \iota_\alpha r_\alpha$ is given by 
	\[
	\eta_{x}= t_x. \qedhere
	\]
\end{proof}

\begin{lemma}\label{injective}
	Let $v_\alpha,v_\alpha' \in \mathcal{V}(X_\alpha)$ be two maps such that $\phi_\alpha(v_\alpha)= \phi_\alpha(v_\alpha')$. Then $p_\alpha(v_\alpha) = p_\alpha(v_\alpha')$. 
\end{lemma}

\begin{proof}
	The maps $v_\alpha$ and $v_\alpha'$ being equivalent in the colimit means there is some finite sequence of relations $v_\alpha = v_0 \sim v_1 \sim ... \sim v_N = v_\alpha'$, where $v_n\neq v_{n+1}$, and of maps $v_n\colon \pi(X,X_n) \to G$ such that
	for each pair $v_n,v_{n+1}$ we have one of the following two diagrams:
	\[
	\begin{tikzcd}[ampersand replacement=\&]
	\pi(X,X_n) \ar[r,"\iota_{n,n+1}"] \ar[dr,"v_n=v_{n+1}\iota_{n,n+1}"']\& 
	\pi(X,X_{n+1}) \ar[d,"v_{n+1}"] \\
	\& G
	\end{tikzcd}
	\begin{tikzcd}[ampersand replacement =\&]
	\pi(X,X_{n+1}) \ar[r,"\iota_{n+1,n}"]\ar[rd,"v_{n+1}=v_n\iota_{n+1,n}"'] \& 
	\pi(X,X_n)\ar[d,"v_n"] \\
	\& G
	\end{tikzcd}
	\]
	Then we have a commuting diagram of the following form 
	
	\[
	\begin{tikzcd}[ampersand replacement=\&]
	\& \pi(X,X_{n,n+1}) \ar[dd] \\
	\pi(X,X_n) \ar[ur, "\iota'_n"] \ar[dr,"v_n"']\& \& \pi(X,X_{n+1}) \ar[ul,"\iota'_{n+1}"'] 
	\ar[dl,"v_{n+1}"] \\
	\& G \&
	\end{tikzcd}
	\]
	where we let $X_{n,n+1}$ be the larger of $X_n$ and $X_{n+1}$ and 
	one of $\iota'_n$ and $\iota'_{n+1}$ is a strict inclusion, and the other is the identity. The middle arrow is either $v_n$ or $v_{n+1}$.
	Consider the below (non-commuting) diagram
	\begin{adjustwidth}{-2cm}{-2cm}
	\[
	\scalebox{.8}{\begin{tikzcd}[cramped, column sep=0, row sep=tiny, ampersand replacement =\&] 
		\& \& \& \& \& \pi(X) 
		\ar[dddlllll, bend right=25, "r_{0}"', end anchor={[xshift=-4ex]}] 
		\ar[dllll, bend right=20,"r_{0,1}"', pos=0.6, end anchor={[xshift=0ex]}] 
		\ar[dddlll, bend right=15, "r_{1}"', end anchor={[xshift=0ex]}]
		\ar[dddl, bend right=10,"r_{n}"', end anchor={[xshift=-2ex]}]
		\ar[d,"r_{n,n+1}"',pos=0.6]
		\ar[dddr, bend left=10, "r_{n+1}"', end anchor={[xshift=2ex]}]
		\ar[dddrrr, bend left=15,"r_{N-1}"', end anchor={[xshift=0ex]}] 
		\ar[drrrr, bend left=20,"r_{N,N-1}"',pos=0.6, end anchor={[xshift=0ex]}] 
		\ar[dddrrrrr, bend left=25, "r_{N}"', end anchor={[xshift=4ex]}]
		\& \& \& \& \& \\ [+120pt]
		\& \pi(X,X_{0,1}) \ar[dddrrrr, bend right=25, start anchor={[xshift=-2ex]}] 
		\& \& \makebox{} \& \& \pi(X,X_{n,n+1}) \ar[ddd]
		\& \& \makebox{} \& \&  \pi(X,X_{N-1,N}) \ar[dddllll, bend left=25, start anchor={[xshift=+2ex]}] 
		\& \\
		\& \& \&  \makebox[2.5cm]{...}  \& \& \& \& \makebox[2.5cm]{...}  \& \& \& \\
		\pi(X,X_0) \ar[uur,"\iota'_0"] \ar[drrrrr, bend right=20,"v_0"', start anchor={[xshift=-4ex]},end anchor={[yshift=-.5em]}]
		\& \& \pi(X,X_1) \ar[uul,"\iota'_1"'] \ar[uur,end anchor={[xshift=-5ex,yshift=-2ex]}] \ar[drrr,"v_1"', bend right=15]
		\& \& \pi(X,X_n) \ar[uur,"\iota'_n"] \ar[uul,end anchor={[xshift=5ex,yshift=-2ex]}] \ar[dr,"v_n"', bend right=10, start anchor={[xshift=-2ex]}]
		\& \& \pi(X,X_{n+1}) \ar[uul,"\iota'_{n+1}"'] \ar[uur,end anchor={[xshift=-5ex,yshift=-2ex]}] \ar[dl,"v_n+1", bend left=10, start anchor={[xshift=2ex]}]
		\& \& \pi(X,X_{N-1}) \ar[uur,"\iota'_{N-1}"] \ar[uul,end anchor={[xshift=5ex,yshift=-2ex]}] \ar[dlll,"v_{N-1}", bend left=15] 
		\& \& \pi(X,X_N) \ar[uul,"\iota'_{N}"'] \ar[dlllll,"v_N", bend left=20, start anchor={[xshift=4ex]}, end anchor={[yshift=-.5em]}]
		\\ [+120pt]
		\& \& \& \& \& G  \& \& \& \& \&
		\end{tikzcd}}
	\]
\end{adjustwidth}
	where the maps $r$ 
	are as constructed in the proof of Lemma~\ref{le:equivalence_pi(X)}. 
	
	We will show there is a natural transformation $v_0 r_0$ to $v_N r_N$. Since $X_0=X_N=X_\alpha$, $\iota_0=\iota_N$ 
	 and hence $v_0 r_0\cong v_N r_N$ implies $v_0 r_0\iota_0 \cong v_N r_N \iota_N$. This implies $v_0 \cong v_N$ since $r_\beta\iota_\beta=1_{\pi(X,X_\beta)}$ for all finite representative $X_\beta\subseteq X$ by Lemma~\ref{le:equivalence_pi(X)}.
	
	We show all triangles in the diagram commute up to natural transformation.
	The bottom triangles commute exactly by the construction explained in the first part of the proof.
	Notice that $\iota_n' r_n  = r_{n,n+1}  \iota_n r_n$, where $\iota_n\colon \pi(X,X_n)\to \pi(X)$ is the inclusion.
	By Lemma~\ref{le:equivalence_pi(X)} we have that $ r_{n,n+1}  \iota_n r_n \simeq r_{n,n+1}$.
\end{proof}

\exa{Let $X=S^1\sqcup S^1$. Then, letting $X_0\subset X$ be a subset with precisely one point in each connected component, $\Grpd(\pi(X,X_0),G)=G\times G$ as discussed in Example~\ref{ex:ObS1cupS1}. Taking maps up to natural transformation corresponds to pairs of conjugacy classes of elements of $G$, so we have $\FG(X)=\C(G/G \times G/G)$. }
\exa{ \label{ex:S1cupS1_embed_conj}
	Consider again Example~\ref{ex:FGmer}, which in turn refers to Examples~\ref{ex:mer} and \ref{ex:mer_b}.
	The equivalence class $[g_1]$ in the basis of $\FG(Y)$
	consists of all maps sending $S^1$ to something in the conjugacy class of $g_1$. This allows us to refine the result of Example~\ref{ex:FGmer} as follows.
\ali{
	\lan[(g_1)] \; \vert\; \FG(M)\; \vert\; [(f_1,f_2)]\ran
	&=|G|^{-1}\vert \left\{c,d\in G \vert c^{-1}d^{-1}f_2d f_1c\sim g_1 \right\} \vert \\
	&= \;\; \vert \left\{d\in G \vert d^{-1}f_2d f_1\sim g_1 \right\} \vert. 
}	}

\exa{Let $X$ be the complement of the embedding of two circles shown, and explained in the caption of, Figure~\ref{fig:emb_S1cupS1_based}. Then, letting $X_0\subset X$ be the subset shown, we have $\Grpd(\pi(X,X_0),G)=G\times G$ as discussed in Example~\ref{ex:ObembS1cupS1}. 
Since all objects are mapped to the unique object in $G$, taking maps up to natural transformation is means taking maps up to conjugation by elements of $G$ at each basepoint, hence in this case maps are labelled by pairs of elements of $G$, up to simultaneous conjugation, so we have $\FG(X)=\C((G \times G)/G)$.  }

\exa{\label{ex:ZG_surfacetang}
	\begin{figure}
		\centering
		\def\svgwidth{0.5\columnwidth}
		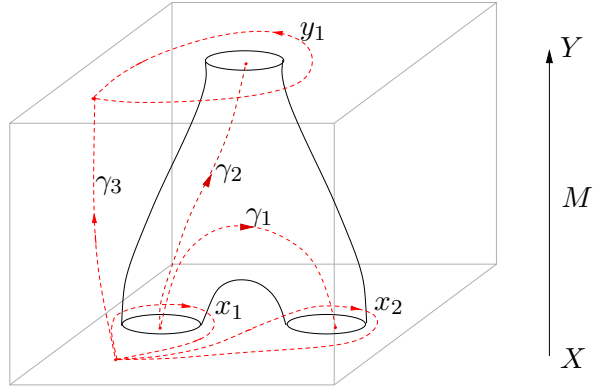
		\caption[Figure demonstrating that Example~\ref{ex:embmer} is a \chomcob{}]{This figure represents the \ccc{} from Example~\ref{ex:embmer}. The red points and lines show a possible choice of basepoints $M_0$ and paths. It can be seen that the equivalence  classes of the marked paths generate $\pi(M,M_0)$.
			We can see $X$ and $Y$ are \homfin{} by considering the intersection of the marked points and paths with $X$ and $Y$ respectively. Thus it is a \chomcob{}.} 
		\label{fig:embmer_c}
	\end{figure}
	Consider the \chomcob{} shown in Figure~\ref{fig:embmer}.
	This represents a manifold $M$, which is the complement of the marked subset in $\II^3$, and $X$ and $Y$ are given by the bottom and top boundary respectively. This becomes a cospan with the inclusion maps.
	This is in fact a \homcob{} from Examples~\ref{ex:embmer} and \ref{ex:embmer_b}.
	
	We calculate $\FG\big(\classche{\cmor{M}{i}{X}{Y}{j}}\big)$.
	We choose to use the \bhomcob{} shown in Figure~\ref{fig:embmer_c} for calculation. The set $M_0$ consists of all marked points and $X_0$ and $Y_0$ consist of the intersection of $M_0$ with $X$ and $Y$ respectively. 
	
	We have from Example~\ref{ex:S1cupS1_embed_conj} that basis elements in $\FG(X)$ are given by equivalence classes $[(f_1,f_2)]$ where $f_1,f_2\in G$ and $[]$ denotes simultaneous conjugation by the same element of $G$.
	
	Basis elements in $\FG(Y)$ are given by elements of $g$ taken up to conjugation, denoted $[g_1]$.
	
	Let $x\in X$ be the basepoint which is in the connected component of $X$ homotopy equivalent to the punctured disk, and $x'\in X$ some choice of basepoint in the other connected component. It follows from Lemma~\ref{le:add_bps} that there is a bijection sending a map $h\in \Grpd(\pi(M,M_0), G)$ to a quadruple $(h' , h(\gamma_1), h(\gamma_2), h(\gamma_3))\in\Grpd(\pi(M,\{x,x'\})\times G\times G\times G$,
	where $h'$ is the restriction of $h$ to $\pi(M,\{x,x'\})$. Now $\pi(M,\{x,x'\})$ is the disjoint union of the groupoids $\pi(M_1,\{x\})$ and $\pi(M_2,\{x'\})$ where $M_1$ is the path connected component of $M$ containing $x$, and $M_2$ is the path connected component containing $x'$.
	The group $\pi(M_2,\{x'\})$ is trivial, so there is one unique map into $G$.
	The group $\pi(M_1,\{x\})$ is equivalent to the twice punctured disk (see Example~\ref{ex:embmer_b}), which has fundamental group isomorphic to the free product $\Z*\Z$. This isomorphism can be realised by sending the loop $x_1$ to the $1$ in the first copy of $\Z$ and $x_2$ to the $1$ in the second copy of $\Z$.
	Thus we can label elements in $\Grpd(\pi(M_1,\{x\}), G)$ by elements of $G\times G$ where $g_1 \in(g_1,g_2)$ corresponds to the image of $x_1$, and $g_2$ the image of $x_2$. 
	Hence a map in $\Grpd(\pi(M,M_0), G)$ is determined by a five tuple $(a,b,c,d,e)\in G\times G\times G\times G\times G$ where $a$ corresponds to the image of $x_1$, $b$ to the image of $x_2$, and $c$, $d$ and $e$ correspond to the images of $\gamma_1$, $\gamma_2$ and $\gamma_3$ respectively.
	Hence we have
	\ali{
		\lan [g_1]\vert \FG(M)\vert [(f_1,f_2)]\ran &=
		|G|^{-2}\left\{a,b,c,d,e\in G \;\vert \; a=f_1, b=f_2, g_1\sim ebae^{-1}\right \}\\
		&=\left\{e\in G\; \vert \; g_1\sim ef_1f_2e^{-1} \right\}\\
		&=\begin{cases}
			|G| & \mbox{if $g_1\sim f_1f_2$}\\
			0 & \mbox{otherwise}.
		\end{cases}
	}
}

We now use Theorem~\ref{natiso_to_colim} to prove that identities are preserved.

\begin{lemma}\label{le:identity_preserved}
	The identity \homcob{} $\classche{\cmor{X\times \II}{\iota_0^X}{X}{\iota_1^X}{X}}$ for a space $X$ is mapped to the identity matrix by $\FG$.
\end{lemma}
\begin{proof}
	We will show that the matrix element 
	\[
	\left\langle \left[f\right] \middle| 
	\FG(X\times \II) \middle| \left[g\right] \right\rangle.
	\]
	is $1$ if 
	$\left[f\right]=\left[g\right]$ and $0$ otherwise.

	First note, there is an isomorphism
	\[
	\pi\left(X\times \II,(X_0\times\{0\})\cup (X_0\times \{1\})\right)
	\xrightarrow{\sim}\pi(X,X_0)\times \pi\left(\II,\{0,1\}\right).
	\] 
	given by sending a representative path to the pair containing the classes of each projection (see 6.4.4 in \cite{brownt+g}).
	Hence we have that $\left\langle f \; \middle|\; 
	\bFG(X\times \II) \;\middle|\; g \right\rangle$ is given by the cardinality of the set of maps 
	\[
	h\colon \pi(X,X_0)\times \pi\left(\II,\{0,1\}\right) \to G
	\]
	such that 
	\[
	h(s,e_{i})=
	\begin{cases}
	f(s), \indent  i=0,\\
	g(s), \indent  i=1
	\end{cases}
	\]	
	where $e_i$ denotes the constant path at the point $i$.
	Any pair in the product space can be written as a composition of pairs with only one non-identity component.
	The morphisms of $\pi\left(\II,\{0,1\}\right)$ are generated by the equivalence class of the path $\id_\II\colon \II\to \II$.
	Thus a map $h$ is completely defined by specifying its action on pairs of the form $(e_{x_j},\id_\II)$.
	Let $s\colon x_0\to x_1$ be a path in $X$ with $x_0,x_1\in X_0$.
	Notice that
	\[
	(e_{x_1},\id_\II^{-1})(s,e_{1})(e_{x_0},\id_\II)=(s,e_{0}) \indent
	\implies \indent h(e_{x_1},\id)^{-1}g(s)h(e_{x_0},\id_\II)=f(s).\]
	Hence such an $h$ exists if and only if the $h(e_{x_i},\id_\II)$ are a natural transformation from $f$ to $g$. 
	By Theorem~\ref{natiso_to_colim} this means the matrix element corresponding to $f$ and $g$ is zero unless $[f]=[g]$.
	
	Now we consider the matrix element
	\[
	\left\langle \left[f\right] \middle| 
	\FG(X\times \II) \middle| \left[f\right] \right\rangle.
	\]
	A map $h$ is a defined by a choice of $h(e_{x_0},\id)\in G$ for each $x_0\in X$, and all choices define a natural transformation. Using the definition of $\FG$ from Lemma~\ref{le:FG_sum}, we must sum over all $\left\langle f \; \middle|\; 
	\bFG(X\times \II) \;\middle|\; f' \right\rangle$ with $f'\sim f$ equivalent in the colimit, which, by Theorem~\ref{natiso_to_colim}, means there is a natural transformation $f$ to $f'$. Hence all choices of $h(e_{x_0},\id)=g\in G$ will contribute to the sum. There are $|G|^{|X_0|}$ choices, then with the normalisation, the matrix element is $1$.
\end{proof}

To complete this section, we give our final alternative way to define the map on objects.

\begin{theorem}\label{th:FG_pi(X)}
	Let $X$ be a \homfin{} space. Then
	\[
	\FG(X)\cong \C(\Grpd(\pi(X),G)/\cong),
	\]	
	where $\cong$ denotes the relation given by taking the set of groupoid maps up to natural transformation.
\end{theorem}
\begin{proof}
	We have from Theorem~\ref{natiso_to_colim} that $\FG(X)\cong \C(\Grpd(\pi(X, X_0),G)/\cong)$, for some finite representative set $X_0$.
	
	We have from Lemma~\ref{le:equivalence_pi(X)} that $\pi(X,X_0)$ and $\pi(X)$ are equivalent as categories. Let $\iota_0\colon \pi(X,X_0)\to X$ and $r_0\colon \pi(X)\to \pi(X,X_0)$ be as in the proof of Lemma~\ref{le:equivalence_pi(X)}.
	
	Let $[f]\in \Grpd(\pi(X),G)/\cong$, then there is a map 
	\ali{
		\phi\colon \Grpd(\pi(X),G)/\cong&\to \Grpd(\pi(X, X_0),G)/\cong\\ \phi([f])&\mapsto[f\circ \iota_0].
	}
	We show this map is well defined. Suppose $f'\in[f]$, so there is natural transformation, say $\eta$, from $f$ to $f'$. Then $f'\circ \iota_0\sim f\circ \iota_0$ using the restriction of  $\eta$ to $x\in X_0$. 
	
	This $\phi$ has inverse 
	$\phi'$, where $\phi'(g)=g\circ r_0$.
	Again this map is well defined, this time if $\eta$ is a natural transformation in $\Grpd(\pi(X, X_0),G)/\cong$, then  the maps $\eta_{r_0(x)}$ give the required natural transformation.
\end{proof}

\subsection{Monoidal functor \texorpdfstring{$\FG\colon \HomCob\to \VectC$}{ZG HomCob to VectC} }\label{sec:FG_mon}
We now show that the functor $\FG$ is symmetric monoidal when considering $\HomCob$ with the monoidal structure described in Section~\ref{sec:monHomCob} and $\VectC$ with the monoidal structure  in Proposition~\ref{pr:Vectkmon}.
We will need the following lemma.
\begin{lemma}\label{le:colim_prod}
	Let $X$ and $Y$ be \homfin{} spaces.
	There is a bijection 
	\[\kappa \colon \colim(\VV_{X\sqcup Y})	\xrightarrow{\sim} \colim(\VV_X)\times \colim(\VV_Y)\]
	where $\VV_X$ is as in Lemma~\ref{le:functor_V}
\end{lemma}
\begin{proof}
	For any subsets $X_\alpha \subseteq X$ and $Y_{\alpha'} \subseteq Y$, and points $x\in X_\alpha$ and $y\in Y_{\alpha'}$, we have that $\pi(X\sqcup Y,X_{\alpha}\sqcup Y_{\alpha'})(x,y)$ is empty. Thus there is an isomorphism of groupoids $\pi(X\sqcup Y,X_\alpha\sqcup Y_{\alpha'})\xrightarrow{\sim}\pi(X,X_\alpha)\sqcup \pi(Y,Y_{\alpha'})$ and we have a bijection
	$\Grpd(\pi(X\sqcup Y,X_\alpha\sqcup Y_{\alpha'}),G)\xrightarrow{\sim} \Grpd(\pi(X,X_{\alpha}),G)\times\Grpd(\pi(Y,Y_{\alpha'}),G)$ sending a map to the appropriate pair of restrictions. Equivalently we have a bijection $\VV_{X\sqcup Y}(X_\alpha\sqcup Y_{\alpha'})\xrightarrow{\sim}\VV_X(X_\alpha)\times \VV_Y(Y_{\alpha'})$.
	Thus $\colim(\VV_{X\sqcup Y})$
	is isomorphic to the colimit over the diagram with vertices of the form $\VV_X(X_\alpha)\times \VV_Y(Y_{\alpha'})$ and maps of the form $(\phi_{\alpha\beta}^X,\phi_{\alpha'\beta'}^Y)$, which we denote $\colim(\VV_{X\sqcup Y})'$. We construct a bijection between $\colim(\VV_{X\sqcup Y})'$ and $\colim(\VV_X)\times \colim (\VV_Y)$.
	
	Suppose $[(f,g)]=[(f',g')]$ in $\colim(\VV_{X\sqcup Y})'$ with $(f,g)\in \VV_X(X_\alpha)\times \VV_Y(Y_{\alpha'})$ and $(f',g')\in \VV_X(X_\beta)\times\VV_Y(Y_{\beta'})$.
	By the construction of the colimit, there exists a sequence of product sets $\VV_X(X_0)\times \VV_Y(Y_0),...,\VV_X(X_n)\times \VV_Y(Y_n)$ with $\VV_X(X_0)\times \VV_Y(Y_0)=\VV_X(X_\alpha)\times \VV_Y(Y_{\alpha'})$ and $\VV_X(X_n)\times \VV_Y(Y_n)=\VV_X(X_{\beta})\times \VV_Y(Y_{\beta'})$, and a sequence  of maps $\phi_0,...,\phi_{n-1}$ connecting $(f,g)$ and $(f',g')$ where each $\phi_i$ is either a map $ \VV_X(X_n)\times \VV_Y(Y_n)\to \VV_X(X_{n+1})\times \VV_Y(Y_{n+1})$ or a map $\VV_X(X_{n+1})\times \VV_Y(Y_{n+1})\to \VV_X(X_n)\times \VV_Y(Y_n)$.
	The projections of this sequence of maps give sequences of maps connecting $f$ and $f'$ in $\colim(\VV_X)$ and $g$ and $g'$ in $\colim (\VV_Y)$.
	Thus there is a well defined map 
	\ali{
	\kappa'\colon \colim(\VV_{X\sqcup Y})'	&\to \colim(\VV_X)\times \colim(\VV_Y)\\
	[(f,g)]&\mapsto ([f],[g]).
	}
	It is easy to see this map is a surjection. 
	To see that it is an injection, suppose now that $[f]=[f']$ in $\colim(\VV_X)$ and $[g]=[g']$ in $\colim(\VV_Y)$ then there are sequences $\phi^f_0,...,\phi^f_n$ and $\phi^g_0,...,\phi^g_n$ as in the proof of well definedness. Now the sequence given by $(\phi^f_0,\id),...,(\phi^f_n,\id),(\id,\phi^g_0),...(\id,\phi^g_1)$ is a sequence connecting $(f,g)$ and $(f',g')$ in $\colim(\VV_{X\sqcup Y})'$.
\end{proof}

\begin{lemma}
	The functor $\FG \colon \HomCob \to \VectC$
	 endowed with $(\FG)_0=1_\C\colon \C\to\C$ and natural transformations
	\ali{
	(\FG)_2(X,Y)\colon \FG(X) \otimes_{\C} \FG(Y)\to \FG(X\sqcup Y)
	}
	 which acts on basis elements as
	\ali{
	[f]\otimes_{\C}[g]&\mapsto \kappa^{-1}([f],[g])
	} 
	with $\kappa$ as in Lemma~\ref{le:colim_prod},
	is strong monoidal.\\
	(Here $\VectC$ has the monoidal structure from Lemma~\ref{pr:Vectkmon} and the monoidal structure on $\HomCob$ is as in Section~\ref{sec:monHomCob}.)
\end{lemma}
\begin{proof}
	Notice $\colim(\mathcal{V}_{\emptyset})=\emptyset$ as $\mathcal{V}_{\emptyset}$ has just one vertex, the empty set, and no maps. Hence $\FG(\emptyset)=\C$ so $(\FG)_0$ is well defined.
	
	The vector space $\FG(X) \otimes_{\C} \FG(Y)$ has a basis isomorphic to $\colim(\VV_X)\times \colim(\VV_Y)$.
	Thus the map $(\FG)_2(X,Y)$ is the linear extension of $\kappa^{-1}$, hence an isomorphism by Lemma~\ref{le:colim_prod}.
	
	The only complication in checking the associativity relation is understanding the image of the associator, $\FG(\alpha_{X,Y,Z})$. On basis elements $\FG(\alpha_{X,Y,Z})((f\otimes_{\C}g)\otimes_{\C}h)=f\otimes_{\C}(g\otimes_{\C} h)$. 
	Similarly we can check the unitality relations using that, on basis elements, we  have $\FG(\lambda_{X})(\emptyset\otimes_{\C} f)= f$ and $\FG(\rho_{X})(f\otimes_{\C} \emptyset)= f$.
	The proofs of each identity are similar to the proof of Lemma~\ref{le:identity_preserved}, that the identity is preserved, so we don't write out the details here.
\end{proof}

\begin{lemma}\label{le:FG_symm_mon}
	The monoidal functor $\FG \colon \HomCob \to \VectC$ is symmetric monoidal.
\end{lemma}
\begin{proof}
	As in the previous proof it is straightforward to check the relevant identity.
\end{proof}

\begin{lemma}
	The functor
	\[
	\tilde{\mathsf{Z}}_G=\FG\circ\Cob{n}\colon \mathbf{Cob}(n)\to \VectC
	\]
	where $\Cob{n}$ is as in Proposition~\ref{pr:fun_cob}, is a topological quantum field theory for all $n\in \N$
	, i.e. is a symmetric monoidal functor.
\end{lemma} 

\begin{proof}
	We have from Propositions~\ref{pr:fun_cob} and \ref{cob_to_HomCob} that $\Cob{n}$ is a symmetric monoidal functor into $\HomCob$ and from Theorem~\ref{th:FG_functor} and Lemma~\ref{le:FG_symm_mon} that $\FG$ is a symmetric monoidal functor $\HomCob\to \VectC$.
\end{proof}

\appendix

\section{Appendix}

\subsection{Colimits in \texorpdfstring{$\Grpd$}{Grpd}
}

Here we give sufficient results to prove Theorem~\ref{th:grpd_pushoutfg}, that the pushout of finitely generated groupoids is finitely generated.
We do this by explicitly constructing coequalisers in $\Grpd$. 
Our main reference for this is \cite{higgins}, although the parts on universal morphisms are also covered in \cite[Ch.8]{brownt+g}. We note that everything done here can also be done in $\Cat$, the category of small categories.

We proceed towards explicitly constructing coequalisers in $\Grpd$ by first introducing universal morphisms.

\medskip

Let $X$ be a set. Throughout this section we will also use $X$ to denote the trivial groupoid with the set $X$ as objects and only identity morphisms.
The meaning will be clear from context.

\begin{definition}
	Let $F \colon \GG \to \mathcal{H}$ be a functor and denote by $Ob(F)$ the unique functor making the following square, where the vertical maps are inclusions, commute.
	\[
	\begin{tikzcd}[ampersand replacement=\&]
		Ob(\GG) \ar[r,"Ob(F)"]\ar[d,"\iota_\GG"'] \& Ob(\mathcal{H}) \ar[d,"\iota_\mathcal{H}"] \\
		\GG \ar[r,"F"] \& \mathcal{H}  
	\end{tikzcd}
	\]
	Then $F$ is called {\em universal} if this square is a pushout.
\end{definition}

\lemm{\label{le:construct_universalgroupoid}
	Let $X$ be a set, $\GG$ a groupoid and $\sigma \colon Ob(\GG) \to X$ a function.
	There is a groupoid $U_\sigma(\GG)=(X,U_\sigma(\GG)(-,-), *_{U_\sigma(\GG)}, 1_{-}, (-)\mapsto (-)^{-1} )$ where:
	\begin{itemize}
		\item[(II)] For a pair $x,y \in X$ a word of length $n$ from $x$ to $y$ is a sequence 
		\[
		a=a_n...a_1
		\]
		of morphisms $a_i\colon g_i\to g_i'$ in $\GG$ such that
		\begin{enumerate}[noitemsep, label={},topsep=0pt,label=(\roman*)]
			\item for all $ \, i=\{1,\ldots,n-1\}$, $g_i'\neq g_{i+1}$,
			\item for all $i=\{1,\ldots,n-1\}$, $\sigma(g_i') = \sigma(g_{i+1})$,
			\item $\sigma(g_1)=x$ and $\sigma(g_n')=y$,
			\item for all $i\in\{1,\dots,n\}$,  $a_i\neq 1_{g_i}$.
		\end{enumerate}
		For a pair $x,y\in X$ the set $U_\sigma(\GG)(x,y)$ is the set all words from $x$ to $y$ when $x\neq y$ and in the case $x=y$ we also add the empty word which we will denote $1_x$.
		(Notice that $U_\sigma(\GG)(x,y)$ may well be empty, and certainly will be in the case $x\neq y$ and $x$ and $y$ are not in the image of $\sigma$.)
		\item[(III)] The composition $*_{U_\sigma(\GG)}$ is given by concatenating words and, where possible, evaluating compositions in $\GG$ cancelling identities.
		\item[(IV)] For $x\in X$, the identity morphism is the empty word $1_x$.
		\item[(V)] Suppose $a=a_n\dots a_1$ is a word in $U_\sigma(\GG)(x,y)$, then it has inverse $a^{-1}=a_1^{-1}...a_n^{-1}$, where $a_i^{-1}$ is the inverse in $\GG$. Notice that this is in $U_\sigma(\GG)(y,x)$.
	\end{itemize}
}
\begin{proof}
	($\CC 1$) It is immediate from the construction that the empty word acts as an identity under concatenation.\\
	($\CC 2$) Evaluating compositions is associative because concatenation is associative and the composition in $\GG$ is associative. \\
	($\GG 1$) It is immediate from the construction that the described word is an inverse.
\end{proof}

\begin{lemma}\label{le:construct_universalfunctor}
	Let $\GG$ be a groupoid, $X$ a set, $\sigma \colon Ob(\GG) \to X$ a function and $U_\sigma(\GG)$ as constructed in Lemma~\ref{le:construct_universalgroupoid}.
	There is a functor $\sigma' \colon \GG \to U_\sigma(\GG)$, defined as follows. On objects $\sigma'=\sigma$.
	For a morphism $a\colon g\to g'$ in $\GG$ we have $\sigma'(a)=1_{\sigma(g)}$ if $a=1_g$ and $\sigma'(a)=a$, considered as a length one word in $U_\sigma(\GG)$, otherwise. Note that $a$ is a word from $g$ to $g'$.
\end{lemma}
\begin{proof}
	First note that the all identities in $\GG$ are mapped to identities in $ U_\sigma(\GG)$ by construction.
	
	Suppose $a\colon g\to g'$ and $a'\colon g'\to g''$ are morphisms in $\GG$. 
	If $a=1_g$ then $\sigma'(1_g*_\GG a')=\sigma'(a')=a'$ which is the concatenation of $a'$ with the empty word.
	If $a'=1_g$, then similarly $\sigma(a*_\GG a')$ is precisely the concatenation $\sigma'(a)\sigma'(a')$.
	If neither $a$, nor $a'$ is an identity, then
	\[
	\sigma'(a*_\GG a')=a*_\GG a'
	\]
	which is precisely the concatenation $\sigma'(a)\sigma'(a')$ with all possible compositions in $\GG$ evaluated.
\end{proof}

\begin{lemma}\label{le:universal_map}
	Let $\GG$ be a groupoid, $X$ a set and $\sigma \colon Ob(\GG) \to X$ a function.
	The functor $\sigma' \colon \GG \to U_\sigma(\GG)$ as constructed in Lemma~\ref{le:construct_universalfunctor} is universal.	
\end{lemma}
\begin{proof}
	To prove $\sigma'$ is universal we construct, for any groupoid $\mathcal{K}$ and functors $\tau$ and $\phi$ with $\tau \circ\sigma=\phi\circ\iota_\GG$, a unique map $\phi^*$ making the following diagram commute.
	\[\begin{tikzcd}[ampersand replacement=\&]
		Ob(\GG) \ar[r,"\sigma"] \ar[d,"\iota_\GG"'] \& X \ar[d,"\iota_X"] \ar[ddr,"\tau",bend left] \& \\
		\GG \ar[r,"\sigma'"]\ar[drr,"\phi"',bend right] \& U_\sigma(\GG) \ar[rd,"\phi^*",dashed]\& \\
		\& \& \mathcal{K} \\	
	\end{tikzcd}\]
	We must have that, on objects $x\in Ob(U_\sigma(\GG))=X$, $\phi^*(x)=\tau(x)$, and that $\phi^*(1_x)=1_\tau{x}$.
	Now let $a_1$ be a word of length $1$ in $U_\sigma(\GG)$, then $a_1$ is a morphism in $\GG$ and, by commutativity, we must have  $\phi^*(a_1)=\phi(a_1)$.
	For words of length $n$, $a=a_n\ldots a_1$ in $U_\sigma(\GG)$, by functoriality we must have
	\begin{align*}
		\phi^*(a)&= \phi^*(a_n)*_\mathcal{K}\ldots *_\mathcal{K}\phi^*(a_1) \\
		&=\phi(a_n)*_\mathcal{K}\ldots *_\mathcal{K}\phi(a_1).
	\end{align*}
	Notice
	for any $a_{i}\colon g_i\to g_i'$ and $a_{i+1}\colon g_{i+1}\to g_{i+1}'$ we have $\sigma(g_i')=\sigma(g_{i+1})$, hence $\tau (\sigma(g_i'))=\tau(\sigma(g_{i+1}))$ and so, by commutativity of the diagram, $\phi(g_i')=\phi(g_{i+1})$, so we have that $\phi^*$ is well defined.
	By construction composition is preserved on word concatenations. Composition is preserved also by evaluating compositions and removing identities because $\phi$ preserves composition.
	We have that $\phi^*$ is unique by construction.
\end{proof}

We now construct coequalisers in $\Grpd$.

\medskip

There is a forgetful functor $\Grpdforget\colon \Grpd \to \Set$ which sends a groupoid $\GG$ to the set $Ob(\GG)$ and a functor to the corresponding set map determined by its action on objects.
 The functor $\Grpdforget\colon \Grpd \to \Set$ 
has right adjoint $\Delta \colon \Set \to \Grpd$ constructed as follows.
Let $X$ be a set, there is a groupoid which has object set $X$ and exactly one morphism $(x,y)$ from $x$ to $y$ for any pair $x,y\in X$, with
		the composition defined by $(x,y)(y,z)=(x,z)$. The identity morphisms are $(x,x)$ and $(x,y)^{-1}=(y,x)$.
		This is called the {\em indiscrete groupoid} and denoted $\Delta(X)$.
 The functor $\Delta\colon \Set \to \Grpd$ sends a set $X$ to $\Delta(X)$,
		and sends a function $f\colon X\to Y$ to the unique functor from $\Delta(f)\colon \Delta(X)\to \Delta(Y)$ which acts as $f$ on objects.

Since left adjoints preserve colimits (Theorem~\ref{th:lapcl}), and the forgetful functor $\Grpdforget\colon \Grpd \to \Set$ is a left adjoint, we can find the object set of a coequaliser of a diagram $D$ in $\Grpd$ by evaluating the
set coequaliser 
of $\Grpdforget\circ D$ in $\Set$.

\medskip

It is straightforward to prove the following proposition.

\begin{proposition}
	Let $f,g\colon X\to Y$ be functions in  $\Set$. Then  there exists a coequaliser
	\[
	\begin{tikzcd}[ampersand replacement = \&]
		X \ar[r,"f",shift left]\ar[r," g"',shift right] \& Y \ar[r,"p"] \& Y/\sim,
	\end{tikzcd}
	\]
	where $\sim$ is the reflexive, symmetric and transitive closure of the relation $$\{(f(x),g(x))\; \vert\; x\in N\} $$ on $Y$, and $p$ is the canonical map sending $y\in Y$ to its equivalence class $[y]\in Y/\sim$. \qed
\end{proposition}

\lemm{\label{le:universalcoeq}
	Let $f,g\colon \GG_0\to\GG_1$ be functors of groupoids and let $f',g'\colon Ob(\GG_0)\to Ob(\GG_1)$ denote the images under the forgetful functor, $\Grpdforget(f)$ and $\Grpdforget(g)$ respectively.
	Let $\sigma \colon Ob(\GG_1) \to Ob(\GG_1)/\sim_{Ob} $ be the coequaliser of $f'$ and $g'$ in $\Set$, and let $\sigma'\colon \GG_1\to U_\sigma(\GG_1)$ denote the universal map constructed from $\sigma$ as in Lemma~\ref{le:construct_universalfunctor}.
	
	For each pair $x,y\in Ob(\GG_1)/\sim_{Ob} $ let $R_{x,y}$ be the relation on $U_\sigma(\GG_1)(x,y)$ with
	\[
	(a_n\ldots a_1, a_n'\ldots a_1')\in R_{x,y}
	\]
	if there exists a morphism $b\in \GG_0$ such that for some $i\in\{1,\ldots,n\}$,  $\sigma' f(b)=a_i$ and $\sigma' g(b) =a_i'$ and for all other $j\neq i$, $a_j=a_j'$.
	(I)
	The collection of equivalence relations $\bar{R}=(U_\sigma(\GG_1)(x,y),\bar{R}_{x,y})$ is a congruence, hence there is a quotient groupoid $U_\sigma(\GG_1)/\bar{R}$.
	
	(II) The following diagram is a coequaliser
	\begin{align} 
		\begin{tikzcd}[ampersand replacement = \&]
			\GG_0 \ar[r,"\sigma' f",shift left]\ar[r,"\sigma' g"',shift right] \& U_\sigma(\GG_1)  \ar[r,"\gamma^*"] \& \nicefrac{U_\sigma(\GG_1)}{\bar{R}},
		\end{tikzcd}
		\label{eq:universalcoeq}
	\end{align}
	where $\gamma^*$ is the quotient functor induced by $\bar{R}$. 
}
\begin{proof}
	(I) This is straightforward to check.\\
	(II) It is immediate from the construction that \ref{eq:universalcoeq} commutes.
	Suppose we have a groupoid $\HH$ and a map $\psi\colon U_\sigma(\GG_1)\to \HH$ with $\psi\circ \sigma'\circ f = \psi\circ \sigma'\circ g$.	
	Let $a_1$ and $a_1'$ be words of length $1$ in $U_\sigma(\GG_1)$ and suppose there exists $b\in \GG_0$ such that $\sigma' \circ f(b)=a_1$ and $\sigma' \circ g(b) =a_1'$.
	Then, by assumption, we must have $\psi (a_1)=\psi(a_1')$.

	Now suppose there exists words $
	a= a_n\ldots a_1$ and $a'=a_n'\ldots a_1'$ in $U_\sigma(\GG_1)$
	and a morphism $b\in \GG_0$ such that for some $i\in\{1,\ldots,n\}$,  $\sigma' \circ f(b)=a_i$ and $\sigma'\circ  g(b) =a_i'$, and for all other $j\neq i$, $a_j=a_j'$. Then by functoriality we must have $\psi(a)=\psi(a')$.

	Thus we arrive at precisely the relation described in the Lemma. So $\psi$ must factor through $U_\sigma(\GG_1)/\bar{R}$, and the map $U_\sigma(\GG_1)/\bar{R}\to \mathcal{H}$ is uniquely determined by considering the action of $\psi$ on preimages of elements in $U_\sigma(\GG_1)/\bar{R}$. Hence \ref{eq:universalcoeq}
	is a coequaliser.
\end{proof}

\lemm{\label{le:grpd_coeq}
	Let $f,g\colon \GG_0\to\GG_1$ be functors of groupoids.
	Then
	\begin{align}
		\begin{tikzcd}[ampersand replacement = \&]
			\GG_0 \ar[r,"f",shift left]\ar[r,"g"',shift right] \& \GG_1 \ar[r,"\gamma^*\sigma'"] \& \nicefrac{U_\sigma(\GG_1)}{\bar{R}}
		\end{tikzcd}
		\label{eq:grpdcoeq}
	\end{align}
	is a coequaliser, where we use the notation of the previous Lemma.
}
\begin{proof}
	We have from the Lemma~\ref{le:universalcoeq}
	that \eqref{eq:universalcoeq}
	is a coequaliser. Suppose we have a functor $\psi\colon \GG_1\to \HH $ with $\psi \circ f=\psi \circ g$.
	Using the universal property of the coequaliser in $\Set$, this implies the existence of a unique functor $\psi_{Ob}\colon Ob(\GG_1)/\sim_{Ob} \to \mathcal{H}$ with $\psi_{Ob} \circ \sigma = \psi\circ \iota_\GG $.
	Then, since $\sigma'$ is universal, there exists a unique map $\psi'\colon U_\sigma(\GG_1)\to \HH$ with $\psi=\psi'\circ \sigma'$. Hence using the universal property of the coequaliser there exists a unique map $\Psi\colon \nicefrac{U_\sigma(\GG_1)}{\bar{R}}\to \HH$ making \eqref{eq:universalcoeq} commute.
	Then $\Psi\circ \gamma^* \circ \sigma'=\psi$ and so $\Psi$ makes \eqref{eq:grpdcoeq} commute.
	
	Note that this is unique, since any $\Psi'\colon \nicefrac{U_\sigma(\GG_1)}{\bar{R}}\to \HH$
	making \eqref{eq:grpdcoeq} commute will also commute with $\psi'$ and \eqref{eq:universalcoeq}, and by the universal property of the coequaliser, this map is unique.
\end{proof}

We have now shown that we can construct a coequaliser in $\Grpd$ of any pair of maps.
The following result gives a way to construct a pushout in terms of a coequaliser. 

\lemm{\label{le:pushout_coeq}
	Let $f_1\colon C_0\to C_1$ and $f_2\colon C_0\to C_2$ be morphisms in a category $\CC$. Then
	\[
	\begin{tikzcd}
		C_0 \ar[r,"f_1"]\ar[d,"f_2"'] \& C_1 \ar[d,"p_1"]  \\
		C_2 \ar[r,"p_2"'] \& C
	\end{tikzcd}
	\]
	is a pushout of $f_1$ and $f_2$ if and only if 
	\[
	\begin{tikzcd}[column sep=large]
		C_0 \ar[r,"{i_1*_\CC f_1}",shift left]\ar[r,"{i_2*_\CC f_2}"',shift right] \& C_1\amalg C_2 \ar[r,"p"] \& C
	\end{tikzcd}
	\]
	is a coequaliser,
	where $C_1 \xrightarrow{i_1} C_1\amalg C_2 \xleftarrow{i_2}C_2$ is a coproduct, and $p$ is the map obtained from applying the universal property of the coproduct to the maps $C_1\xrightarrow{p_1} C \xleftarrow{p_2} C_2$.
	.}
\begin{proof}
	Let $h\colon C_1\amalg C_2\to H$ be a morphism with $h*_\CC i_1*_\CC f_1=h*_\CC i_2*_\CC f_2 $.
	Then, using the universal property of the coproduct, $h$ uniquely determines a pair of maps $h_1\colon C_1 \to H$ and $h_2\colon C_2\to H$ with $h*_\CC i_1=h_1$ and $h*_\CC i_2=h_2$, and hence $h_1*_\CC f_1=h_2*_\CC f_2 $.
	Then there exists a unique map $h'\colon C\to H$, with 
	$p_1*_\CC h'= h_1$ and $p_2*_\CC h'= h_2$ if and only if $h'$ is also the unique map satisfying $p*_\CC h'=h$.
\end{proof}

It is, in fact, possible to obtain all colimits in terms of
coproducts and coequalisers, although we won't need that level of generality here.
Thus, having constructed the coequaliser, we now know that $\Grpd$ has all colimits.

\setlength\emergencystretch{2em} 
\normalem 
\printbibliography

\end{document}